\documentclass[11pt]{article}

\usepackage[margin=1in]{geometry}
\usepackage{setspace}
\usepackage{float}
\usepackage{array, makecell} %
\usepackage{xcolor,colortbl}
\usepackage{framed,color}
\usepackage{xcolor}
\usepackage{bbm}

\usepackage{amsthm,amsmath,amssymb}
\usepackage[bookmarks=true,hypertexnames=false,pagebackref]{hyperref}
\hypersetup{colorlinks=true, citecolor=blue, linkcolor=red,
  urlcolor=blue}
\usepackage{newpxtext}
\usepackage[libertine,vvarbb]{newtxmath}

\usepackage[T1]{fontenc} 
\usepackage{textcomp} 

\usepackage[scr=rsfso]{mathalfa}
\usepackage{nicefrac}
\usepackage{cleveref}
\usepackage{graphicx}
\usepackage{aliascnt}
\usepackage[section]{placeins}
\usepackage{tcolorbox}
\usepackage{epigraph}
\usepackage{xifthen}
\usepackage{latexsym}
\usepackage{framed}
\usepackage{fullpage}
\usepackage{bm}
\usepackage{braket}
\usepackage{enumitem}
\usepackage{thmtools,thm-restate}

\usepackage{algorithm}
\usepackage{algorithmic}

\newtheorem{theorem}{Theorem}[section]

\newtheorem{fact}[theorem]{Fact}
\newtheorem{corollary}[theorem]{Corollary}
\newtheorem{lemma}[theorem]{Lemma}
\newtheorem{definition}[theorem]{Definition}
\newtheorem{conjecture}[theorem]{Conjecture}
\newtheorem{problem}[theorem]{Problem}
\newtheorem{example}[theorem]{Example}

\newtheorem{openproblem}{Open Problem}
\newtheorem{remark}[theorem]{Remark}
\newtheorem*{remark*}{Remark}

\usepackage{mleftright}

\newcommand{\As}{\mathcal{A}}
\newcommand{\Bs}{\mathcal{B}}

\newcommand{\Es}{\mathcal{E}}

\newcommand{\A}{\mathsf{A}}
\newcommand{\B}{\mathsf{B}}

\renewcommand{\kappa}{\ell}

\newcommand{\poly}{{\sf poly}}

\DeclareMathOperator{\Tr}{Tr}

\newcommand{\sdotfill}{\textcolor[rgb]{0.8,0.8,0.8}{\dotfill}} 

\newcommand{\ID}{\mathrm{ID}}

\newcommand{\V}{\mathrm{vec}}

\newcommand{\tom}[1]{{\color{teal} Tom: #1}}

\title{Impossibility of One-Way One-Round Quantum 4-Coloring via Matrix-Space Stability}
\author{
Tom Gur\thanks{University of Cambridge. Email: \texttt{tom.gur@cl.cam.ac.uk}. Supported by ERC Starting Grant 101163189 and UKRI Future Leaders Fellowship MR/X023583/1.}\and
Longcheng Li\thanks{University of Cambridge. Email: \texttt{lilongcheng116@gmail.com}. Supported by ERC Starting Grant 101163189.}
}

\date{}
\begin{document}

\maketitle

\begin{abstract}
    We show that one-way one-round quantum-LOCAL algorithms cannot $4$-color directed cycles with high probability. This is the first lower bound in the high-probability quantum LOCAL setting that goes beyond the non-signaling and bounded-dependence models, exploiting the structure of distributed quantum algorithms.
    
    Our proof establishes a bidirectional connection between distributed quantum computing and extremal combinatorics. We obtain our lower bound by proving a Mantel-type stability theorem for weighted matrix spaces.
\end{abstract}

\newpage

\tableofcontents
\newpage

\section{Introduction}

The LOCAL model \cite{linial1992locality,peleg2000distributed} abstracts a
basic question in distributed computing: what can be computed in a network
where each node sees only its local neighborhood? In this model, each
node has a unique identifier and communicates with its
neighbors in synchronous rounds. In each round, a node may exchange a message
with each of its neighbors, and perform
local computation. Both the local computation and the message size are
\emph{unbounded}, and the complexity of a LOCAL algorithm is measured solely by the
number of rounds of communication, referred to as its \emph{locality}.

Over the last decades, a wide range of distributed graph problems have been
studied extensively in the LOCAL model \cite{suomela2013survey}. A central focus
has been on problems governed by local constraints, such as graph coloring,
maximal independent set, and maximal matching in bounded-degree graphs. Rather
than studying these problems case by case, substantial recent effort has been
devoted to understanding their complexity as a whole, revealing a 
clean landscape of locality \cite{chang2020complexity,balliu2025localproblems}.

It is natural to ask whether quantum information changes this locality landscape. In
the quantum-LOCAL model, each node is a quantum computer that is connected to its neighbors via a quantum channel \cite{gavoille2009can}, and the nodes do not share randomness or entanglement. Analogously to the classical model, each node performs arbitrary local quantum operations and exchanges a quantum state with
its neighbors in each round.

A recent line of work \cite{legall2019quantum,balliu2025localproblems,balliu2025distributed} has demonstrated
quantum advantage for several specially constructed problems
using quantum nonlocal games.
At the same time, other results have ruled out quantum
advantage for many natural problems, including 
$2$-coloring even cycles \cite{gavoille2009can,fraigniaud2026distributed},
approximate graph coloring \cite{coiteuxroy2024approximate},
distributed linear programming \cite{balliu2025dequantizing}, and more \cite{dhar2024local,brandt2026post}.

\paragraph{Coloring directed cycles.}
Despite the extensive study discussed above, the quantum locality of the basic symmetry-breaking problem of coloring
a directed $n$-cycle with a small number of colors remains poorly understood. In
this problem, each node outputs one of $q$ colors and the goal is to obtain a
proper $q$-coloring of the cycle. This problem plays an important role in the
LOCAL model, as fast cycle coloring algorithms can be used as a primitive for
many other symmetry-breaking problems, including maximal independent set,
maximal matching, edge coloring, and vertex coloring in bounded-degree graphs
\cite{chang2019exponential}.


For every fixed $q\geq 3$, the locality of $q$-coloring directed cycles is known
to be $\Theta(\log^* n)$ in both deterministic and randomized classical LOCAL
models \cite{cole1986deterministic,linial1992locality,naor1991lower}. In the
quantum-LOCAL model, however, the only nontrivial lower bounds for $3$-coloring
directed cycles rule out \emph{one-way one-round} algorithms
\cite{holroyd2017finitary}, where each node sends one quantum message
to its successor and receives one quantum message from its predecessor in the
directed cycle before computing its color. For more than three colors, even this
restricted case remains open.

This limited understanding reflects a barrier to proving lower bounds in the quantum-LOCAL model \cite{akbari2025online}. Namely, all known quantum-LOCAL lower bounds
\cite{gavoille2009can,arfaoui2014without,balliu2025dequantizing,fraigniaud2026distributed}
establish impossibility in stronger causality-based models, such
as the non-signaling and bounded-dependence models. The latter abstracts away the round-based structure of the algorithm and
retains only the requirement that
the outputs of sufficiently distant nodes should be 
independent, as their light-cones do not intersect. 
Since these dependence constraints also hold for quantum-LOCAL algorithms,
lower bounds in the bounded-dependence model imply quantum-LOCAL ones.

In particular, proper $4$-colorings of cycles with minimal dependence are known 
to exist in the bounded-dependence model \cite{holroyd2018finitely}, and 
hence such arguments \emph{cannot} rule out the possibility of quantum algorithms. Motivated by this, we study the following open problem: 
\begin{center}
    \emph{
    Can a one-way one-round quantum-LOCAL algorithm  
    $4$-color directed cycles with
    high probability?
    }
\end{center}
By the discussion above, resolving this problem requires new techniques that specifically exploit the
structure of quantum algorithms, rather than bounded dependence alone.

\subsection{Our Results}

Our main contribution is a negative answer to the question above.
More precisely, we prove the following theorem.

\begin{theorem} \label{thm:q4-informal}
    There exists a universal constant $C>0$ such that for any
    directed $n$-cycle and any
    one-way one-round anonymous quantum algorithm that 
    assigns each node $v_i$ a color $c_i\in[4]$,
    we have
    \[
    \Pr[c_i=c_{i+1}]\geq C \quad \forall i\in \mathbb{Z}_n.
    \]
\end{theorem}

Here, an \emph{anonymous} quantum algorithm is simply a quantum-LOCAL algorithm in
which all nodes are initially identical without unique identifiers. In the
high-probability setting, anonymity is without loss of generality, as each
node can independently sample an identifier from a sufficiently large set
$[n^{C_0}]$, which yields pairwise distinct identifiers with probability at
least $1-1/n^{C_0-2}$.

Moreover, the output of any one-way one-round anonymous quantum algorithm is
\emph{$1$-dependent}: the colors of two vertex sets $U$ and $V$ are independent
whenever the graph distance between $U$ and $V$ exceeds $1$.
Hence the collision
events on a collection of edges are jointly independent if any two edges are
separated by at least one vertex. By choosing $\lfloor n/3\rfloor$ such edges
$\{(3i, 3i+1):0\leq i < \lfloor n/3\rfloor\}$ and applying
\Cref{thm:q4-informal}, we find that the probability of producing a proper
coloring is at most $(1-C)^{\lfloor n/3\rfloor}$, which is exponentially small in $n$. We
thus obtain the following corollary.

\begin{corollary}
    No one-way one-round quantum-LOCAL algorithm can $4$-color directed cycles
    with high probability.
\end{corollary}

Combined with the construction of $1$-dependent proper $4$-colorings of cycles
\cite{holroyd2018finitely}, our impossibility result gives the first natural separation
between the bounded-dependence model and one-way one-round quantum-LOCAL algorithms.

\subsection{Matrix-Space Stability}

We prove \Cref{thm:q4-informal} by expressing one-way one-round
quantum-LOCAL $q$-coloring into an extremal-combinatorial stability problem in
matrix spaces. 

The key notion underlying this connection is the \emph{multiplicative
energy} of a matrix space, which captures the total collision probability in the noncommutative quantum setting. Loosely speaking, we associate each color $a\in[q]$ with a matrix space
$W_a\subseteq\mathbb C^{N\times N}$, where $N$ is the local dimension.
These spaces together form an orthogonal decomposition of
$\mathbb C^{N\times N}$. We encode the quantum message state via a PSD weight matrix
$\Lambda\in\mathbb C^{N\times N}$, and
define the $\Lambda$-energy of a matrix space as follows.

\begin{definition}[$\Lambda$-energy of a matrix space]
\label{def:energy-projective}
Let $W\subseteq\mathbb C^{N\times N}$ be a matrix space and let
$\Lambda\in\mathbb C^{N\times N}$ be a PSD matrix. For any
Frobenius-orthonormal basis $\{M_k\}_k$ of $W$, define the $\Lambda$-energy of
$W$ by
\[
    \mathcal E_\Lambda(W)
    :=\sum_{k,\ell}\|\Lambda M_k\Lambda M_\ell\Lambda\|_F^2.
\]
\end{definition}

Note that this definition is independent of the chosen orthonormal basis. The term
\emph{energy} reflects its form as a sum of squared norms, measuring the total
strength of weighted products between basis elements of $W$.
In particular, the energy is zero exactly when the weighted product of any two
matrices in $W$ is zero. For more details about this notion of energy and the intuition behind it, see \Cref{sec:classical-overview}.


In the following theorem, we show that the collision probability in 
one-way one-round quantum $q$-coloring is indeed captured by the total energy of matrix-space
decompositions.

\begin{theorem}[informal, see \Cref{thm:reduction}] \label{thm:reduction-informal}
    For every fixed integer $q\geq 2$, the following are equivalent:
    \begin{enumerate}
        \item[(i)]
        For every normalized PSD $\Lambda$ and every orthogonal decomposition of $\mathbb{C}^{N\times N}$ into $q$ subspaces $\bigoplus_{a\in[q]} W_a$,
        their total $\Lambda$-energy  
        $\sum_{a\in[q]}\mathcal E_\Lambda(W_a)=\Omega_q(1)$.
        \item[(ii)] Every one-way one-round anonymous quantum algorithm that outputs 
        a $q$-coloring $(c_i\in[q])_{i\in\mathbb{Z}_n}$ of a directed cycle has collision probability 
        $\Pr[c_i=c_{i+1}]=\Omega_q(1)$. 
        Consequently, one-way one-round $q$-coloring is impossible in the quantum-LOCAL model.
    \end{enumerate}
\end{theorem}

We remark that \Cref{thm:reduction-informal} provides the first lower-bound technique that operates directly
within the quantum-LOCAL model and is not captured by bounded dependence.
Using this technique, our four-color impossibility  
in \Cref{thm:q4-informal} reduces to proving a
dimension-independent lower bound on the total energy of any orthogonal
decomposition of $\mathbb C^{N\times N}$ into four subspaces.

\begin{theorem}[informal, see \Cref{thm:q4}] \label{thm:q4-energy-informal}
    There exists a universal constant $C>0$ such that,
    for every dimension $N$, every normalized PSD
    matrix $\Lambda\in\mathbb C^{N\times N}$ and every orthogonal decomposition
    \[
        \mathbb C^{N\times N}=\bigoplus_{a=1}^4W_a,
    \]
    we have
    \[
    \sum_{a=1}^4\mathcal E_\Lambda(W_a) \geq C.
    \]
\end{theorem}

To prove this bound, we first ask how large a matrix space can be if its energy
is zero. Under this weighted setting, the appropriate measure of size is its \emph{$\Lambda$-mass},
\[
    m_\Lambda(W):=\sum_k\|\Lambda M_k\Lambda\|_F^2,
\]
where $\{M_k\}_k$ is a Frobenius-orthonormal basis of $W$. For $\Lambda=I$,
this mass is simply the dimension of $W$. For any normalized $\Lambda$, it is a weighted
notion of density, since the masses of the spaces in any orthogonal decomposition
of $\mathbb C^{N\times N}$ sum to $1$. Zero-energy spaces have mass at most
$1/4$, and we call a space \emph{maximal} if its mass
reaches this extremal threshold.

A key ingredient in our proof is a \emph{stability theorem} for these maximal
spaces. In extremal combinatorics, a stability theorem captures the phenomenon that an object
close to an extremal bound must also be \emph{close in structure} to an object
attaining that bound. We show that a maximal matrix space with nearly
zero energy is close to a maximal space with exactly zero energy.

\begin{theorem}[informal, see \Cref{thm:weighted-stability}] 
If $W$ is maximal and has $\mathcal E_\Lambda(W)=o(1)$, then there
is a nearby pair $(\Lambda',W')$ of a PSD weight matrix and a
maximal matrix space, such that $\mathcal E_{\Lambda'}(W')=0$.
\end{theorem}

For four colors, averaging guarantees a color space of mass at least $1/4$. Then
we can apply the stability theorem to this heaviest space, which leads to a
universal lower bound on the total $\Lambda$-energy; see
\Cref{sec:quantum-overview} for details.

\paragraph{Subsequent work.}
Since the connection established in \Cref{thm:reduction-informal} is bidirectional, it
also provides a route to quantum upper bounds. Indeed, a follow-up work
\cite{flin2026quantumadvantagedistributedsymmetry} built on this connection to
construct an $O(1)$-round quantum algorithm that $3$-colors directed cycles.
Since the classical LOCAL model requires $\Omega(\log^* n)$
rounds for this task, this establishes the first asymptotic quantum advantage in
the LOCAL model for a natural symmetry-breaking problem. 
This overturns a common belief that no natural graph problem admits such an
advantage in the LOCAL model. Before this, quantum advantages had been shown
only for artificial problems and most research on the quantum LOCAL model had
focused on lower bounds. Our framework, though developed for lower bounds, laid
the foundation for this result.

We note a methodological distinction between the two works: the main result of the follow-up
work was discovered by GPT-6 while attempting to extend our four-color
lower bound, whereas the results of the present work were obtained without
AI assistance. To the best of our knowledge, existing LLM tools had not
resolved this problem before our energy-based characterization was introduced.
In this sense, our techniques provided the key human-derived
insight that enabled the subsequent AI-assisted breakthrough.

\subsection{Related Work}

The classical locality of coloring directed cycles with any fixed number of
colors $q\geq 3$ is well understood. Cole and Vishkin gave a deterministic
$O(\log^* n)$-round algorithm, while Linial proved a matching lower bound
\cite{cole1986deterministic,linial1992locality}. Naor later extended the same lower
bound to randomized algorithms that succeed with high probability
\cite{naor1991lower}, so randomization gives no asymptotic improvement for this problem.
More generally, $(\Delta+1)$-coloring graphs of maximum
degree $\Delta$ can be solved deterministically in
$O(\Delta+\log^* n)$ rounds, which is $O(\log^* n)$ when $\Delta$ is constant
\cite{barenboim2014distributed}.

Several recent works have sought analogous quantum lower bounds. \cite{holroyd2017finitary,gavoille2019localisation,gall2022non} show that a one-way one-round quantum
algorithm cannot $3$-color directed cycles.
\cite{fraigniaud2026distributed} proved an $\Omega(\log^* n)$ lower bound for
$3$-coloring rooted trees, and ruled out even constant quantum advantage for
$2$-coloring even cycles. All of these lower bounds also apply to the stronger
bounded-dependence model. Recently, \cite{coiteuxroy2026cycles} developed a
technique for the quantum model and proved that any anonymous quantum algorithm using
finitely many qubits that $3$-colors cycles with probability $1$
requires $\Omega(n)$ rounds. 
Although anonymous nodes can generate unique random
identifiers with high probability, this reduction is unavailable under the
probability-$1$ restriction. Their result therefore does not imply a lower bound
in the quantum-LOCAL model and is complementary to our high-probability
separation.

Another line of work gives constructions of finitely dependent proper colorings.
A distribution over vertex labels is \emph{$k$-dependent} if the labels on any
two vertex sets at graph distance greater than $k$ are independent.
\cite{holroyd2016finitely} first constructed a stationary $2$-dependent
$3$-coloring and a stationary $1$-dependent $4$-coloring of $\mathbb Z$, and
subsequently \cite{holroyd2018finitely} extended these constructions to finite
cycles. Since any one-way one-round anonymous quantum algorithm on a directed
cycle produces a stationary $1$-dependent distribution, these constructions
imply that bounded-dependence arguments alone cannot establish the quantum lower
bound desired in this work. More generally,
\cite{akbari2025online} showed that
$(\Delta+1)$-coloring a graph with constant maximum degree $\Delta$ 
can be solved in the bounded-dependence model with constant locality, 
imposing a strong barrier to proving quantum lower bounds 
for this problem.

\subsection{Open Problems}
We highlight two open problems that we find particularly interesting.

The first is to determine the exact number of colors required
for one-way one-round quantum algorithms to properly color directed cycles. 
Our result rules out $4$ colors, while the follow-up work
\cite{flin2026quantumadvantagedistributedsymmetry} shows that $2\cdot 10^4$
colors suffice, leaving a substantial gap between 
$5$ and $2\cdot 10^4$.

\begin{openproblem}
Determine the smallest integer $q$ for which one-way one-round quantum-LOCAL
algorithms can properly $q$-color directed cycles with high probability.
\end{openproblem}

Second, to the best of our knowledge, our result gives the first separation between the
bounded-dependence and quantum-LOCAL models: proper $1$-dependent $4$-colorings
of cycles exist, but one-way one-round quantum-LOCAL algorithms cannot produce a
proper $4$-coloring with high probability. This is a constant gap in locality, and it remains open whether there exists a distributed graph problem
that achieves an asymptotic separation between the two models.

\begin{openproblem}
Does there exist a distributed graph problem whose locality 
in the quantum-LOCAL model is asymptotically larger than in the bounded-dependence model?
Furthermore, determine how large this gap can be.
\end{openproblem}

\paragraph{Organization.}
\Cref{sec:technical-overview} gives an overview of the proof.
\Cref{sec:prelim} introduces the necessary background and notation. 
\Cref{sec:energy} develops the formal definition of mass and energy, and establishes their basic properties.
Next, \Cref{sec:zero-energy} characterizes zero-energy operators and proves the
core stability theorem for low-energy operators. Finally,
\Cref{sec:one-way} connects the energy formulation to one-way one-round quantum
coloring and proves our main impossibility theorem for four colors.

\section{Technical Overview}\label{sec:technical-overview}

Our proof of \Cref{thm:q4-informal} recasts one-way one-round coloring as an extremal-combinatorics problem. We begin in~\Cref{sec:classical-overview} by reproving the classical lower bound
from this perspective. Then, in~\Cref{sec:quantum-overview}, we give an overview of our proof in the quantum setting, which builds on a noncommutative generalization of the classical extremal problem.





\subsection{Recasting the Classical Lower Bound} \label{sec:classical-overview}
We first consider one-way one-round anonymous randomized algorithms
that color a directed cycle with $q$ colors. 
Let $N$ denote the number of possible local random seeds.
Recall that every node $v_i$ independently samples a uniform seed $r_i\in[N]$,
sends it to $v_{i+1}$, and outputs $c_i=f(r_{i-1},r_i)$
for a deterministic function $f:[N]^2\to[q]$. This is the standard block-factor
view of constant-radius stochastic processes
\cite{naor1991lower,holroyd2017finitary}.
By symmetry, every edge has the same local collision probability. It therefore suffices
to analyze \(\Pr[c_i=c_{i+1}]\), the probability that the endpoints of a fixed edge output
the same color.

For each color $a\in[q]$, define the directed edge set
\(
    E_a:=\{(x,y)\in[N]^2:f(x,y)=a\},
\)
and the corresponding directed graph \(G_a=([N],E_a)\).
The sets $E_1,\ldots,E_q$ partition $[N]^2$. Adjacent nodes
$v_i$ and $v_{i+1}$ both output color $a$ exactly when
$(r_{i-1},r_i)$ and $(r_i,r_{i+1})$ are both edges of $G_a$. Consequently,
\begin{equation}\label{eq:overview-classical-collision}
    \Pr[c_i=c_{i+1}]
    =\frac{1}{N^3}\sum_{a\in[q]}
    \Es(G_a),
\end{equation}
where
\[
    \Es(G_a):=
      \left|\{(x,y,z)
      \in[N]^3:(x,y),(y,z)\in E_a\}\right|.
\]
For a directed graph $G=([N],E)$, $\Es(G)$ counts \emph{directed $2$-walks},
i.e., sequences of two consecutive directed edges, with repeated vertices allowed.
We also call $\Es(G)$ the \emph{energy} of $G$, because 
it coincides with the unweighted energy of the corresponding matrix space 
$W_G=\operatorname{span}\{E_{xy}:(x,y)\in E\}$, as we explain in
\Cref{sec:quantum-overview}.
The classical lower bound is therefore equivalent
to the following graph-theoretic problem.
\begin{problem} \label{prob:classical}
Does every $q$-coloring of $[N]^2$ contain 
$\Omega_q(N^3)$ monochromatic directed $2$-walks?    
\end{problem}

A general block-factor argument gives a positive answer for every fixed $q$
\cite[Proposition~5]{holroyd2017finitary}, but this argument relies on steps such as
conditioning on the intermediate random seed without disturbing the system, which does not survive quantization
due to the no-cloning principle \cite{balliu2025localproblems}. Instead, we give an alternative proof
for $q\leq 4$ via an extremal graph argument, and its natural noncommutative
counterpart leads to the quantum lower bound.

\paragraph{The extremal graph picture.} The first step is to analyze the structure of graphs without $2$-walks.

\begin{lemma} \label{lem:directed-mantel}
    Given any directed graph $G=([N], E)$, if $G$ contains no directed $2$-walks, then
    \[
        |E|\leq \frac{N^2}{4}.
    \]
    Equality holds if and only if there is a partition $[N]=L\sqcup R$ with
    $|L|=|R|=\frac{N}{2}$ such that $E=L\times R$.
\end{lemma}

Indeed, if there are no $2$-walks, then a vertex cannot have both an incoming and an outgoing edge. If $L$ is
the set of vertices with positive out-degree and $R=[N]\setminus L$, then
every edge points from $L$ to $R$. Hence
$E \subseteq L\times R$, and 
\[
    |E|\leq |L||R|
    = |L|\left(N-|L|\right)
    \leq \frac{N^2}{4}.
\]
Equality forces $G$ to be a balanced complete directed cut. \Cref{lem:directed-mantel} can be viewed as a directed
analogue of \emph{Mantel's theorem},
which states that an undirected graph without triangles has at most $N^2/4$ edges,
with equality attained precisely by the
balanced complete bipartite graph \cite{mantel1907problem}. 

In \Cref{prob:classical}, for $q\leq 3$ colors, averaging yields a color class with at least $N^2/3$
directed edges, substantially exceeding the extremal threshold $N^2/4$ by $\Omega(N^2)$ edges.
This suggests the possibility of showing what is known in extremal combinatorics as \emph{supersaturation} (i.e., the phenomenon where exceeding an extremal threshold not only implies one copy of a substructure, but an abundance of it): 
a graph $G$ with substantially more than $N^2/4$ edges must contain many $2$-walks, yielding the desired energy bound for the $q\leq 3$ case.

For $q=4$,
however, averaging guarantees only $N^2/4$ edges, so we will also need a \emph{stability}
statement: if 
$G$ has nearly $N^2/4$ edges and few $2$-walks, then $G$ must be
close to a balanced complete directed cut. 
Applying stability to the largest color class would impose structural
constraints on the four color classes, forcing many monochromatic
$2$-walks, and in turn, the desired energy lower bound. The following lemma provides both the
supersaturation and stability bounds.

\begin{lemma}[Robust version of \Cref{lem:directed-mantel}] \label{lem:directed-stability}
    Every directed graph $G=([N],E)$ admits a partition $[N]=L\sqcup R$ such
    that
    \[
        |E\setminus(L\times R)|\leq 2\sqrt{N\cdot \Es(G)}.
    \]
    We have two immediate consequences:
    \begin{enumerate}
        \item (Supersaturation) If $|E|-N^2/4=\Omega(N^2)$, then $\Es(G)=\Omega(N^3)$.
        \item (Stability) If $|E|=(1/4+o(1))N^2$ and $\Es(G)=o(N^3)$, then the corresponding $L$ and $R$ satisfy
        \[
        |L|,|R|=(1/2+o(1))N\quad\text{and}\quad
        |E\mathbin\triangle(L\times R)|=o(N^2).
        \]
    \end{enumerate}
\end{lemma}

To prove the lemma, we separate vertices into sources and sinks according to an out-degree threshold.
If $\Es(G)=0$, \Cref{lem:directed-mantel} already places $E$ inside a 
directed cut.
Otherwise, choose the partition $L\sqcup R$
\[
    L:=\{y:d^+(y)>\theta\}
    \quad\text{and}\quad
    R:=\{y:d^+(y)\leq\theta\},
\]
where $\theta:=\sqrt{\Es(G)/N}$.
For each $y\in R$, delete all its outgoing edges, and for each $y\in L$,
delete all its incoming edges. 
This is equivalent to deleting all edges outside $L\times R$,
leaving the remaining edge set $E'\subseteq L\times R$.

Then we observe that directed $2$-walks can be counted by
enumerating the intermediate vertex $y$ of the walk, and thus
\begin{equation}\label{eq:2-walk}
    \Es(G)=\sum_{y\in[N]}d^-(y)d^+(y),
\end{equation}
where \(d^-(y)\) and \(d^+(y)\) are the in-degree and out-degree of vertex \(y\).

Using this, we bound the number of deleted edges.\footnote{This threshold rule is not optimal. If, at
each vertex, we instead delete the smaller of its incoming and outgoing edge
sets, then the number of deletions is at most
$\sum_y\min\{d^-(y),d^+(y)\}\leq\sqrt{N\cdot\Es(G)}$.  We use the threshold rule 
because it extends to the general quantum setting.}
A directed edge is deleted if its tail lies in \(R\) or its head lies in \(L\).
Bounding these two contributions separately gives
\[
\begin{aligned}
    |E\setminus (L\times R)|
    \leq \sum_{d^+(y)\leq\theta}d^+(y)
       +\sum_{d^+(y)>\theta}d^-(y) 
    &\leq N\theta
       +\frac{\sum_{d^+(y)>\theta}d^-(y)d^+(y) }{\theta} \\
    &\leq N\theta+\frac{\Es(G)}{\theta}
     =2\sqrt{N\cdot\Es(G)}.
\end{aligned}
\]

Since $|L\times R|\leq N^2/4$, any excess of $|E|$ above $N^2/4$ must be
deleted, giving supersaturation. For stability, under the assumption
$\Es(G)=o(N^3)$, only $o(N^2)$ edges are deleted, leaving nearly $N^2/4$ edges. Then the remaining edges $E'$ must almost saturate a directed cut $L\times R$
with nearly balanced sides; otherwise, $L\times R$ could not accommodate so
many edges. Together with $o(N^2)$ deleted edges, this gives
$|E\mathbin\triangle(L\times R)|=o(N^2)$.

\paragraph{Classical impossibility for $q\leq 4$ colors.} 

For $q\leq 3$, \Cref{prob:classical} follows immediately
from \Cref{lem:directed-stability}: By averaging, there must exist a color class, say $G_1$, with
at least $N^2/3$ edges, and thus supersaturation implies $\Es(G_1)=\Omega(N^3)$.

Now consider $q=4$. Suppose, toward a contradiction,
that $\sum_{a\in[q]} \Es(G_a)=o(N^3)$. By averaging, some color class, say $G_1$, has at least $N^2/4$
edges. Since $\Es(G_1)=o(N^3)$, supersaturation forces
$G_1$ to have only $(1/4+o(1))N^2$ edges. Then stability gives a nearly balanced partition 
$[N]=L\sqcup R$ such that 
$|E_1\triangle(L\times R)|=o(N^2)$.
Changing the colors of these $o(N^2)$ exceptional edges changes the number of
monochromatic directed $2$-walks by only $o(N^3)$.
Therefore we can round $E_1$ to 
exactly $L\times R$. Observe that after rounding, the subgraph $L\times L$ contains no edges of color $1$,
and thus is colored by only the other three
colors. Since $|L|=(1/2+o(1))N$, the three-color lower bound inside $L$
produces $\Omega(|L|^3)=\Omega(N^3)$ monochromatic directed $2$-walks, 
contradicting the original assumption. 
Therefore, for any $4$-coloring of $[N]^2$, we have
$\sum_{a\in[q]}\Es(G_a)=\Omega(N^3)$.

\subsection{The Quantum Lower Bound} \label{sec:quantum-overview}

We now turn to the quantum setting. Both
\Cref{lem:directed-mantel,lem:directed-stability} have natural noncommutative
analogues. As a warm-up, we first explain the simple special case where the quantum messages are restricted to the maximally entangled state, which can be thought of as the uniform case (see below).
Then, we discuss the case of general quantum messages, which presents several challenges unique to the quantum setting.

\subsubsection{The Uniform Case}
We first consider one-way one-round anonymous quantum algorithms using a special quantum message. Each node $v_i$ starts by preparing a maximally entangled state
\[
    \ket{\Phi}=\frac{1}{\sqrt N}\sum_{j\in[N]} \ket{j}_{\A_i}\otimes \ket{j}_{\B_i},
\]
which is the quantum analogue of a uniform random seed. It 
then sends $\B_i$ to $v_{i+1}$
and receives $\B_{i-1}$ from $v_{i-1}$; 
this is the quantum counterpart of retaining one copy of its random seed and
sending the other to the successor.
Finally, it applies a $q$-outcome projective measurement 
$\{P_a\}_{a\in[q]}$ on the joint registers $\B_{i-1}\otimes\A_i$ 
and outputs the measurement outcome $c_i\in[q]$ as its color.

With the above in mind, we proceed to analyze the collision probability. Since the measurement projectors
$\{P_a\}_{a\in[q]}$ act on the space $\mathbb C^N\otimes\mathbb C^N$, their images form an orthogonal decomposition of $\mathbb C^N\otimes\mathbb
C^N$. To connect to the classical picture, we 
identify each projector $P_a$ with a matrix space $W_a\subseteq\mathbb C^{N\times N}$
under the vectorization
map\footnote{ Vectorization maps matrices to bipartite vectors by
$\V(\ket{i}\bra{j}) = \ket{i}\otimes \ket{j}$. It is a bijection and preserves the
Hilbert--Schmidt inner product; see \Cref{sec:norm}. } $\V:\mathbb{C}^{N\times N}\to\mathbb C^N\otimes\mathbb C^N$.
Specifically, we write the spectral decomposition as
\[
    P_a=\sum_k\V(M_{a,k})\V(M_{a,k})^\dagger,
\]
where $\{M_{a,k}\}_k$ is Frobenius-orthonormal in $\mathbb{C}^{N\times N}$,
and define the matrix space $W_a:=\operatorname{span}\{M_{a,k}\}_k \subseteq \mathbb{C}^{N\times N}$. 
The measurement therefore induces
the orthogonal decomposition
\[
    \mathbb C^{N\times N}=W_1\oplus\cdots\oplus W_q.
\]
A direct calculation gives a similar formula for the local collision probability:
\begin{equation}\label{eq:overview-unweighted-energy}
    \Pr[c_i=c_{i+1}]
    =\frac{1}{N^3}\sum_{a\in[q]}\mathcal E(W_a),
\end{equation}
where
\[
\mathcal E(W_a):=\sum_{k,\ell}\|M_{a,k}M_{a,\ell}\|_F^2.
\]

The quantity $\mathcal E(W_a)$ is invariant under the choice of orthonormal basis.
It is referred to as the \emph{unweighted energy} of $W_a$, 
which can be obtained from \Cref{def:energy-projective} by setting the weight $\Lambda=I$.
This definition is precisely the noncommutative analogue of 
the energy of a directed graph. To see this,
for a directed graph $G=([N],E)$, define a matrix
space $W_G=\operatorname{span}\{E_{xy}:(x,y)\in E\}$,
where $E_{xy}$ are standard matrix units. Observe that 
$E_{xy}E_{wz}$ is nonzero if and only if $y=w$, 
namely $(x,y,z)$ forms a directed $2$-walk.
Thus $\Es(G)=\mathcal E(W_G)$, and we recover the classical identity \Cref{eq:overview-classical-collision}.

By \Cref{eq:overview-unweighted-energy}, the desired lower bound for local collision probability in this unweighted quantum setting is equivalent
to the following matrix-theoretic problem.
\begin{problem} \label{conj:unweighted-energy}
    Does every orthogonal decomposition, 
    \(
        \mathbb C^{N\times N}=\bigoplus_{a\in[q]}W_a,
    \)
    satisfy
    \[
        \sum_{a\in[q]}\mathcal E(W_a)\geq C_q N^3,
    \]
    where $C_q>0$ is a positive constant depending only on $q$?
\end{problem}

Indeed, after dividing by $N^3$, the conclusion is exactly a
dimension-independent lower bound $C_q>0$ on the local collision probability. We next 
develop the noncommutative analogues of \Cref{lem:directed-mantel,lem:directed-stability}, and then
sketch the solution of \Cref{conj:unweighted-energy} for $q\leq 4$.

\paragraph{Noncommutative analogue of the extremal graph picture.}
For a matrix space $W$ and its orthonormal basis $\{M_k\}$, 
$W$ has zero energy if and
only if $M_kM_\ell=0$ for every $k,\ell$, or equivalently
\[
    W^2:=\operatorname{span}\{XY:X,Y\in W\}=\{0\}.
\]

The following lemma characterizes the extremal structure of zero-energy matrix
spaces, generalizing \Cref{lem:directed-mantel} to the noncommutative setting.
The extremal threshold remains $N^2/4$, with the dimension of a zero-energy
matrix space replacing the size of a graph with no $2$-walks. The vertex
partition $[N]=L\sqcup R$ is replaced by an orthogonal decomposition $\mathbb
C^N=U\oplus U^\perp$.

\begin{lemma}[Noncommutative analogue of \Cref{lem:directed-mantel}]
    \label{lem:unweighted-directed-mantel}
    If $W\subseteq\mathbb C^{N\times N}$ satisfies $\mathcal E(W)=0$, then
    \[
        \dim(W)\leq \frac{N^2}{4}.
    \]
    Equality holds if and only if there is an orthogonal decomposition
    $\mathbb C^N=U\oplus U^\perp$ such that 
    \[
    \dim(U)=\dim(U^\perp)=\frac{N}2
    \qquad\text{and}\qquad
    W=\mathrm L(U^\perp,U),
    \]
    where $\mathrm L(U^\perp,U)$ is the space of linear operators from $U^\perp$ to $U$.
\end{lemma}

We construct the decomposition by setting $U:=\sum_{X\in W}\operatorname{Im}(X)$. Then every $X\in W$ maps into $U$,
and $W^2=\{0\}$ implies that every $X\in W$ annihilates $U$. Therefore
$W\subseteq\mathrm L(U^\perp,U)$, and
\[
    \dim(W)\leq \dim(U)(N-\dim(U))
    \leq \frac{N^2}{4}.
\]
We call $\mathrm{L}(U^\perp,U)$ a \emph{cut space}, 
the noncommutative analogue of a complete directed cut. 
Equality forces $U\oplus U^\perp$ to be balanced and $W=\mathrm L(U^\perp,U)$. 
If we define a unitary matrix $S\in\mathrm{U}(N)$
whose columns form an orthonormal basis of $U$ and $U^\perp$, 
then $W$ is exactly the space supported on the
upper-right block 
\[
W = \left\{
S
\left(\begin{array}{c|c}
            0 & X \\ \hline
            0 & 0
\end{array}\right)
S^\dagger
:
X\in\mathbb C^{\frac{N}2\times\frac{N}2}
\right\},
\]
which generalizes the balanced complete directed cut to the noncommutative setting.

As in the classical argument, we will need a statement of the following robust form.
\begin{lemma}[Noncommutative analogue of \Cref{lem:directed-stability}]
    \label{lem:unweighted-directed-stability}
    Every $W\subseteq\mathbb C^{N\times N}$ admits an orthogonal
    decomposition $\mathbb C^N=U\oplus U^\perp$ such that
    \[
    \mathrm{Tr}((I-P_L)P_W) \leq 2 \sqrt{N\cdot \mathcal E(W)},
    \]
    where $P_W$ and $P_L$ are the orthogonal projectors onto $W$ and
    $\mathrm L(U^\perp,U)$ respectively. Consequently, 
    \begin{enumerate}
        \item (Supersaturation) If $\dim(W)-N^2/4=\Omega(N^2)$, then $\mathcal E(W)=\Omega(N^3)$.
        \item (Stability) If $\dim(W)=(1/4+o(1))N^2$ and $\mathcal E(W)=o(N^3)$, then the corresponding $U$ satisfies
        \[
        \dim(U),\dim(U^\perp)=(1/2+o(1))N\qquad\text{and}\qquad
        \|P_W-P_L\|_1=o(N^2).
        \]
    \end{enumerate}
\end{lemma}

The quantity $\Tr((I-P_L)P_W)$ 
measures how much of $W$ falls outside $\mathrm L(U^\perp,U)$, which is exactly the noncommutative version of the
number of edges outside a directed cut,
and $\|P_W-P_L\|_1$ measures 
the difference of $W$ from $\mathrm L(U^\perp,U)$.
Indeed, for a directed
graph $G=([N],E)$ and a vertex partition $[N]=L\sqcup R$,
if we set $W=\operatorname{span}\{E_{xy}:(x,y)\in E\}$ and
$U=\operatorname{span}\{e_x:x\in L\}$, then
$\Tr((I-P_L)P_W)=|E\setminus(L\times R)|$
and 
$\|P_W-P_L\|_1=|E\mathbin\triangle(L\times R)|$,
which recovers \Cref{lem:directed-stability}.

The proof of \Cref{lem:unweighted-directed-stability} uses the same source--sink
separation as in the classical case, applied to the spectrum of an out-degree
matrix. Let $\{M_k\}$ be an orthonormal basis of $W$, 
and define a basis-independent PSD matrix $B$ with its spectral decomposition as
\[
    B:=\sum_k M_kM_k^\dagger
      =\sum_{i=1}^N\beta_i u_i u_i^\dagger,
\]
where $N\geq\beta_1\geq\cdots\geq\beta_N\geq0$. 

The matrix $B$ is the spectral analogue of the out-degree sequence. 
For the coordinate space defined by a directed graph, $B$ is diagonal with entries 
being precisely the
out-degrees.  In general, its eigenvectors 
$\{u_i\}$ can be viewed as ``vertices'', 
and the corresponding eigenvalue $\beta_i$ is the ``out-degree'' of
$u_i$.  If $\mathcal E(W)=0$, \Cref{lem:unweighted-directed-mantel} already places $W$ inside
a cut space. Otherwise, we choose
the orthogonal decomposition $U\oplus U^\perp$ by 
partitioning the eigenvectors as
\[
    U:=\operatorname{span}\{u_i:\beta_i>\theta\}
    \quad\text{and}\quad
    U^\perp:=\operatorname{span}\{u_i:\beta_i\leq\theta\},
\]
where $\theta:=\sqrt{\mathcal E(W)/N}$, the same 
threshold as in the classical case.

To define the analogue of in-degree, for $i,j\in[N]$, set 
$y_{ij}:=\sum_k|u_i^\dagger M_ku_j|^2$ and 
its column sum $\alpha_j:=\sum_i y_{ij}$. 
Then $y_{ij}$ measures the ``flow'' from $u_i$ to $u_j$, and 
$\alpha_j$ is the ``in-degree'' of $u_j$.  
Naturally, the $i$-th row sum of $y_{ij}$ equals $\beta_i$. 
One can verify that the matrix-space energy satisfies the noncommutative
counterpart of \Cref{eq:2-walk}:
\begin{equation} \label{eq:energy-noncomm}
    \mathcal E(W)=\sum_{j\in[N]}\beta_j\alpha_j.
\end{equation}
All the flows from $U$ to $U^\perp$,
i.e., from high eigenspace to low eigenspace, are contained in the 
cut space $\mathrm L(U^\perp,U)$, 
\footnote{We note that
the orientation conventions of directed graphs and matrix multiplication are reversed:
a directed edge $(x,y)$ goes from $x$ to $y$, while matrix $E_{xy}$ maps $e_y$ to $e_x$.
}
so the mass of $W$ outside $\mathrm L(U^\perp,U)$ is
\[
    \Tr((I-P_L)P_W)
    =\sum_{\beta_i\leq\theta\ \mathrm{or}\ \beta_j>\theta}y_{ij}.
\]
Observe that the summation can be bounded by 
the sum of rows with $\beta_i\leq\theta$
plus the sum of columns with $\beta_j>\theta$,
which correspond to deleted out-degrees and deleted in-degrees respectively
in the classical picture. Thus
\[
\begin{aligned}
    \Tr((I-P_L)P_W)
    \leq \sum_{\beta_i\leq\theta}\beta_i
       +\sum_{\beta_j>\theta}\alpha_j 
    \leq N\theta
       +\frac{1}{\theta}\sum_{\beta_j>\theta}\beta_j\alpha_j 
    \leq N\theta+\frac{\mathcal E(W)}{\theta}
     =2\sqrt{N\cdot\mathcal E(W)},
\end{aligned}
\]
which further implies the supersaturation and stability 
by the same classical reasoning.

\paragraph{Quantum impossibility for $q\leq 4$ in the unweighted setting.}
The quantum proof now follows the outline of the classical one.  For
$q\leq3$ colors, averaging gives a space $W_a$ of dimension at least $N^2/3$,
noticeably larger than the extremal threshold $N^2/4$, so
supersaturation immediately gives $\mathcal E(W_a)=\Omega(N^3)$.

For $q=4$, suppose toward a contradiction that the total energy is $o(N^3)$.  By averaging, one
of the four spaces, say $W_1$, has dimension at least $N^2/4$. Since its energy is
also $o(N^3)$, supersaturation forces
$\dim(W_1)=(1/4+o(1))N^2$.  Stability then gives an orthogonal
decomposition $\mathbb C^N=U\oplus U^\perp$ such that
\(
    \dim(U),\dim(U^\perp)=(1/2+o(1))N
\) and \(
    \|P_{W_1}-P_L\|_1=o(N^2)
\), 
where $P_{W_1}$ and $P_L$ are the orthogonal projectors onto $W_1$ and
$\mathrm{L}(U^\perp,U)$ respectively.

We next round $W_1$ to the exact zero-energy space.
Replace $W_1$ by $\mathrm{L}(U^\perp,U)$ and adjust the remaining three spaces so that the
four still form a complete orthogonal decomposition. The adjustment can be done within
a trace distance of $o(N^2)$ and changes the total energy by only $o(N^3)$.
After rounding, the first color is supported on $\mathrm{L}(U^\perp,U)$ and hence vanishes on the
diagonal block $\mathrm{L}(U,U)$.
Restricting the other three spaces to $\mathrm{L}(U,U)$ gives a decomposition of
the smaller space $\mathrm{L}(U,U)$ with only three colors.\footnote{
Strictly speaking, the spaces $\{W_a\}_{2\leq a\leq 4}$ after restriction may no longer be orthogonal, 
but their projectors $\{P_{W_a}\}_{2\leq a\leq 4}$ after restriction still form a complete POVM on $\mathrm{L}(U,U)$.
Therefore, in the main proof, we extend the energy definition and all arguments to 
PSD operators
on $\mathbb{C}^{N\times N}$
, not just orthogonal projectors; see \Cref{sec:energy}.
}
Each basis matrix $M_k$ from one of the remaining three color spaces lies in
$\mathrm{L}(U^\perp,U)^\perp$. Relative to the decomposition $U\oplus U^\perp$, it therefore has
the block form
\[
    M_k = \left(\begin{array}{c|c}
            X_k & 0 \\ \hline
            Y_k & Z_k
    \end{array}\right),
\]
where $X_k\in\mathrm{L}(U,U),Y_k\in\mathrm{L}(U,U^\perp),Z_k\in\mathrm{L}(U^\perp,U^\perp)$.
Since multiplication is preserved in the upper-left block $\mathrm{L}(U,U)$:
\[
    M_kM_\ell = \left(\begin{array}{c|c}
            X_kX_\ell & 0 \\ \hline
            * & *
    \end{array}\right),
\]
we have $\|M_kM_\ell\|_F^2\geq \|X_kX_\ell\|_F^2$. Summing this
inequality and using the energy formula \eqref{eq:overview-unweighted-energy} shows 
that restricting to $\mathrm{L}(U,U)$ does not increase the energy.\footnote{
    One may notice that $\{X_k\}$ may not be orthogonal, 
    but the energy formula still works; see \Cref{fact:energy}.
}
Since $\dim(U)=(1/2+o(1))N$, applying the $3$-color case gives
$\Omega(\dim(U)^3)=\Omega(N^3)$ total energy
in $\mathrm{L}(U,U)$, and thus in the original four-color decomposition, which contradicts the assumption that
the total energy is $o(N^3)$.

\subsubsection{The General Case}
In the classical model, there is essentially no loss in starting with a
uniform random seed: unbounded local computation can 
approximate any desired finite distribution with 
sufficiently many uniform bits. However, the quantum setting is
different. The bipartite state shared across an edge after communication may
have nonuniform Schmidt coefficients, so it cannot be assumed to be maximally
entangled.\footnote{A maximally entangled state can be converted into any
bipartite pure state by local operations assisted by further classical communication
\cite{nielsen2010quantum}, but that additional communication is unavailable in
our one-way one-round model.} 
Since we cannot assume the weight matrix $\Lambda=I$ in the general case, we will have to deal with non-commutative $\Lambda$, which leads to several challenges we discuss below.

Since the Schmidt bases can be absorbed into the local measurement, it suffices
to consider quantum states of the form
\[
\ket{\Psi} = \sum_{j=1}^N \lambda_j \ket{j}_{\A}\otimes \ket{j}_{\B},
\]
where $\lambda_j\geq 0$ for $j\in[N]$ and $\sum_{j=1}^N \lambda_j^2 = 1$.

The general one-way one-round quantum algorithm can be 
described as follows.
Each node $v_i$ initially prepares $\ket{\Psi}_{\A_i\B_i}$, then sends register $\B_i$ to its successor, and
finally performs the same $q$-outcome projective measurement
$\{P_a\}_{a\in[q]}$ on the joint registers $\B_{i-1}\otimes\A_i$ and outputs the measurement outcome $c_i\in[q]$ as its color.
\footnote{
    Projective measurements suffice in this setting:
    Since we allow arbitrary quantum messages, 
    we can use Naimark's dilation theorem to 
    realize any POVM as a projective measurement \cite{nielsen2010quantum}.
} 

As before, we write $W_a$ for the matrix space associated
with $P_a$, and $\{M_{a,k}\}_k$ for an orthonormal basis of
$W_a$. Let $\Lambda:=\operatorname{diag}(\lambda_1,\ldots,\lambda_N)$
be the weight matrix encoding the Schmidt coefficients of the quantum message,
which is a PSD matrix normalized by $\|\Lambda\|_F^2=1$. 
Then the local collision probability becomes
\[
    \Pr[c_i=c_{i+1}]
    = \sum_{a\in[q]}\mathcal E_\Lambda(W_a),
\]
where 
\[
    \mathcal E_\Lambda(W_a)
    =\sum_{k,\ell}
      \|\Lambda M_{a,k}\Lambda M_{a,\ell}\Lambda\|_F^2
\]
is precisely the $\Lambda$-energy of $W_a$ defined in \Cref{def:energy-projective}. 
If we set $\Lambda=I/\sqrt{N}$, then we recover the unweighted energy 
in \eqref{eq:overview-unweighted-energy} with a normalizing factor $1/N^3$.
Consequently, the impossibility with general quantum messages is equivalent to
lower bounding the \emph{total weighted energy}.

\begin{problem}
Does every normalized $\Lambda$ and orthogonal decomposition,
$\mathbb{C}^{N\times N}=\bigoplus_{a\in[q]}W_a$, satisfy
\[
\sum_{a\in[q]} \mathcal E_\Lambda(W_a)\geq C_q,
\]
where $C_q>0$ is a positive constant depending only on $q$?
\end{problem}

The extremal structure in \Cref{lem:unweighted-directed-mantel}
persists in this weighted setting, with $\Lambda$-mass
replacing the dimension. For a matrix space $W$ with orthonormal
basis $\{M_k\}$, define its $\Lambda$-mass as
\[
    m_\Lambda(W)=\sum_k\|\Lambda M_k\Lambda\|_F^2.
\]

For any normalized $\Lambda$, the $\Lambda$-mass can be viewed as a weighted notion of density,
since we have $\sum_{a\in[q]} m_\Lambda(W_a)=1$. 
Operationally, the mass $m_\Lambda(W_a)$ equals the marginal probability of a single node
$v_i$ outputting the color $a \in[q]$ in the above quantum algorithm.
We show the following theorem that can be viewed as the weighted 
version of the extremal structure in \Cref{lem:unweighted-directed-mantel}. 
Specifically, we will need the extremal structure to be a cut space $\mathrm{L}(U^\perp,U)$, where
$U$ and $U^\perp$ carry equal weight and are both invariant under $\Lambda$. We discuss these properties after stating the theorem below.

\begin{theorem}[informal, see \Cref{thm:ze-mm}] \label{thm:ze-mm-informal}
If $\mathcal E_\Lambda(W)=0$, then
$m_\Lambda(W)\leq1/4$. 
If equality holds,
after restricting to the support of $\Lambda$, there exists an orthogonal decomposition
$\mathbb C^N=U\oplus U^\perp$ satisfying
\begin{enumerate}[
    label=(P\arabic*),
    ref=(P\arabic*),
    leftmargin=1.5cm,
]
    \item\label{W1}
    $W=\mathrm L(U^\perp,U)$;
    \hfill(cut space)

    \item\label{W2}
    $\Tr(\Lambda^2\Pi_U)
      =\Tr(\Lambda^2\Pi_{U^\perp})
      =\frac12$;
    \hfill(balanced weight)

    \item\label{W3}
    $\Lambda U=U$ and $\Lambda U^\perp=U^\perp$;
    \hfill($\Lambda$-invariance)
\end{enumerate}
where $\Pi_U$ and $\Pi_{U^\perp}$ are the orthogonal projectors onto $U$ and $U^\perp$.

\end{theorem}

The first two properties have a classical interpretation. For a directed
graph $G=([N],E)$ and its associated space $W_G$, assign each vertex $x$ positive\footnote{
The weights are positive because \Cref{thm:ze-mm-informal} restricts to the support of $\Lambda$. 
This restriction is necessary since edges incident to zero-weight vertices do not affect the energy and thus are unconstrained by the zero-energy condition.
}
weight $w_x:=\lambda_x^2$. Then $m_\Lambda(W_G)=\sum_{(x,y)\in E}w_xw_y$ and
$\mathcal E_\Lambda(W_G)=\sum_{(x,y),(y,z)\in E}w_xw_yw_z$. Thus $\Lambda$-mass
is the total weight of edges, and $\Lambda$-energy is the total weight of
directed $2$-walks. Since all weights are positive, 
a zero-energy $G$ must have no $2$-walks, and thus be contained in a directed cut. As $\sum_x
w_x=1$, the mass reaches the extremal threshold $1/4$ precisely when $G$ is a
complete directed cut $L\times R$, as in \ref{W1}, with each side carrying half
of the total vertex weight, as in \ref{W2}. 
In the reduction from $4$ to $3$ colors, these two properties allow us to restrict to the smaller space $\mathrm{L}(U,U)$, where one color vanishes while a constant fraction of the total $\Lambda$-mass remains.

The invariance property \ref{W3} is specific to the quantum setting, since classically $U$ and $U^\perp$ are
spanned by standard basis vectors, so a diagonal $\Lambda$ automatically preserves them. 
This invariance 
is crucial to the $4$-color reduction because it ensures that weighted multiplication is preserved inside  $\mathrm{L}(U,U)$. 
Establishing its approximate counterpart is the main technical challenge in
proving the weighted stability theorem below. 

We next need a robust version of this extremal statement. 
To prove the weighted supersaturation,
we use the operational interpretation that mass is the marginal probability of
a color, while energy is the probability of seeing that color at two
adjacent nodes. Then a bounded-dependence inequality
\cite{gandolfi1989extremal} implies that mass $1/4+\eta$ forces
$\Omega(\eta)$ energy; see \Cref{coro:weighted-large-energy}. 
Moreover, we prove the following weighted stability theorem.

\begin{theorem}[informal, see \Cref{thm:weighted-stability}] \label{thm:weighted-stability-informal}
If $W$ has $\Lambda$-mass $1/4$ and $\Lambda$-energy $\epsilon=o(1)$,
then there exist a nearby PSD weight matrix $\Lambda'$ and a nearby cut space
$W'=\mathrm L(U^\perp,U)$ such that $W'$ has $\Lambda'$-mass $1/4$ and
$\Lambda'$-energy exactly zero. 
\end{theorem}

By plugging the weighted stability theorem into the previous four-color
arguments, we can show that any decomposition $\bigoplus_{a\in[4]}W_a$, under
any weight $\Lambda$, has a total $\Lambda$-energy $\sum_{a\in[4]}\mathcal
E_\Lambda(W_a)=\Omega(1)$. This gives a constant lower bound on the local
collision probability for arbitrary quantum messages, which proves
\Cref{thm:q4-informal}.

\paragraph{Weighted stability from approximate invariance.} 
The new technical challenge in proving
\Cref{thm:weighted-stability-informal} is to establish an 
approximate version of the invariance property \ref{W3}.
As in the uniform case, we apply
the spectral source--sink decomposition to the weighted out-degree matrix
\(
    B:=\sum_k M_k\Lambda^2M_k^\dagger,
\)
where $\{M_k\}$ is an orthonormal basis of $W$. Writing the spectral decomposition
as $B=\sum_{i=1}^N \beta_i u_i u_i^\dagger$ and partitioning the eigenvectors
$\{u_i\}$ according to the eigenvalue threshold $\theta:=\sqrt\epsilon$ yields
the candidate decomposition $\mathbb C^N=U\oplus U^\perp$ and cut space
$W'=\mathrm L(U^\perp,U)$. The eigenvectors $\{u_i\}$ again play the role
of vertices in the classical picture. However, 
the eigenvectors need not be invariant under the weight $\Lambda$, so the weighted energy does not admit a simple
scalar expression in terms of ``flows'' between eigenvectors as in \eqref{eq:energy-noncomm}.

To overcome this obstacle, we prove that small energy forces \emph{approximate
$\Lambda$-invariance}, i.e., $\Lambda$ approximately preserves both $U$ and
$U^\perp$. We can then construct a nearby weight matrix $\Lambda'$ by deleting
the off-diagonal blocks of $\Lambda$ with respect to $U\oplus U^\perp$ and
rescaling the two diagonal blocks so that each side carries half of the total
weight.\footnote{
    The new weight $\Lambda'$ may not be diagonal, 
    but it remains a normalized PSD matrix. Every such $\Lambda'$ can be vectorized as a bipartite
    quantum state $\V(\bar{\Lambda}')$, and the quantum algorithm using this
    state produces the corresponding $\Lambda'$-mass and $\Lambda'$-energy; see
    \Cref{lem:op-energy}.  Thus our formal definitions do not require the weight
    matrix to be diagonal.
}
This produces an exact zero-energy pair $(W',\Lambda')$ with mass $1/4$.
Combining approximate invariance with the estimates from the uniform case
gives the desired closeness of $(W,\Lambda)$ to $(W',\Lambda')$.

Finally, to quantify the approximate $\Lambda$-invariance, we bound the
off-diagonal quantity
\[
    \delta:=\|\Pi_U\Lambda\Pi_{U^\perp}\|_F^2.
\]
We prove that $\delta$ is small by switching to 
a view of quantum state discrimination. Towards this, consider the two subnormalized states
\[
    \sigma_0:=\Lambda \Pi_{U^\perp}\Lambda,
    \qquad
    \sigma_1:=\Lambda \Pi_U\Lambda.
\]
The pretty good measurement of the states above is essentially $\{\Pi_{U^\perp},\Pi_U\}$, and a direct
calculation shows that its total error is exactly $2\delta$.  Thus, instead
of bounding this off-diagonal block directly, 
we explicitly construct a measurement that distinguishes $\sigma_0$ from $\sigma_1$ with $O(\sqrt\epsilon)$ error; see \Cref{lem:delta-bound}.
The pretty-good-measurement bound \cite{barnum2002reversing} then implies $\delta=O(\sqrt\epsilon)$, so $\Lambda$ approximately preserves $U\oplus U^\perp$. Combining
this approximate invariance with the bounds in the unweighted case
gives the desired
closeness of $(W,\Lambda)$ to $(W',\Lambda')$.

\section{Preliminaries} \label{sec:prelim}

\subsection{Matrix Norms}\label{sec:norm}

\paragraph{Notation.}
We use $[n]$ to denote the set $\{1,2,\ldots,n\}$.
Given two linear spaces $S$ and $T$, we denote by $\mathrm{L}(S,T)$ the space of
linear operators from $S$ to $T$, $\mathrm{L}(S):=\mathrm{L}(S,S)$,
denote by $S^\perp$ the orthogonal complement of $S$,
and denote by $\Pi_S$ the orthogonal projector onto $S$.
Given a matrix $X\in\mathbb{C}^{M\times N}$, we denote by 
$\bar{X}$ the conjugate of $X$, $X^\top$ the transpose of $X$, $X^\dagger := \bar{X}^\top$ the conjugate transpose of $X$,
and $X^+$ the Moore--Penrose inverse (pseudoinverse) of $X$.

\begin{definition}[Frobenius norm]
    Given a matrix $X\in\mathbb{C}^{M\times N}$, the Frobenius norm of $X$ is defined as
    \[
    \|X\|_F := \sqrt{\Tr(X^\dagger X)} = \sqrt{\sum_{i,j} |X_{i,j}|^2}.
    \]
\end{definition}

\begin{definition}[Schatten $p$-norm]
    Given $p \in [1, \infty]$ and
    any matrix $X\in\mathbb{C}^{M\times N}$, the Schatten $p$-norm of $X$ is defined as
    \[
    \|X\|_p := \left(\Tr(|X|^p)\right)^{1/p}
    \quad\text{for $p\in[1,\infty)$,}\qquad
    \|X\|_\infty = \max_{\|v\|_2=1} \|Xv\|_2,
    \]
    where $|X| := \sqrt{X^\dagger X}$.
    In particular, $\|X\|_1$ is the trace norm, $\|X\|_2$ coincides with $\|X\|_F$, and
    $\|X\|_\infty$ is the operator norm.
\end{definition}

\begin{fact}[Schatten H\"older's inequality]
    Given $p, q \in [1, \infty]$ such that $\frac1p+ \frac1q = 1$, for any matrices $X,
    Y\in\mathbb{C}^{M\times N}$, we have
    \[
    |\Tr(X^\dagger Y)| \leq \|X\|_p \|Y\|_q \quad\text{and}\quad
    \|X^\dagger Y\|_1 \leq \|X\|_p \|Y\|_q.
    \]
\end{fact}

\paragraph{Vectorization.} We introduce a useful correspondence between matrices and bipartite vectors.
Let $\mathcal{A}:=\mathbb{C}^M$ and $\mathcal{B}:=\mathbb{C}^N$, with standard bases
$\{e_i\}_{i\in[M]}$ and $\{e_j\}_{j\in[N]}$, respectively. Let $E_{i,j}\in\mathbb{C}^{M\times N}$
be the matrix unit with a $1$ in entry $(i,j)$ and $0$ elsewhere. Define
\[
\V:\mathbb{C}^{M\times N}\to \mathcal{A}\otimes \mathcal{B}
\qquad\text{by}\qquad
\V(E_{i,j}) := e_i\otimes e_j.
\]

Given a space $W\subseteq \mathbb{C}^{M\times N}$, we abuse the notation
\[
\V(W):=\{\V(A):A\in W\}\subseteq \mathbb{C}^{M}\otimes \mathbb{C}^{N}.
\]

We present some useful vectorization identities. 
For vectors $u\in\mathcal{A}$ and $v\in\mathcal{B}$,
\[
\V(uv^\top)=u\otimes v,
\qquad
\V(uv^\dagger)=u\otimes \bar v.
\]
For matrices $X,Y\in\mathbb{C}^{M\times N}$, vectorization preserves the Hilbert--Schmidt inner product:
\begin{equation*}
\langle X,Y\rangle := \Tr(X^\dagger Y)
=
\langle \V(X),\V(Y)\rangle.
\end{equation*}
For $X_0\in\mathrm{L}(\mathcal{A})\cong\mathbb{C}^{M\times M}$ and
$X_1\in\mathrm{L}(\mathcal{B})\cong\mathbb{C}^{N\times N}$,
\begin{equation*} 
(X_0\otimes X_1)\V(Y) = \V\!\left(X_0YX_1^\top\right).
\end{equation*}
Finally, for $X,Y\in\mathbb{C}^{M\times N}$,
\begin{equation*}
\Tr_{\mathcal{B}}\!\left(\V(X)\V(Y)^\dagger\right)=XY^\dagger,
\qquad
\Tr_{\mathcal{A}}\!\left(\V(X)\V(Y)^\dagger\right)=X^\top \bar{Y}.
\end{equation*}

We also introduce a seminorm
on bipartite matrices, called the $\Lambda$-norm,
defined by the trace norm weighted by a PSD matrix $\Lambda$.

\begin{definition}[$\Lambda$-norm]
    Let $\Lambda\in\mathbb{C}^{N\times N}$ be a PSD matrix.
    Given any matrix $X\in \mathrm{L}(\mathbb{C}^N\otimes \mathbb{C}^N)$, define its $\Lambda$-norm as
    \[
    \|X\|_\Lambda := \|(\Lambda\otimes \bar{\Lambda}) X(\Lambda\otimes \bar{\Lambda})\|_1.
    \]
\end{definition}

The factor $\Lambda\otimes \bar{\Lambda}$ arises naturally from the vectorization
identity \(\V(\Lambda X \Lambda) = (\Lambda\otimes \bar{\Lambda})\V(X)\).
The $\Lambda$-norm is only a seminorm for general PSD $\Lambda$,
but becomes a proper norm if $\Lambda\succ 0$.

\subsection{Quantum Information}

We introduce some tools from quantum information
used in this paper.

\begin{definition}[Pretty Good Measurement] \label{def:pgm}
    Given $\sigma_0, \sigma_1 \succeq 0$, define 
    $\sigma := \sigma_0 + \sigma_1$, and assume $\Tr(\sigma)=1$.
    The \emph{Pretty Good Measurement} (PGM) is a $2$-outcome POVM defined as
    \[
    M_0 := \sqrt{\sigma^+}\sigma_0 \sqrt{\sigma^+} + \frac{1}{2}\Pi_{\mathrm{ker}(\sigma)},
    \quad
    M_1 := \sqrt{\sigma^+} \sigma_1 \sqrt{\sigma^+} + \frac{1}{2}\Pi_{\mathrm{ker}(\sigma)}.
    \]
\end{definition}

\begin{lemma}[\cite{barnum2002reversing}] \label{lem:pgm}
    Given $\sigma_0, \sigma_1 \succeq 0$
    with $\Tr(\sigma_0 + \sigma_1)=1$, 
    define the optimal 
    testing error
    \[
    \mathrm{OPT}(\sigma_0, \sigma_1) := 
    \min_{0\preceq M \preceq I} \Tr(M\sigma_1) + \Tr((I-M)\sigma_0).
    \]
    Let $\{M_0, M_1\}$ be the pretty good measurement for $\sigma_0, \sigma_1$. Then
    the PGM testing error 
    \[
    \mathrm{PGM}(\sigma_0, \sigma_1):=
    \mathrm{Tr}(M_0\sigma_1) + \Tr(M_1\sigma_0) 
    \leq 
    2 \mathrm{OPT}(\sigma_0, \sigma_1).
    \]
\end{lemma}

The following is a perturbation lemma for $4$-outcome POVMs, with distances measured in the \(\Lambda\)-norm. It will be used in the final POVM rounding step in \Cref{sec:robust}, and its proof is deferred to Appendix \ref{appx:xy}.

\begin{lemma} \label{lem:project-povm}
Given a $4$-outcome POVM $\{P_i\}_{i\in[4]}$ on
$\mathbb{C}^N\otimes \mathbb{C}^N$
and another PSD operator $Z$ with $0\preceq Z\preceq I$, let $\Lambda\in \mathbb{C}^{N\times N}$
be PSD with $\|\Lambda\|_F=1$. If
\[
\|P_1-Z\|_\Lambda = \epsilon,
\]
then we can construct another POVM $\{P_i'\}_{i\in[4]}$ such that $P_1'=Z$ and
for each $2\leq i\leq 4$, 
\[
\|P_i-P_i'\|_\Lambda = O(\sqrt{\epsilon}).
\]
\end{lemma}

\subsection{Models in Distributed Computing} \label{sec:prelim-local}

We recall the classical LOCAL model
\cite{linial1992locality,peleg2000distributed} and its quantum counterpart
\cite{gavoille2009can,arfaoui2014without}, then specialize to one-way one-round
algorithms on directed cycles.

\begin{definition}[LOCAL models] \label{def:local}
Let $G=(V,E)$ be an $n$-vertex graph. Each node $v\in V$ initially knows $n$ and its unique identifier $\ID_v\in[\poly(n)]$, and may receive a
local input $x_v$. 
Incident edges are distinguished by local port numbers.
All nodes execute the same algorithm.

Computation proceeds in synchronous rounds. In each round, every node
sends a message to each neighbor, receives a message from each neighbor, and updates its local state. After the final round, each node
produces a classical local output $y_v$.
\begin{enumerate}
    \item In the \emph{randomized-LOCAL model}, local computations
    and messages are classical, and each node has access to independent
    private random bits. 
    \item In the \emph{quantum-LOCAL model}, each node 
    may perform arbitrary local quantum operations, 
    and send quantum messages.
\end{enumerate}
The nodes have no shared randomness or prior entanglement. Local computation
and message length are unbounded but finite. 
The \emph{round complexity} of a LOCAL algorithm is its number of
communication rounds; the \emph{locality} of a problem is the minimum round
complexity of a LOCAL algorithm that solves it.
\end{definition}

We say a LOCAL algorithm succeeds \emph{with high probability} if it produces a
valid solution with probability at least $1-1/\poly(n)$, where the probability
is over private randomness and quantum measurements. A LOCAL algorithm
\emph{solves} a distributed problem if it succeeds with high probability on
every valid input graph, local input, and assignment of unique identifiers.

In this paper, we focus on the following distributed problem.

\begin{definition}[Directed-cycle $q$-coloring] \label{def:cycle-coloring}
Fix an integer $q\geq 2$. Consider a directed $n$-cycle ($n\geq 3$) with nodes
$(v_i)_{i\in\mathbb{Z}_n}$, where each node knows 
the direction of the cycle. The \emph{directed-cycle
$q$-coloring problem} requires each node $v_i$ to output a color $c_i\in[q]$
such that $c_i\neq c_{i+1}$ for every $i\in\mathbb{Z}_n$.
\end{definition}

For $q\geq 3$, proper cycle $q$-colorings always exist.
We focus on the following restricted model.

\begin{definition}[One-way one-round anonymous quantum algorithms on cycles]
    \label{def:anonymous-one-round}
Consider a directed cycle with $n$ nodes $(v_i)_{i\in\mathbb{Z}_n}$.
In a \emph{one-way one-round anonymous quantum algorithm}, all nodes
are initially identical and execute the same algorithm. In the single communication round, each
node $v_i$ sends a quantum message to its successor $v_{i+1}$ and receives one from its predecessor
$v_{i-1}$. The local computation and message length are unbounded but finite. 
\end{definition}

In the high-probability setting, 
anonymous quantum algorithms are equivalent to 
quantum-LOCAL algorithms, since the nodes can independently
sample unique identifiers from a sufficiently large 
set $[n^{C_0}]$,
which yields pairwise distinct identifiers
with probability at least $1-1/n^{C_0-2}$.

\paragraph{Finite dependence.}
Let $(X_i)_{i\in\mathbb{Z}_n}$ be random labels on a cycle. Their joint
distribution is \emph{stationary} if it is invariant under cyclic shifts,
that is, $(X_i)_{i\in\mathbb{Z}_n}$ and $(X_{i+1})_{i\in\mathbb{Z}_n}$ are equal in law.
It is \emph{$k$-dependent} if, for any two vertex sets
$U,V\subseteq\mathbb{Z}_n$ with distance 
larger than $k$ on the cycle, 
the random vectors
$(X_i)_{i\in U}$ and $(X_i)_{i\in V}$ are independent. These notions extend
naturally to processes $(X_t)_{t\in\mathbb{Z}}$ on the infinite line, with
stationarity meaning invariance under integer shifts.

Every one-way one-round anonymous quantum algorithm produces a stationary $1$-dependent output
distribution. Stationarity follows because all nodes have identical initial
states and execute the same procedure. $1$-dependence 
follows from observing that the light cones of 
two nodes at distance greater than $1$ do not intersect.

\section{Energy Formulation} \label{sec:energy}
In this section, we formally define the notions of $\Lambda$-mass and
$\Lambda$-energy and establish their basic properties. These notions are
defined for any PSD operators
$P\in\mathrm{L}(\mathbb{C}^N\otimes \mathbb{C}^N)$ and
$\Lambda\in\mathbb{C}^{N\times N}$. In particular, they apply when $P$ is an
orthogonal projector onto a matrix subspace under the vectorization equivalence
$\mathbb{C}^{N\times N}\cong \mathbb{C}^N\otimes \mathbb{C}^N$, and, more
generally, when $P$ is a POVM element on the same space. 

\begin{definition}[$\Lambda$-mass and $\Lambda$-energy] \label{def:energy} 
Given PSD $P\in\mathrm{L}(\mathbb{C}^N\otimes \mathbb{C}^N)$ 
and PSD $\Lambda\in\mathbb{C}^{N\times N}$, define the
$\Lambda$-mass of $P$ by
\[
m_{\Lambda}(P) := 
\|P\|_\Lambda = 
\mathrm{Tr}\left[P(\Lambda^2\otimes \bar{\Lambda}^2)\right],
\]
and the $\Lambda$-energy of $P$ by
\[
\mathcal{E}_{\Lambda}(P) := \mathrm{Tr}\left(
\Lambda A \Lambda B
\right),
\]
where
$A:=\mathrm{Tr}_{\mathcal{A}}((\Lambda^2\otimes I)P)^\top$ and 
$B:=\mathrm{Tr}_{\mathcal{B}}((I\otimes \bar{\Lambda}^2)P)$.

\end{definition}

\Cref{sec:spectral} proves an equivalent spectral formula for
$\Lambda$-energy, matching \Cref{def:energy-projective} in the introduction. \Cref{sec:examples} then illustrates the definition through several
examples, several of which already appear in the technical overview. Finally, \Cref{sec:operational}
interprets mass and energy as probabilities arising from a specific quantum
experiment, while \Cref{sec:properties} establishes properties of the energy
that will be used in the main proof.

\subsection{Spectral Formula} \label{sec:spectral}

The following fact gives an alternative expression for $\Lambda$-energy,
connecting it to the multiplicative 
structure of the corresponding matrix space.

\begin{lemma} \label{fact:energy}
Given $P$ and $\Lambda$ as above, let
$\{M_k\}_k\subseteq\mathbb{C}^{N\times N}$ be any finite family of matrices satisfying
\[
P = \sum_{k} \V(M_k)\V(M_k)^\dagger.
\]
Then
\[
\mathcal{E}_\Lambda(P) = \sum_{k,\ell} \|\Lambda M_k\Lambda M_\ell \Lambda\|_F^2.
\]
Such $\{M_k\}_k$ always exists.
In particular, it can be obtained by taking a spectral decomposition of $P$.
\end{lemma}

\begin{proof}
    Note that $\Lambda$ is Hermitian,
    so $\Lambda=\Lambda^\dagger$ and $\bar{\Lambda}=\Lambda^\top$.
    For each $M_k$, we have
    \begin{align*}
        \Tr_\As\!\left((\Lambda^2\otimes I)\V(M_k)\V(M_k)^\dagger\right)
        &= \Tr_\As\!\left((\Lambda\otimes I)\V(M_k)\V(M_k)^\dagger(\Lambda\otimes I)\right) \\
        &= \Tr_\As\!\left(\V(\Lambda M_k)\V(\Lambda M_k)^\dagger\right) 
        = (\Lambda M_k)^\top \overline{\Lambda M_k}
        = (M_k^\dagger \Lambda^2 M_k)^\top,
    \end{align*}
    where the first equality is by cyclicity of partial trace, the second
    equality is by $(X_0\otimes X_1)\V(Y) = \V(X_0YX_1^\top)$, the third equality is by
    $\Tr_\As\!\left(\V(X)\V(Y)^\dagger\right)=X^\top \bar{Y}$. Then
    \[
    A := \Tr_\As((\Lambda^2\otimes I)P)^\top
    = \sum_{k} \Tr_\As((\Lambda^2\otimes I)\V(M_k)\V(M_k)^\dagger)^\top
    = \sum_{k} M_k^\dagger \Lambda^2 M_k.
    \]
    
    Similarly, we have
    \begin{align*}
        \Tr_\Bs\!\left((I\otimes \bar{\Lambda}^2)\V(M_k)\V(M_k)^\dagger\right)
        &= \Tr_\Bs\!\left((I\otimes \bar{\Lambda})\V(M_k)\V(M_k)^\dagger(I\otimes \bar{\Lambda})\right) \\
        &= \Tr_\Bs\!\left(\V(M_k\Lambda)\V(M_k\Lambda)^\dagger\right) 
        = (M_k\Lambda)(M_k\Lambda)^\dagger
        = M_k \Lambda^2 M_k^\dagger,
    \end{align*}
    and then
    \[
    B := \Tr_\Bs((I\otimes \bar{\Lambda}^2)P) = \sum_{k} M_k \Lambda^2 M_k^\dagger.
    \]

    Thus we have
    \begin{align*}
    \Es_\Lambda(P) = \Tr(\Lambda A \Lambda B)
    &= \sum_{k,\ell}\Tr\left(
        \Lambda M_k^\dagger \Lambda^2 M_k \Lambda M_\ell \Lambda^2 M_\ell^\dagger
        \right)\\
    &= \sum_{k,\ell}\Tr\left(
        \Lambda M_\ell^\dagger \Lambda M_k^\dagger \Lambda \cdot \Lambda M_k \Lambda M_\ell \Lambda
        \right) 
    = \sum_{k,\ell} \|\Lambda M_k\Lambda M_\ell \Lambda\|_F^2. 
    \end{align*}

    In particular,
    taking spectral decomposition $P=\sum_k p_k v_k v_k^\dagger$ and
    letting $M_k = \sqrt{p_k} \V^{-1}( v_k)$ for each $k$ give a valid choice of $\{M_k\}_k$.
\end{proof}

For completeness, we also show that $\Lambda$-mass can be expressed in a similar way,
which matches the definition appearing in the technical overview.

\begin{lemma}
    Given the same setting as in \Cref{fact:energy}, we have
    \[
    m_\Lambda(P) = \sum_k \|\Lambda M_k \Lambda\|_F^2.
    \]
\end{lemma}

\begin{proof}
    By definition, we have
    \begin{align*}
    m_\Lambda(P) = 
    \Tr\left[P(\Lambda^2\otimes \bar{\Lambda}^2)\right]
    &= \sum_k 
    \Tr\left[
        (\Lambda\otimes \bar{\Lambda})
        \V(M_k)\V(M_k)^\dagger
        (\Lambda\otimes \bar{\Lambda})
    \right]\\
    &= \sum_k \Tr\left[
        \V(\Lambda M_k \Lambda)\V(\Lambda M_k \Lambda)^\dagger
    \right]
    = \sum_k \|\Lambda M_k \Lambda\|_F^2. \tag*{\qedhere}
    \end{align*}
\end{proof}

\subsection{Examples} \label{sec:examples}

In this paper, we are mainly interested in the energy of an operator $P$ with
$0\preceq P\preceq I$, as we will use it to analyze POVM elements.
If we further restrict $P$ to be an orthogonal projector
onto a matrix space $W$, then the energy of $P$ captures the multiplicative
structure of $W$. Below are several illustrative examples.

\begin{example}[Energy of matrix space]
    Let $\Lambda=I$.
    Let $\Pi_W$ be the orthogonal
    projector onto a matrix space $W\subseteq \mathbb{C}^{N\times N}$.
    Then $m_I(\Pi_W) = \Tr(\Pi_W) = \dim(W)$, and by \Cref{fact:energy}, we have
    \[
    \Es_I(\Pi_W) = \sum_{k,\ell} \|M_k M_\ell\|_F^2,
    \]
    where $\{M_k\}$ is an orthonormal basis of $W$. 
    Thus $\Es_I(\Pi_W)$ measures the multiplicative energy of space $W$.

    Moreover, $\Es_I(\Pi_W)=0$ if and only if 
    $W^2 := \mathrm{span}\{XY : X,Y \in W\} = \{0\}$.
\end{example} 

\begin{example}[Counting $2$-walks in a graph]
    Let $\Lambda=I$.
    Given a directed graph $G=(V,E)$, 
    define a matrix space 
    \[
    W := \mathrm{span}\{E_{uv}: (u,v)\in E\} \subseteq \mathbb{C}^{V\times V}.
    \]
    Then for the orthogonal projector $\Pi_W$ onto $W$, we have
    $m_I(\Pi_W) = |E|$, and
    \[
    \Es_I(\Pi_W) = 
    \sum_{(u,v),(x,w)\in E} \|E_{uv}E_{xw}\|_F^2
    = \sum_{(u,v),(x,w)\in E} \mathbf{1}_{v=x}
    = \# \{(u,v,w): (u,v),(v,w)\in E\},
    \]
    which counts the number of $2$-walks in $G$.
    
    Moreover, $\Es_I(\Pi_W)=0$ if and only if $G$ is contained in
    a directed cut. 
\end{example}

\begin{example}[Energy of a weighted graph]
    We extend the previous example to the setting where
    each node $v\in V$ has a weight $w_v\geq 0$, and define
    \[
    \Lambda := \mathrm{diag}(\sqrt{w_v})_{v\in V}.
    \]
    Then $m_\Lambda(\Pi_W) = \sum_{(u,v)\in E} w_u w_v$
    is the weighted sum of directed edges in $G$, and 
    \[
    \Es_\Lambda(\Pi_W) =
    \sum_{(u,v),(x,w)\in E} \|\Lambda E_{uv}\Lambda E_{xw}\Lambda\|_F^2
    = 
    \sum_{
        \text{$2$-walk (u,v,w) in $G$}
    } w_u w_v w_w
    \]
    is the weighted sum of $2$-walks in $G$.
\end{example}

\subsection{The Operational Interpretation} \label{sec:operational}

The next lemma characterizes the operational meanings of 
mass and energy: $\Lambda$-mass is the marginal probability of obtaining the outcome corresponding to POVM operator $P$,
and $\Lambda$-energy is the probability of obtaining the corresponding outcomes for $P\otimes P$ on some quantum
state specified by $\Lambda$.

\begin{lemma} \label{lem:op-energy} 
    
    Given PSD $P \in \mathrm{L}(\mathbb{C}^N\otimes \mathbb{C}^N)$ with $0\preceq P\preceq I$ and PSD
    $\Lambda \in\mathbb{C}^{N\times N}$ with $\|\Lambda\|_F=1$, define a bipartite quantum state
\[
\ket{\Psi} := \V(\bar{\Lambda}),
\]
and its reduced density matrices
\[
\rho := \mathrm{Tr}_{\mathcal{A}}(\ket{\Psi}\bra{\Psi}), \quad
\bar{\rho} := \mathrm{Tr}_{\mathcal{B}}(\ket{\Psi}\bra{\Psi}).
\]
Then we have
\[
m_\Lambda(P)=\mathrm{Tr}\left[
P(\rho\otimes \bar{\rho} )
\right],
\quad
\mathcal{E}_\Lambda(P) = \mathrm{Tr}\left[
(P\otimes P)(\rho\otimes \ket{\Psi}\bra{\Psi}\otimes \bar{\rho})
\right].
\]
\end{lemma}

\begin{proof}
    Since $\Lambda$ is PSD, write the spectral decomposition
    \[
    \Lambda = \sum_{i} \lambda_i \ket{\varphi_i}\bra{\varphi_i}
    \quad\text{and}\quad
    \bar{\Lambda} = \sum_{i} \lambda_i \ket{\bar{\varphi}_i}\bra{\bar{\varphi}_i},
    \]
    where $\ket{\varphi_i}$ are eigenvectors 
    and $\lambda_i\geq 0$ are eigenvalues of $\Lambda$.
    By $\V(xy^\dagger)=x\otimes \bar{y}$, we have
    \[
    \ket{\Psi} = \V(\bar{\Lambda}) = \sum_{i} \lambda_i \ket{\bar{\varphi}_i}\otimes \ket{\varphi_i},
    \]
    \[
    \rho = \Tr_\As(\ket{\Psi}\bra{\Psi})
    =\sum_{i} \lambda_i^2 \ket{\varphi_i}\bra{\varphi_i} = \Lambda^2,
    \quad\text{and}\quad
    \bar{\rho} = \Tr_\Bs(\ket{\Psi}\bra{\Psi})
    =\sum_{i} \lambda_i^2 \ket{\bar{\varphi_i}}\bra{\bar{\varphi}_i} = \bar{\Lambda}^2.
    \]
    Then
    \[
    \Tr(P(\rho\otimes \bar{\rho}))
    = \Tr(P(\Lambda^2\otimes \bar{\Lambda}^2))
    = m_\Lambda(P).
    \]
    
    By definition of $\rho$, we have
    \begin{equation} \label{eq:pp-psi}
    \Tr[(P\otimes P)(\rho\otimes \ket{\Psi}\bra{\Psi}\otimes \bar{\rho})]
    = \sum_{a,c} \lambda_a^2 \lambda_c^2
    \Tr\left[(P\otimes P)(\ket{\varphi_a,\Psi,\bar{\varphi}_c}\bra{\varphi_a,\Psi,\bar{\varphi}_c})\right].
    \end{equation}
    If we write the spectral decomposition of $P$ as
    \(
    P = \sum_{k} \V(M_k)\V(M_k)^\dagger,
    \)
    then
    \begin{equation} \label{eq:pp-psi2}
    \Tr\left[(P\otimes P)(\ket{\varphi_a,\Psi,\bar{\varphi}_c}\bra{\varphi_a,\Psi,\bar{\varphi}_c})\right]
    = \sum_{k,\ell} |\gamma_{k,\ell}^{a,c}|^2,
    \end{equation}
    where
    \(
    \gamma_{k,\ell}^{a,c} := 
    \bra{\varphi_a,\Psi,\bar{\varphi}_c}(\V(M_k)\otimes \V(M_\ell)).
    \)

    By definition of $\ket{\Psi}$
    and the fact that $\lambda_b$ are real, we have
    \begin{align*}
    \gamma_{k,\ell}^{a,c} &=
    \sum_b \lambda_b
    \bra{\varphi_a,\bar{\varphi}_b,\varphi_b,\bar{\varphi}_c}
    \V(M_k)\otimes \V(M_\ell) \\
    &= \sum_b 
        \lambda_b
        \bra{\varphi_a,\bar{\varphi}_b}\V(M_k) \cdot
        \bra{\varphi_b,\bar{\varphi}_c}\V(M_\ell)\\
    &= \sum_b
        \lambda_b\cdot
        \V(\ket{\varphi_a}\bra{\varphi_b})^\dagger\V(M_k) \cdot
        \V(\ket{\varphi_b}\bra{\varphi_c})^\dagger \V(M_\ell),
    \end{align*}
    where the last equality is by $\V(xy^\dagger)=x\otimes \bar{y}$. 
    As vectorization preserves the Hilbert--Schmidt inner product, we have
    \[
    \V(\ket{\varphi_a}\bra{\varphi_b})^\dagger\V(M_k) 
    = \Tr\left(\left(\ket{\varphi_a}\bra{\varphi_b}\right)^\dagger M_k\right)
    = \bra{\varphi_a}M_k\ket{\varphi_b}.
    \]
    Thus
    \[
    \gamma_{k,\ell}^{a,c} =
    \sum_b 
        \lambda_b\bra{\varphi_a}M_k\ket{\varphi_b} \cdot
        \bra{\varphi_b}M_\ell\ket{\varphi_c}=
    \bra{\varphi_a}M_k \Lambda M_\ell\ket{\varphi_c}.
    \]    
    Plugging this into \eqref{eq:pp-psi2} and then \eqref{eq:pp-psi}, we have
    \begin{align*}
    \Tr[(P\otimes P)(\rho\otimes \ket{\Psi}\bra{\Psi}\otimes \bar{\rho})]
    &= \sum_{a,c} \lambda_a^2 \lambda_c^2
    \sum_{k,\ell} \left|\bra{\varphi_a}M_k \Lambda M_\ell\ket{\varphi_c}\right|^2 \\
    &= \sum_{k,\ell} \sum_{a,c} |\bra{\varphi_a} \Lambda M_k \Lambda M_\ell \Lambda \ket{\varphi_c}|^2 \\
    &= \sum_{k,\ell} \|\Lambda M_k \Lambda M_\ell \Lambda\|_F^2
    = \mathcal{E}_\Lambda(P),
    \end{align*}
    where the second equality follows from $\Lambda \ket{\varphi_x} = \lambda_x \ket{\varphi_x}$,
    and the third equality follows because the vectors $\ket{\varphi_x}$ form an orthonormal basis,
    and the last equality is by \Cref{fact:energy}.
\end{proof}

\subsection{Useful Properties} \label{sec:properties}

We prove some natural properties 
of the energy function, which will be used
in later sections.
The first fact is that the energy is monotone with respect to PSD ordering.

\begin{fact}
    Given two PSD operators $P, Q\in \mathrm{L}(\mathbb{C}^N\otimes \mathbb{C}^N)$,
    and PSD $\Lambda\in\mathbb{C}^{N\times N}$, if $P\preceq Q$,
    then $\Es_\Lambda(P) \leq \Es_\Lambda(Q)$.
\end{fact}

\begin{proof}
    Let $R = Q - P$, then $R \succeq 0$.
    Define 
    \[
    \Es_\Lambda(Q) - \Es_\Lambda(P) =
    \Es_\Lambda(P + R) - \Es_\Lambda(P) =
    \Tr(\Lambda A_P \Lambda B_R) + 
    \Tr(\Lambda A_R \Lambda B_P) + 
    \Tr(\Lambda A_R \Lambda B_R),
    \]
    where $A_P := \Tr_\As((\Lambda^2\otimes I)P)^\top\succeq 0$,
    $B_P := \Tr_\Bs((I\otimes \bar{\Lambda}^2)P)\succeq 0$, and $A_R$, $B_R$
    are defined similarly for $R$. 
    Since $A_P, B_P, A_R, B_R \succeq 0$, 
    the three terms 
    $\Tr(\Lambda A_P \Lambda B_R)$, 
    $\Tr(\Lambda A_R \Lambda B_P)$, and $\Tr(\Lambda A_R \Lambda B_R)$
    are all non-negative, and thus $\Es_\Lambda(Q) - \Es_\Lambda(P) \geq 0$.
\end{proof}

We conclude with a $\Lambda$-Lipschitz lemma for the energy function. Namely, for fixed
$P$, $\Es_\Lambda(P)$ is Lipschitz continuous in $\Lambda$ with respect to the Frobenius norm. 
We defer its proof to Appendix \ref{appx:weighted-lip-psd}.

\begin{lemma}\label{lem:weighted-lip-psd}
    Given PSD $P\in \mathrm{L}(\mathbb{C}^N\otimes \mathbb{C}^N)$
    with $0\preceq P\preceq I$,
    and PSDs $\Lambda, \Lambda' \in\mathbb{C}^{N\times N}$ such that 
    $\|\Lambda\|_F=\|\Lambda'\|_F=1$, we have
    \[
    \left|\Es_{\Lambda}(P)-\Es_{\Lambda'}(P)\right|
    \leq 
    6\,\|\Lambda-\Lambda'\|_F.
    \]
\end{lemma}

\section{Zero-Energy Operators} \label{sec:zero-energy}

This section establishes the structural results on low-energy operators that
drive the four-color impossibility proof. We first show that a POVM operator
with zero $\Lambda$-energy has $\Lambda$-mass at most $1/4$ and characterize the
extremal case as the projector onto an off-diagonal matrix space associated with
a balanced $\Lambda$-invariant decomposition (\Cref{thm:ze-mm}). We then prove a
robust version: any POVM operator of mass $1/4$ and small energy can be approximated,
together with its weight matrix, by an exact zero-energy extremizer
(\Cref{thm:weighted-stability}). The exact characterization yields the
reduction from four to three colors in \Cref{sec:unrobust}, while the robust theorem
is used for the general impossibility proof in \Cref{sec:robust}.

\subsection{Structure of Zero-Energy Operators}
\label{sec:ze-structure}


We first restrict our attention to strictly positive $\Lambda$ to avoid
degeneracy.
The following theorem characterizes the structure of operators with maximal mass and exactly zero energy.

\begin{theorem}[Characterization of operators with zero energy and maximal mass] \label{thm:ze-mm} 
    Given PSD $P\in \mathrm{L}(\mathbb{C}^N\otimes \mathbb{C}^N)$ with $0\preceq
    P\preceq I$, and PSD $\Lambda\in\mathbb{C}^{N\times N}$ with
    $\|\Lambda\|_F=1$ and $\Lambda\succ 0$, if
    \[
    m_\Lambda(P) = \frac14 \quad\text{and}\quad
    \mathcal{E}_\Lambda(P) = 0,
    \]
    then there exists an orthogonal decomposition $\mathbb{C}^N = U \oplus U^\perp$ such that
    \begin{enumerate}[
    label=(P\arabic*),
    ref=(P\arabic*),
    leftmargin=1.5cm,
]
        \item \label{P1} $P=\Pi_{\mathrm{L}(U^\perp,U)}$, i.e., 
        the orthogonal projector onto the space
        \[
        \mathrm{L}(U^\perp, U) := 
        \mathrm{span}\{xy^\dagger: x\in U, y\in U^\perp\} =
        \left\{    
        \left(\begin{array}{c|c}
            0 & X \\ \hline
            0 & 0
        \end{array}\right)
        : X \in \mathbb{C}^{|U|\times |U^\perp|}
        \right\}
        \]
        in the basis of $U\oplus U^\perp$;
        \item \label{P2}
        \(
            \Tr(\Lambda^2\Pi_U)=
            \Tr(\Lambda^2\Pi_{U^\perp})=
            \frac12,
        \)
        where $\Pi_U$ (resp.~$\Pi_{U^\perp}$) is the orthogonal projector onto
        $U$ (resp.~$U^\perp$);
        \item \label{P3} $U$ and $U^\perp$ are $\Lambda$-invariant, i.e.,
        \(
        \Lambda U=U, 
        \Lambda U^\perp=U^\perp
        \).
    \end{enumerate}
\end{theorem}

\begin{remark}
    \Cref{thm:ze-mm} can be viewed as the noncommutative generalization of the
    following classical fact. Given a directed graph $G=(V,E)$ in which each vertex
    $v\in V$ has a weight $w_v>0$, suppose that the weights sum to $1$.
    If $G$ contains no $2$-walks, then there must exist a partition 
    \[
    V = L \sqcup R,
    \]
    such that $E$ only contains edges from $L$ to $R$.
    Moreover, the total weight of edges
    \[
    m := \sum_{(u,v)\in E} w_u w_v \leq 
    \sum_{v\in L} w_v 
    \left(
    1- \sum_{v\in L} w_v
    \right)
    \leq  \frac14.
    \]
    If the maximal $m=1/4$ is reached, then
    $G$ must be a complete directed cut and $W_L = W_R = 1/2$,
    which corresponds to Properties \ref{P1} and \ref{P2}.

    Property \ref{P3} has no classical analogue, as the 
    corresponding weight matrix $\Lambda = \mathrm{diag}(\sqrt{w_v})$ is diagonal
    in the computational basis,
    which always preserves spaces spanned by standard unit vectors.
\end{remark}

\paragraph{Definitions.}
To prove \Cref{thm:ze-mm}, we 
construct the desired $U\oplus U^\perp=\mathbb{C}^N$ as below.
Write the spectral decomposition of 
$P$ as
\[
P = \sum_{k} \V(M_k)\V(M_k)^\dagger,
\]
where $M_k \in \mathbb{C}^{N\times N}$ are pairwise orthogonal matrices.
Define the corresponding matrix space
\[
W := \mathrm{span}\{M_k\} \subseteq \mathbb{C}^{N\times N}.
\]
Then $W$ defines an orthogonal decomposition $\mathbb{C}^N = U \oplus U^\perp$ where
\[
U:= \sum_{X \in W} \mathrm{Im}(X).
\]

We then prove two helper lemmas about zero-energy operators. The first lemma proves a relaxed version of Property \ref{P1}.

\begin{lemma} \label{lem:ze-structure}
    If $\mathcal{E}_\Lambda(P)=0$, then
    \[
    P \preceq 
    \Pi_W \preceq 
    \Pi_{\mathrm{L}((\Lambda U)^\perp, U)} = \Pi_U\otimes\bar{\Pi}_{(\Lambda U)^\perp}
    \]
    where $\Lambda U := \{\Lambda x: x\in U\}$, and $\Pi_S$ denotes
    the orthogonal projection onto a linear space $S$.
\end{lemma}

\begin{proof}
    By $0\preceq P\preceq I$, we have that $\|M_k\|_F^2 = \|\V(M_k)\|^2 \leq 1$.
    Thus
    \begin{equation} \label{eq:p-piw}
    P = \sum_{k} \V(M_k)\V(M_k)^\dagger \preceq 
    \sum_{k} \frac{1}{\|M_k\|_F^2} \V(M_k)\V(M_k)^\dagger = \Pi_W.
    \end{equation}

    By the zero-energy assumption and \Cref{fact:energy}, we have
    \[
    \Es_\Lambda(P) = \sum_{k,\ell} \|\Lambda M_k\Lambda M_\ell \Lambda\|_F^2 = 0.
    \]
    Together with $\Lambda\succ 0$, it implies that 
    \(
    M_k\Lambda M_\ell = 0
    \)
    for any $k,\ell$. As $\{M_k\}$ spans $W$, we have
    \[
    X\Lambda Y = 0 \quad\text{for any } X,Y \in W.
    \]
    For any $X \in W$, we have
    $\mathrm{Im}(X) \subseteq \sum_{Y \in W} \mathrm{Im}(Y) = U$.
    Moreover, we claim that 
    \(
    \Lambda U \subseteq \mathrm{ker}(X).
    \)
    Indeed, for any vector $u \in \Lambda U$, there exist $Y_1,\ldots,Y_{\dim(W)} \in W$ and $z_1,\ldots z_{\dim(W)} \in
    \mathbb{C}^N$ such that $u = \Lambda\sum_i Y_i z_i$. Then 
    \(
    X u = \sum_i (X \Lambda Y_i) z_i = 0,
    \)
    since $X \Lambda Y_i = 0$ for $X,Y_i \in W$. Thus 
    \[
    \mathrm{Im}(X) \subseteq U \quad\text{and}\quad
    \Lambda U \subseteq \mathrm{ker}(X) \quad\text{for any } X \in W,
    \]
    which implies
    \[
    W \subseteq \mathrm{L}((\Lambda U)^\perp, U).
    \]
    Thus
    \begin{equation} \label{eq:piw-pil}
    \Pi_W \preceq \Pi_{\mathrm{L}((\Lambda U)^\perp, U)}.
    \end{equation}
    Let $\{u_i\}$ and $\{v_i\}$ be orthonormal bases of $U$ and $(\Lambda U)^\perp$, respectively. Then
    \begin{equation} \label{eq:pil-pipi}
    \Pi_{\mathrm{L}((\Lambda U)^\perp, U)}
    = \sum_{i,j} \V\left(u_i v_j^\dagger\right)\V\left(u_i v_j^\dagger\right)^\dagger
    = \sum_i u_iu_i^\dagger \otimes \sum_j \overline{v_jv_j^\dagger}
    = \Pi_U \otimes \bar{\Pi}_{(\Lambda U)^\perp},
    \end{equation}
    where the second equality is by $\V(xy^\dagger)=x\otimes \bar{y}$.
    The conclusion follows by \eqref{eq:p-piw}, \eqref{eq:piw-pil}, and \eqref{eq:pil-pipi}.
\end{proof}

The next lemma bounds the quantities in Property \ref{P2}.

\begin{lemma} \label{lem:ze-alpha-beta}
    Define
    $
    \alpha := \Tr(\Lambda^2 \Pi_U)$
    and
    $
    \beta := \Tr(\Lambda^2 \Pi_{\Lambda U}).
    $
    Then $0\leq \alpha \leq \beta\leq 1$.
\end{lemma}

\begin{proof}
    First, by $\Tr(\Lambda^2)=\|\Lambda\|_F^2=1$ and $0\preceq \Pi_U,
    \Pi_{\Lambda U}\preceq I$, we have $\alpha,\beta \in[0,1]$. 

    Let $\{v_i\in \mathbb{C}^N\}_{i\in[r]}$ be an orthonormal basis of $U$, and define matrix
    \[
    V := [v_1,\ldots,v_r] \in \mathbb{C}^{N\times r},
    \]
    and set
    \[
    X := V^\dagger \Lambda^2 V, \quad
    Y := V^\dagger \Lambda^4 V \in \mathbb{C}^{r\times r}.
    \]
    Note that $X\succ 0$ since $\Lambda\succ 0$ and $V$ has full column rank.
    By observing that $\Pi_U = VV^\dagger$, we have
    \begin{equation}
    \label{eq:alpha}        
    \alpha = \mathrm{Tr}(\Lambda^2 \Pi_U) = \mathrm{Tr}(V^\dagger \Lambda^2 V) = \mathrm{Tr}(X).
    \end{equation}
    As the space $\Lambda U$ is spanned by columns of $\tilde{V} := \Lambda V$, 
    its orthogonal projector can be computed as
    \[
    \Pi_{\Lambda U} = 
    \tilde{V}(\tilde{V}^\dagger \tilde{V})^{-1}\tilde{V}^\dagger
    = \Lambda V (V^\dagger \Lambda^2 V)^{-1} V^\dagger \Lambda
    = \Lambda V X^{-1} V^\dagger \Lambda.
    \]
    Then
    \begin{equation} \label{eq:beta}
    \beta = \Tr(\Lambda^2 \Pi_{\Lambda U}) = 
    \Tr(\Lambda^2 \Lambda V X^{-1} V^\dagger \Lambda) =
    \Tr(V^\dagger \Lambda^4 V X^{-1}) = \Tr(Y X^{-1}).
    \end{equation}

    Combining \eqref{eq:alpha} and \eqref{eq:beta}, we have
    \begin{equation} \label{eq:alpha-beta}
    \beta - \alpha = \Tr(Y X^{-1}) - \Tr(X) = \Tr((Y-X^2)X^{-1}) \geq 0,
    \end{equation}
    because $X^{-1} \succ 0$, and 
    \begin{equation} \label{eq:y-x2}
    Y - X^2 = V^\dagger \Lambda^4 V - (V^\dagger \Lambda^2 V)^2 = 
    V^\dagger \Lambda^2 (I-VV^\dagger) \Lambda^2 V =
    V^\dagger \Lambda^2 \Pi_{U^\perp} \Lambda^2 V 
    \succeq 0. \qedhere
    \end{equation}
\end{proof}

Combining the above two lemmas,
we upper bound the mass of a zero-energy operator.

\begin{lemma} \label{lem:ze-mm-upper}
    If $\mathcal{E}_\Lambda(P)=0$, then
    \(
    m_\Lambda(P) \leq \frac14.
    \)
\end{lemma}

\begin{proof}
    By \Cref{lem:ze-structure}, we have
    \[
    P \preceq \Pi_{\mathrm{L}((\Lambda U)^\perp, U)} = 
    \Pi_U\otimes\bar{\Pi}_{(\Lambda U)^\perp},
    \]
    which implies
    \[
    m_\Lambda(P) = \Tr(P(\Lambda^2\otimes \bar{\Lambda}^2)) \leq
    \Tr(\Pi_{\mathrm{L}((\Lambda U)^\perp, U)}(\Lambda^2\otimes \bar{\Lambda}^2))  
    = m_\Lambda(\Pi_{\mathrm{L}((\Lambda U)^\perp, U)}),
    \]
    and 
    \[
    m_\Lambda(\Pi_{\mathrm{L}((\Lambda U)^\perp, U)})
    = \Tr(\Lambda^2 \Pi_U) \cdot \overline{\Tr(\Lambda^2 \Pi_{(\Lambda U)^\perp})} =
    \alpha (1-\beta),
    \]
    where $\alpha := \Tr(\Lambda^2 \Pi_U)$ and $\beta := \Tr(\Lambda^2 \Pi_{\Lambda U})$.
    By \Cref{lem:ze-alpha-beta}, we have $0\leq \alpha \leq \beta\leq 1$, and thus
    \begin{equation} \label{eq:ze-chain}
    m_\Lambda(P) \leq
    m_\Lambda(\Pi_{\mathrm{L}((\Lambda U)^\perp, U)}) =
    \alpha (1-\beta) \leq \alpha (1-\alpha) \leq \frac14. \qedhere
    \end{equation}
\end{proof}

Finally, we prove \Cref{thm:ze-mm} by analyzing the equality condition forced by \eqref{eq:ze-chain}.

\begin{proof}[Proof of \Cref{thm:ze-mm}]
    By \Cref{lem:ze-mm-upper}, we have $m_\Lambda(P) \leq 1/4$. Since
    $m_\Lambda(P) = 1/4$, all inequalities in \eqref{eq:ze-chain} must be saturated. In particular, we have
    \[
    m_\Lambda(P) = m_\Lambda(\Pi_{\mathrm{L}((\Lambda U)^\perp, U)}), \quad
    \beta = \alpha, \quad\text{and}\quad
    \alpha = \frac12.
    \]

    We first prove \ref{P2}. By $\alpha=\frac12$, we have
    \[
    \Tr(\Lambda^2 \Pi_U) = \alpha = \frac12,
    \quad\text{and}\quad
    \Tr(\Lambda^2 \Pi_{U^\perp}) = 
    \Tr(\Lambda^2) - \Tr(\Lambda^2 \Pi_U) = 1 - \frac12 = \frac12.
    \]

    Next, we prove \ref{P3}. 
    Define $X = V^\dagger \Lambda^2 V$ and $Y = V^\dagger \Lambda^4 V$
    as in the proof of \Cref{lem:ze-alpha-beta}. By 
    $\alpha=\beta$
    and \eqref{eq:alpha-beta}, we have
    \(
    0 = \beta - \alpha = \Tr((Y-X^2)X^{-1}), 
    \)
    which implies $Y-X^2=0$, as $Y-X^2\succeq 0$ by \eqref{eq:y-x2} and $X^{-1}\succ 0$. 
    Combining with \eqref{eq:y-x2}, we have
    \[
    0 = Y-X^2 = V^\dagger \Lambda^2 \Pi_{U^\perp} \Lambda^2 V =
        \left(\Pi_{U^\perp}\Lambda^2V\right)^\dagger
        \left(\Pi_{U^\perp}\Lambda^2V\right).
    \]
    
    Thus $\Pi_{U^\perp}\Lambda^2V = 0$. As $U$ is spanned by columns of $V$, we
    have 
    \[
    \Lambda^2 U \subseteq U,
    \] i.e., $U$ is invariant under $\Lambda^2$.
    Since $\Lambda$ is Hermitian, any invariant subspace of $\Lambda^2$ is spanned
    by eigenvectors of $\Lambda^2$. 
    Since $\Lambda$ and $\Lambda^2$ share the same eigenvectors, 
    we also have $\Lambda U\subseteq U$. Moreover, 
    because $\Lambda$ is invertible, we have
    \[
        \Lambda U=U.
    \]
    
    Then for any vector $u\in U, u^\perp\in U^\perp$, 
    we have $\braket{u^\perp, \Lambda u}=0$ since $\Lambda u \in \Lambda U=U$. 
    Since $\Lambda$ is Hermitian,
    we have
    \(
    \braket{\Lambda u^\perp, u} = 
    \braket{u^\perp, \Lambda u} = 0
    \).
    Thus $U^\perp$ is also $\Lambda$-invariant, i.e., 
    \[\Lambda U^\perp=U^\perp.\]

    Finally, we prove \ref{P1}. By \Cref{lem:ze-structure}, we know that
    \[
    P \preceq \Pi_{\mathrm{L}((\Lambda U)^\perp, U)}.
    \]
    We claim that $P = \Pi_{\mathrm{L}((\Lambda U)^\perp, U)}$. Indeed, 
    if $P \ne \Pi_{\mathrm{L}((\Lambda U)^\perp, U)}$, then
    $(\Pi_{\mathrm{L}((\Lambda U)^\perp, U)}-P)$ is a nonzero PSD.
    By $\Lambda^2\otimes \bar{\Lambda}^2\succ 0$, we would have
    \[
    m_\Lambda(\Pi_{\mathrm{L}((\Lambda U)^\perp, U)}) - m_\Lambda(P) = 
    \Tr\left[(\Pi_{\mathrm{L}((\Lambda U)^\perp, U)} - P)(\Lambda^2\otimes \bar{\Lambda}^2)\right] >0,
    \]
    contradicting the fact that $m_\Lambda(P) = m_\Lambda(\Pi_{\mathrm{L}((\Lambda U)^\perp, U)})$.
    Furthermore, by $\Lambda U=U$ from \ref{P3}, we have
    \[
    P = \Pi_{\mathrm{L}((\Lambda U)^\perp, U)}
    = \Pi_{\mathrm{L}(U^\perp, U)}.
    \]
    This proves all three claimed properties.
\end{proof}

\subsection{Approximation of Zero-Energy Operators}

In this section,
we show that if an operator of maximal mass has small energy, then it can be approximated by 
an operator with exactly zero energy,
as stated in \Cref{thm:weighted-stability}.
We remark that this theorem does not require $\Lambda$ to be strictly positive.

\begin{theorem} \label{thm:weighted-stability}
Assume $N\geq 2$.
Given PSD $P\in \mathrm{L}(\mathbb{C}^N\otimes \mathbb{C}^N)$ with $0\preceq
P\preceq I$, and PSD $\Lambda\in\mathbb{C}^{N\times N}$ with $\|\Lambda\|_F=1$, if
\[
\mathcal{E}_\Lambda(P) = \epsilon \geq 0, 
\quad\text{and}\quad
m_{\Lambda}(P)=\frac14,
\]
then there exists PSD $P'$ with $0\preceq P'\preceq I$ and PSD $\Lambda'$ with 
$\|\Lambda'\|_F=1$ such that
\[
\mathcal{E}_{\Lambda'}(P') = 0,\quad 
m_{\Lambda'}(P')=\frac14,\quad 
\|\Lambda-\Lambda'\|_F^2 = O(\epsilon^{1/4}),
\quad\text{and}\quad
\|P-P'\|_\Lambda = O(\epsilon^{1/4}).
\]
\end{theorem}

\subsubsection{Candidate Construction} \label{sec:ze-candidate}

Notice that \Cref{thm:weighted-stability} is trivial if $\Es_\Lambda(P)=0$.
In the following, we assume $\Es_\Lambda(P)=\epsilon>0$. 
We construct $P'$ and $\Lambda'$ as follows.
First recall that
\[
\Es_\Lambda(P) = \Tr(\Lambda A \Lambda B),
\quad\text{where}\quad
A := \Tr_\As\left((\Lambda^2\otimes I)P\right)^\top,\ 
B := \Tr_\Bs\left((I\otimes \bar{\Lambda}^2)P\right).
\]
As $0\preceq B=\Tr_\Bs((I\otimes \bar{\Lambda})P(I\otimes \bar{\Lambda}))\preceq 
\Tr_\Bs(I\otimes \bar{\Lambda}^2)=I$, we
can write its spectral decomposition as
\[
B=\sum_i \rho_i v_i v_i^\dagger
\]
where $\{v_i\}$ are orthonormal eigenvectors, and 
$1 \geq \rho_1 \geq \rho_2 \geq \cdots \geq \rho_N \geq 0$ are 
eigenvalues of $B$ in non-increasing order.
Define a threshold index
\[
t := \max \{i\in[N]: \rho_i\geq \sqrt{\epsilon}\}.
\]
Then set up the orthogonal decomposition $\mathbb{C}^N=U\oplus U^\perp$ where
$$
U=\mathrm{span}\{v_i:i\leq t\}
\quad\text{and}\quad
U^\perp=\mathrm{span}\{v_i:i > t\}.
$$
Then the candidate construction is
$$
P':= \Pi_{\mathrm{L}(U^\perp,U)}
\quad\text{and}\quad
\Lambda' := \frac{X}{\sqrt{2\mathrm{Tr}(X^2)}} + \frac{Y}{\sqrt{2\mathrm{Tr}(Y^2)}},
$$
where 
\(
X:=\Pi_U\Lambda\Pi_U
\)
and
\(
Y:=\Pi_{U^\perp}\Lambda\Pi_{U^\perp}
\).

We claim that for small enough $\epsilon$,
the above construction is well-defined, i.e.,
$t$ exists and $X,Y$ are both nonzero;
see Proof of \Cref{thm:weighted-stability} for details.
Then it is immediate that $P'$ is an orthogonal projector, and $\Lambda'$ is PSD with
$\|\Lambda'\|_F=1$. We can also verify that the pair $(P',\Lambda')$ achieves
the maximal mass and zero energy, stated as follows.

\begin{fact} \label{fact:ze-candidate}
$m_{\Lambda'}(P') = \frac14$ and $\mathcal{E}_{\Lambda'}(P')=0$.
\end{fact}

\begin{proof}
    As $\Lambda'$ is block diagonal with respect
    to $U\oplus U^\perp$, 
    spaces $U, U^\perp$ are invariant under $\Lambda'$,
    and thus 
    \[
    \Tr(\Lambda'^2 \Pi_U) = 
    \frac{\Tr(X^2)}{2\Tr(X^2)} + 0 = \frac12
    \quad\text{and}\quad
    \Tr(\Lambda'^2 \Pi_{U^\perp}) = 1 - \Tr(\Lambda'^2 \Pi_U) = \frac12.
    \]
    Then by the fact that $P' = \Pi_{\mathrm{L}(U^\perp,U)} = \Pi_{U}\otimes \bar{\Pi}_{U^\perp}$, we have
    \[
        m_{\Lambda'}(P') = \Tr\left[P'(\Lambda'^2\otimes \bar{\Lambda}'^2)\right] = 
        \Tr(\Pi_U \Lambda'^2) \overline{\Tr(\Pi_{U^\perp} \Lambda'^2)}
    = \frac14.
    \]

Note that the spectral decomposition of $P'$ is 
\[
P' = \sum_{k} \V(M_k)\V(M_k)^\dagger,
\]
where $\{M_k\} := \{v_i v_j^\dagger: i\leq t, j > t\}$.
For any two $M_k = v_iv_j^\dagger, M_\ell = v_p v_q^\dagger$, 
we have $v_j \in U^\perp$ and $v_p \in U$.
As $U$ is invariant under $\Lambda'$, we have $\Lambda' v_p \in U$, and thus
$v_j^\dagger \Lambda' v_p = 0$. Therefore, we have
\[
\Lambda' M_k \Lambda' M_\ell \Lambda' = 
\Lambda' v_i (v_j^\dagger \Lambda' v_p) v_q^\dagger \Lambda' = 0.
\]
By \Cref{fact:energy}, we have $\mathcal{E}_{\Lambda'}(P') 
= \sum_{k,\ell} \|\Lambda' M_k \Lambda' M_\ell \Lambda'\|_F^2 = 0$.
\end{proof}

Therefore, to conclude \Cref{thm:weighted-stability} using $(\Lambda', P')$, it
remains to show that $(\Lambda', P')$ is a good approximation of $(\Lambda, P)$.

\subsubsection{Approximate \texorpdfstring{$\Lambda$}{Lambda}-Invariance of \texorpdfstring{$U$}{U}}
\label{sec:approx-invariance}

We first prove that the subspace $U$ is approximately invariant under $\Lambda$,
when $\mathcal{E}_\Lambda(P)=\epsilon$ is small. Formally,

\begin{lemma} \label{lem:delta-bound}
    Let $Q:=\Pi_U$ and $R:=\Pi_{U^\perp}$. 
    Define the non-invariance measure of $U$ under $\Lambda$ as
    \[
    \delta := \|Q \Lambda R\|_F^2.
    \]
    Then we have $\delta = O(\sqrt{\epsilon})$.
\end{lemma}

Note that $\delta=0$ if and only if $U$ is invariant under $\Lambda$, e.g., when
$\Lambda = I/\sqrt{N}$. To prove the lemma, the key observation is to relate
$\delta$ to the testing error of the pretty-good measurement (PGM, see \Cref{def:pgm}) for two subnormalized states
$\sigma_0=\Lambda R \Lambda$ and $\sigma_1=\Lambda Q \Lambda$. Then we construct
a measurement that distinguishes $\sigma_0$ and $\sigma_1$ with small error,
which implies that the optimal testing error is small, and thus the PGM testing
error is also small.

\begin{proof}[Proof of \Cref{lem:delta-bound}]
    Consider two states
\[
\sigma_0 := \Lambda R \Lambda, \quad
\sigma_1 := \Lambda Q \Lambda,
\]
and define $\sigma=\sigma_0 + \sigma_1 = \Lambda^2$. 
By \Cref{def:pgm}, the PGM for $\sigma_0, \sigma_1$ is
\[
M_0 := \Lambda^{+} \sigma_0 \Lambda^{+} + \frac12 \Pi_{\ker(\sigma)},
\quad
M_1 := \Lambda^{+} \sigma_1 \Lambda^{+} + \frac12 \Pi_{\ker(\sigma)}.
\]

The testing error of $M_0$ is 
\[
\Tr(M_0\sigma_1)
= \Tr(\Lambda^{+} \sigma_0 \Lambda^{+} \sigma_1) + 
\frac12 \Tr(\Pi_{\ker(\sigma)} \sigma_1)
= \Tr(\Lambda^{+} \Lambda R \Lambda \Lambda^{+} \Lambda Q \Lambda) + 0
= \Tr(R \Lambda Q \Lambda) = \delta,
\]
where the first equality is by the definition of $M_0$, the second equality follows from
$\sigma_1 \preceq \sigma$ 
and thus $\ker(\sigma) \subseteq \ker(\sigma_1)$, the third equality is by 
$\Lambda^{+}\Lambda\Lambda = \Lambda$, and the last equality
is by $\delta = \|Q\Lambda R\|_F^2 = \Tr(R\Lambda Q Q \Lambda R) = \Tr(R\Lambda Q\Lambda)$.

Similarly, the testing error of $M_1$ is also
$\Tr(M_1\sigma_0) = \delta$. Thus the testing error of PGM is
\[
\mathrm{PGM}(\sigma_0, \sigma_1)
:= \Tr(M_0\sigma_1) + \Tr(M_1\sigma_0) = 2\delta.
\]

Now we construct a $2$-outcome POVM $\{V, I-V\}$ that distinguishes $\sigma_0$
and $\sigma_1$. We write the spectral decomposition of $P$ as
\[
P = \sum_k \V(M_k)\V(M_k)^\dagger,
\]
and define a completely positive map and its adjoint map
\[
\Phi(X) := \sum_k M_k X M_k^\dagger, \quad
\Phi^*(X):= \sum_k M_k^\dagger X M_k.
\]

If $\Tr(Q\Lambda^2) = 0$, we have $0\leq \delta = \Tr(R\Lambda Q \Lambda) \leq
\Tr(Q\Lambda^2) = 0$, and thus $\delta = 0 = O(\epsilon^{1/2})$.
Assume $\Tr(Q\Lambda^2) > 0$. Define a quantum state
\[
\sigma:=\frac{Q\Lambda^2Q}{\mathrm{Tr}(Q\Lambda^2)},
\]
and feed $\sigma$ through $\Phi^*$ to get 
\[
V:=\Phi^*(\sigma) 
=\frac{1}{\mathrm{Tr}(Q\Lambda^2)} \sum_{k} M_k^\dagger Q \Lambda^2 Q M_k.
\]
This indeed defines a valid POVM operator. For every unit vector $x\in\mathbb{C}^N$,
\Cref{lem:tr-vy} gives
\[
0\leq \langle x,Vx\rangle
=\frac{1}{\Tr(Q\Lambda^2)}
\Tr\!\left(P\bigl(Q\Lambda^2Q\otimes\overline{xx^\dagger}\bigr)\right)
\leq\frac{\Tr(Q\Lambda^2Q)\Tr(xx^\dagger)}{\Tr(Q\Lambda^2)}
=\|x\|^2=1.
\]
Thus $0\preceq V\preceq I$, so $\{V,I-V\}$ is a two-outcome POVM.
By \Cref{lem:povm}, the optimal testing error
\[
\mathrm{OPT}(\sigma_0, \sigma_1) \leq 
\Tr(V\Lambda Q \Lambda) 
+ \Tr((I-V)\Lambda R \Lambda) = O(\epsilon^{1/2}).
\]
By \Cref{lem:pgm}, we have
\[
\mathrm{PGM}(\sigma_0, \sigma_1) \leq 2 \mathrm{OPT}(\sigma_0, \sigma_1) = O(\epsilon^{1/2}),
\]
which implies $\delta = \frac12 \mathrm{PGM}(\sigma_0, \sigma_1) = O(\epsilon^{1/2})$.
\end{proof}

\begin{lemma} \label{lem:povm}
    Let $\{V, I-V\}$ be the POVM defined above. Then
    \[
    \Tr(V\Lambda Q \Lambda) + \Tr((I-V)\Lambda R \Lambda) = O(\epsilon^{1/2}).
    \]
\end{lemma}

\begin{proof}
    Combine \Cref{lem:v-lql} and \Cref{lem:v-lrl}.
\end{proof}

We first prove two helper lemmas. 
For convenience, we define $\alpha := \Tr(Q\Lambda^2)$.
The following helper lemma characterizes the action of $V$.

\begin{lemma} \label{lem:tr-vy}
For any PSD $Y$, we have
$$
\mathrm{Tr}(VY)= \frac{1}{\alpha}\mathrm{Tr}(P(Q\Lambda^2 Q\otimes \bar{Y})).
$$
\end{lemma}

\begin{proof} By the definition of $V$, 
$$
\mathrm{Tr}(VY) = \frac{1}{\alpha} \sum_k \mathrm{Tr}(Y M_k^\dagger Q\Lambda^2 Q M_k).
$$
Compute
\begin{align*}
M_k^\dagger Q\Lambda^2 Q M_k
=(\Lambda Q M_k)^\dagger (\Lambda QM_k)
&= \mathrm{Tr}_\As\!\left(\V(\Lambda Q M_k)\V(\Lambda Q M_k)^\dagger\right)^\top \\
&= \mathrm{Tr}_\As\!\left(\V(M_k)\V(M_k)^\dagger(Q\Lambda^2 Q\otimes I)\right)^\top.
\end{align*}
Combining the two equations gives
\begin{align*}
\mathrm{Tr}(VY) &= \frac{1}{\alpha} \mathrm{Tr}\left(Y\cdot \sum_k \mathrm{Tr}_\As(\V(M_k)\V(M_k)^\dagger(Q\Lambda^2 Q\otimes I))^\top
\right)\\
&= \frac{1}{\alpha}\mathrm{Tr}(Y\cdot \mathrm{Tr}_\As(P(Q\Lambda^2 Q\otimes I))^\top) \\
&= \frac{1}{\alpha}\mathrm{Tr}(\mathrm{Tr}_\As(P(Q\Lambda^2 Q\otimes \bar{Y})))
= \frac{1}{\alpha}\mathrm{Tr}(P(Q\Lambda^2 Q\otimes \bar{Y})).  \tag*{\qedhere}
\end{align*}
\end{proof}

The next lemma bounds the normalization factor $\alpha$.

\begin{lemma} \label{lem:alpha}
We have 
$$
\alpha \geq \mathrm{Tr}(BQ\Lambda^2 Q)  \geq \frac14-\sqrt{\epsilon}.
$$
If $\epsilon \leq 1/64$, we have $\alpha \geq 1/8$.
\end{lemma}

\begin{proof} By construction, we know that
$$
\mathrm{Tr}(B\Lambda^2) = 
\mathrm{Tr}(\Lambda^2 B) = 
\Tr(\Lambda^2\Tr_\Bs((I \otimes \bar{\Lambda}^2)P)) = 
\Tr((\Lambda^2 \otimes \bar{\Lambda}^2)P) = 
m_\Lambda(P) = 
\frac14.
$$
Since $B$ is block diagonal under $Q\oplus R$, we have
$$
\mathrm{Tr}(B\Lambda^2) =  \mathrm{Tr}(QB Q \Lambda^2) + \mathrm{Tr}(RBR\Lambda^2) = \mathrm{Tr}(B Q \Lambda^2 Q) + \mathrm{Tr}(BR\Lambda^2 R).
$$
By $B \preceq I$, we have 
$$
\mathrm{Tr}(BQ\Lambda^2 Q) \leq \mathrm{Tr}(Q\Lambda^2 Q) = \Tr(Q\Lambda^2) = \alpha.
$$
For the second term, 
as $\rho_i \leq \sqrt{\epsilon}$ for any $i > t$, 
we have $RBR \preceq \sqrt{\epsilon} I$. Thus
$$
\mathrm{Tr}(BR\Lambda^2 R)
= \mathrm{Tr}(RBR\Lambda^2)
\leq \sqrt{\epsilon} \mathrm{Tr}(\Lambda^2) = \sqrt{\epsilon}.
$$
Putting things together,
$$
\frac{1}{4} \leq \mathrm{Tr}(B Q \Lambda^2 Q) + \mathrm{Tr}(BR\Lambda^2 R)
\leq \alpha + \sqrt{\epsilon}
\quad\Rightarrow\quad
\alpha \geq \mathrm{Tr}(BQ\Lambda^2Q)\geq \frac{1}{4} - \sqrt{\epsilon}.
$$
For $\epsilon \leq 1/64$, we have $\alpha \geq 1/4 - 1/8 = 1/8$.
\end{proof}

Then we bound the two error terms of the POVM $\{V, I-V\}$.
The first lemma shows that $V$ accepts $\sigma_1=\Lambda Q \Lambda$ with low probability.

\begin{lemma}\label{lem:v-lql}
    $\Tr(V\Lambda Q \Lambda) = O(\epsilon^{1/2})$.
\end{lemma}

\begin{proof}
If $\epsilon> 1/64=\Omega(1)$, then the bound is trivial since 
$\Tr(V\Lambda Q \Lambda)$ is a probability and thus at most $1$.
Now assume $\epsilon \leq 1/64$.
By \Cref{lem:tr-vy} and \Cref{lem:alpha}, we have
$$
\mathrm{Tr}(V \Lambda Q \Lambda) =
\frac{1}{\alpha} \mathrm{Tr}(P (Q\Lambda^2 Q\otimes \overline{\Lambda Q\Lambda}))
\leq 8\cdot \mathrm{Tr}(P (Q\Lambda^2 Q\otimes \overline{\Lambda Q\Lambda})).
$$

Since $Q=I-R$, the inequality
$\|u-v\|^2 \leq 2\|u\|^2 + 2\|v\|^2$ implies that, for every
$x\in\mathbb{C}^N$,
\begin{align*}
\braket{x, Q\Lambda^2Qx}
= \|\Lambda Qx\|^2 
&= \|\Lambda x-\Lambda Rx\|^2  \\
&\leq 2\|\Lambda x\|^2 + 2\|\Lambda Rx\|^2 
= 2\braket{x, \Lambda^2x}
   + 2\braket{x, R\Lambda^2Rx}.
\end{align*}
Hence
\[
Q\Lambda^2Q \preceq 2\Lambda^2 + 2R\Lambda^2R.
\]
Consequently, we obtain
$$
\mathrm{Tr}(P (Q\Lambda^2 Q\otimes \overline{\Lambda Q\Lambda})) \leq 
2\cdot\mathrm{Tr}(P (\Lambda^2\otimes \overline{\Lambda Q\Lambda})) + 
2\cdot\mathrm{Tr}(P (R\Lambda^2 R\otimes \overline{\Lambda Q\Lambda})).
$$
For the first term, using $\sqrt{\epsilon} Q \preceq B$,
$$
\mathrm{Tr}(P (\Lambda^2\otimes \overline{\Lambda Q\Lambda})) = \mathrm{Tr}(A^\top \overline{\Lambda Q\Lambda})
= \mathrm{Tr}(A\Lambda Q\Lambda) \leq \frac{1}{\sqrt{\epsilon}} \mathrm{Tr}(A\Lambda B\Lambda) = 
\frac{\mathcal{E}_\Lambda(P)}{\sqrt{\epsilon}}=\sqrt{\epsilon}.
$$
For the second term, using 
$\overline{\Lambda Q\Lambda} \preceq \bar{\Lambda}^2$ and
$RBR \preceq \sqrt{\epsilon} I$, 
$$
\mathrm{Tr}(P (R\Lambda^2 R\otimes \overline{\Lambda Q\Lambda})) \leq \mathrm{Tr}(P (R\Lambda^2 R\otimes \bar{\Lambda}^2)) 
= \Tr[\Tr_\Bs(P(I\otimes \bar{\Lambda}^2)) R\Lambda^2 R]
= \mathrm{Tr}(BR\Lambda^2R) \leq \sqrt{\epsilon} \mathrm{Tr}(\Lambda^2)=\sqrt{\epsilon}.
$$
Putting things together, we have 
\begin{equation} \label{eq:A}
\mathrm{Tr}(P(Q\Lambda^2Q\otimes  \overline{\Lambda Q\Lambda})) \leq 2\sqrt{\epsilon} + 2\sqrt{\epsilon} = 4\sqrt{\epsilon}.
\end{equation}
Thus $\mathrm{Tr}(V \Lambda Q\Lambda) \leq 32\sqrt{\epsilon} = O(\epsilon^{1/2})$.
\end{proof}

The next lemma shows that $I-V$ accepts $\sigma_0=\Lambda R \Lambda$ with low probability.

\begin{lemma}\label{lem:v-lrl}
    $\Tr((I-V)\Lambda R \Lambda) = O(\epsilon^{1/2})$.
\end{lemma}

\begin{proof}
If $\epsilon> 1/64=\Omega(1)$, then the bound is trivial 
since $\Tr((I-V)\Lambda R \Lambda)$ is a probability and thus at most $1$.
Now assume $\epsilon \leq 1/64$.
By \Cref{lem:tr-vy} and \Cref{lem:alpha}, we have
$$
\mathrm{Tr}((I-V)\Lambda R\Lambda) = 
\mathrm{Tr}(\Lambda R\Lambda) - 
\mathrm{Tr}(V\Lambda R\Lambda)
= (1-\alpha) - 
\frac{1}{\alpha} \mathrm{Tr}(P(Q\Lambda^2 Q\otimes \overline{\Lambda R\Lambda}))
$$
By $\overline{\Lambda R\Lambda}=\bar{\Lambda}^2-\overline{\Lambda Q\Lambda}$, we have
$$
\mathrm{Tr}(P(Q\Lambda^2 Q\otimes \overline{\Lambda R\Lambda}))=
\mathrm{Tr}(P(Q\Lambda^2Q\otimes \bar{\Lambda}^2)) - \mathrm{Tr}(P(Q\Lambda^2Q\otimes  \overline{\Lambda Q\Lambda})).
$$
The first term
$$
\mathrm{Tr}(P(Q\Lambda^2Q\otimes \bar{\Lambda}^2)) = \mathrm{Tr}(BQ\Lambda^2Q)\geq \frac14 - \sqrt{\epsilon},
$$
where the last inequality uses \Cref{lem:alpha}.

By \eqref{eq:A}, the second term 
$$
\mathrm{Tr}(P(Q\Lambda^2Q\otimes  \overline{\Lambda Q\Lambda})) \leq 4\sqrt{\epsilon}.
$$
Thus 
$$
\mathrm{Tr}(P(Q\Lambda^2 Q\otimes \overline{\Lambda R\Lambda})) \geq \frac14 - 5\sqrt{\epsilon}.
$$
Plugging back gives
$$
\mathrm{Tr}((I-V)\Lambda R\Lambda) \leq \frac1{\alpha}\left(\alpha(1-\alpha) - \frac14+5\sqrt{\epsilon}\right) \leq \frac{5\sqrt{\epsilon}}{\alpha} \leq 40\sqrt{\epsilon},
$$
where the last inequality uses \Cref{lem:alpha}.
\end{proof}

\subsubsection{Proof of \texorpdfstring{\Cref{thm:weighted-stability}}{the Weighted Stability Theorem}} \label{sec:approx-proof}

We first introduce the weighted operator
\[
\tilde{P} := (\Lambda\otimes \bar{\Lambda})P(\Lambda\otimes \bar{\Lambda}).
\]
Note that its trace equals the weighted mass,
\[
\Tr(\tilde{P}) = \Tr((\Lambda^2\otimes \bar{\Lambda}^2)P) = m_\Lambda(P) = \frac14,
\]
and its partial traces are weighted marginals,
\[
\Tr_\As(\tilde{P}) = 
\bar{\Lambda}
\Tr_\As\!\left((\Lambda \otimes I)
P
(\Lambda \otimes I)
\right)
\bar{\Lambda}
= (\Lambda A \Lambda)^\top,
\quad\text{and similarly,}\quad
\Tr_\Bs(\tilde{P}) = 
\Lambda B \Lambda.
\]

We then prove two helper lemmas.
The first lemma bounds the leakage of $\tilde{P}$ outside the space $\mathrm{L}(U^\perp, U)$ specified by the candidate projector $P'$.

\begin{lemma} \label{lem:upper-right}
$
\mathrm{Tr}(\tilde{P}(I-P')) = O(\sqrt{\epsilon}).
$
\end{lemma}

\begin{proof}

For $i, j\in [N]$, we define the variable
\[
y_{ij} := (v_i\otimes \bar{v}_j)^\dagger \tilde{P} (v_i\otimes \bar{v}_j) \in [0,1].
\]
Compute its $j$-th column sum
\[
c_j := \sum_{i} y_{ij} 
= \sum_i (v_i\otimes \bar{v}_j)^\dagger \tilde{P} (v_i\otimes \bar{v}_j)
= \bar{v}_j^\dagger \Tr_\As(\tilde{P}) \bar{v}_j
= \overline{v_j^\dagger \Lambda A \Lambda v_j}
= v_j^\dagger \Lambda A \Lambda v_j,
\]
Then we can express energy as
$$
\mathcal{E}_\Lambda(P) = \mathrm{Tr}\left(\Lambda A\Lambda \sum_j \rho_jv_jv_j^\dagger\right)=\sum_j \rho_j\mathrm{Tr}\left(\Lambda A\Lambda v_jv_j^\dagger\right)=\sum_j \rho_j c_j.
$$

Similarly, define the $i$-th row sum
$$
r_i := \sum_j y_{ij} 
= \sum_j (v_i\otimes \bar{v}_j)^\dagger \tilde{P} (v_i\otimes \bar{v}_j)
= v_i^\dagger \Tr_\Bs(\tilde{P}) v_i
= v_i^\dagger \Lambda B\Lambda v_i.
$$

By $P'=\Pi_{\mathrm{L}(U^\perp,U)} = \sum_{i\leq t, j > t} (v_i\otimes \bar{v}_j)(v_i\otimes \bar{v}_j)^\dagger$, we have
\[
\mathrm{Tr}(\tilde{P}(I-P'))=
\mathrm{Tr}\left(
    \tilde{P}
    \left(I-\sum_{i\leq t, j > t} (v_i\otimes \bar{v}_j)(v_i\otimes \bar{v}_j)^\dagger\right)
\right)
= \sum_{i > t\text{ or }j\leq t}
(v_i\otimes \bar{v}_j)^\dagger\tilde{P}(v_i\otimes \bar{v}_j)
= \sum_{i > t\text{ or }j\leq t} y_{ij}.
\]
Then observe that 
$$
\sum_{i > t\text{ or }j\leq t} y_{ij} \leq \sum_{j\leq t}c_j + \sum_{i>t} r_i.
$$

For the first sum, as $\rho_j \geq \sqrt{\epsilon}$ for all $j\leq t$, we have
$$
   \mathcal{E}_\Lambda(P)=\sum_{j} \rho_j c_j \geq \sum_{j\leq t} \rho_j c_j \geq \sqrt{\epsilon} \sum_{j\leq t} c_j.
$$
Thus
\begin{equation} \label{eq:col-sum}
   \sum_{j\leq t} c_j \leq \frac{\mathcal{E}_\Lambda(P)}{\sqrt{\epsilon}} = \sqrt{\epsilon}.
\end{equation}

For the second sum, 
\[
   \sum_{i> t} r_i 
    = \sum_{i> t} v_i^\dagger \Lambda B\Lambda v_i
    = \Tr\left(\Pi_{U^\perp} \Lambda B \Lambda \Pi_{U^\perp}\right).
\]
As $B\preceq I$ and $\rho_j \leq \sqrt{\epsilon}$ for all $j>t$, we have
$B \preceq \Pi_U + \sqrt{\epsilon} \Pi_{U^\perp}$. Then
\begin{align}
\notag\Tr\left(\Pi_{U^\perp} \Lambda B \Lambda \Pi_{U^\perp}\right) &\leq
\Tr\left(\Pi_{U^\perp} \Lambda \Pi_U \Lambda \Pi_{U^\perp}\right) + 
\sqrt{\epsilon} \Tr\left(\Pi_{U^\perp} \Lambda \Pi_{U^\perp }\Lambda \Pi_{U^\perp}\right) \\
\notag&= \delta + \sqrt{\epsilon} \Tr\left(\Pi_{U^\perp} \Lambda \Pi_{U^\perp }\Lambda \Pi_{U^\perp}\right) \\
\label{eq:row-sum}&\leq \delta + \sqrt{\epsilon} \mathrm{Tr}(\Lambda^2) 
\leq \delta + \sqrt{\epsilon}.
\end{align}

Combining \eqref{eq:col-sum} and \eqref{eq:row-sum}, we get
\[
\sum_{i > t\text{ or }j\leq t} y_{ij} \leq 
\sum_{j\leq t}c_j + \sum_{i>t} r_i
\leq 2\sqrt{\epsilon} + \delta = O(\sqrt{\epsilon}),
\]
where the last inequality is by \Cref{lem:delta-bound}.
\end{proof}

The second lemma shows that
the candidate split
$U\oplus U^\perp$ has nearly balanced weights.

\begin{lemma} \label{lem:tr-lambda-pi}
    Define $\eta := \mathrm{Tr}(\Lambda^2 \Pi_U)$. Then
\[
\left|\eta - \frac12\right| 
= O(\epsilon^{1/4}).
\]
\end{lemma}

\begin{proof}
By \Cref{lem:upper-right} and $\mathrm{Tr}(\tilde{P}P')\leq \mathrm{Tr}(\tilde{P})=1/4$, we have
$$
\frac{1}{4} - 
O(\sqrt{\epsilon})
\leq \mathrm{Tr}(\tilde{P}P') \leq \frac14.
$$
Then compute
\begin{align*}
\mathrm{Tr}(\tilde{P}P') 
&= \mathrm{Tr}((\Lambda\otimes\bar{\Lambda}) P (\Lambda\otimes \bar{\Lambda}) P') \\
&\leq \mathrm{Tr}((\Lambda^2\otimes \bar{\Lambda}^2) P') \\
&= \mathrm{Tr}((\Lambda^2\otimes\bar{\Lambda}^2)(\Pi_U\otimes \bar{\Pi}_{U^\perp})) 
= \mathrm{Tr}(\Lambda^2 \Pi_U) \overline{\mathrm{Tr}(\Lambda^2\Pi_{U^\perp})}  = \eta (1-\eta).
\end{align*}
Plugging back, we have
$$
\frac14 - O(\sqrt{\epsilon}) \leq \eta(1-\eta) = \frac14 - \left(\eta-\frac12\right)^2
\quad\Longrightarrow\quad
\left|\eta-\frac12\right| = O(\epsilon^{1/4}).  \qedhere
$$
\end{proof}

Next, using the helper lemmas, we bound the closeness of $\Lambda$ and $\Lambda'$.
\begin{lemma} \label{lem:lambda-close}
$
\|\Lambda-\Lambda'\|_F^2
= O(\epsilon^{1/4}).
$
\end{lemma}

\begin{proof}
Define 
$
Q:=\Pi_U
$
and
$
R:=\Pi_{U^\perp}
$.
Then
\[
\Lambda
=
Q\Lambda Q+R\Lambda R+Q\Lambda R+R\Lambda Q
=
X+Y+Q\Lambda R+R\Lambda Q.
\]
First compute the Frobenius masses of \(X\) and \(Y\). Since
\[
Q\Lambda^2Q
=
Q\Lambda(Q+R)\Lambda Q
=
Q\Lambda Q\Lambda Q+Q\Lambda R\Lambda Q,
\]
we have
\[
\eta
:=
\operatorname{Tr}(\Lambda^2 Q)
=
\operatorname{Tr}(Q\Lambda^2Q)
=
\operatorname{Tr}(X^2)+\operatorname{Tr}(Q\Lambda R\Lambda Q).
\]
By 
\(
\operatorname{Tr}(Q\Lambda R\Lambda Q) =
\|Q\Lambda R\|_F^2 =
\delta,
\)
we have 
\(
\operatorname{Tr}(X^2)=\eta-\delta.
\)
Similarly,
\(
\operatorname{Tr}(Y^2)=1-\eta-\delta.
\)
Let
\[
x:=\operatorname{Tr}(X^2)=\eta-\delta,
\qquad
y:=\operatorname{Tr}(Y^2)=1-\eta-\delta.
\]
Note that \(\{Q\Lambda Q, R\Lambda R, Q\Lambda R, R\Lambda Q\}\) are pairwise orthogonal in Frobenius inner product. Thus
\begin{align*}
\|\Lambda-\Lambda'\|_F^2
&=
\left\|
\left(1-\frac{1}{\sqrt{2x}}\right)X
+
\left(1-\frac{1}{\sqrt{2y}}\right)Y
+
Q\Lambda R+R\Lambda Q
\right\|_F^2 \\
&=
\left(1-\frac{1}{\sqrt{2x}}\right)^2\|X\|_F^2
+
\left(1-\frac{1}{\sqrt{2y}}\right)^2\|Y\|_F^2
+
\|Q\Lambda R\|_F^2 + \|R\Lambda Q\|_F^2 \\
&=
\left(\sqrt{x}-\frac{1}{\sqrt2}\right)^2
+
\left(\sqrt{y}-\frac{1}{\sqrt2}\right)^2
+
2\delta.
\end{align*}
Using
\(
(\sqrt{a}-\sqrt{b})^2\le |a-b|
\) for any $a,b\geq 0$, 
we get
\begin{align*}
\|\Lambda-\Lambda'\|_F^2
&\le
\left|x-\frac12\right|
+
\left|y-\frac12\right|
+
2\delta \\
&=
\left|\eta-\delta-\frac12\right|
+
\left|1-\eta-\delta-\frac12\right|
+
2\delta \\
&= \left|\left(\eta-\frac12\right)-\delta\right|
+
\left|\left(\eta-\frac12\right)+\delta\right|
+
2\delta 
\leq 2 \left|\eta-\frac12\right| + 4\delta.
\end{align*}
By \Cref{lem:tr-lambda-pi} and \Cref{lem:delta-bound}, we have
\[
\|\Lambda-\Lambda'\|_F^2
\le
2 O(\epsilon^{1/4}) 
+ 4 O(\sqrt{\epsilon}) = O(\epsilon^{1/4}), \qedhere
\]
\end{proof}

Finally, we bound the closeness of $P$ and $P'$.

\begin{lemma} \label{lem:p-close}
    $
    \|P-P'\|_\Lambda
    \leq
    O(\epsilon^{1/4}).
    $
\end{lemma}

\begin{proof}
Define 
$$
\tilde{P}:=(\Lambda\otimes \bar{\Lambda})P(\Lambda\otimes \bar{\Lambda}),
\quad \tilde{P}':=(\Lambda\otimes \bar{\Lambda})P'(\Lambda\otimes \bar{\Lambda}),
\quad
T:=P'(\Lambda^2\otimes \bar{\Lambda}^2)P'.
$$
Then 
\[
\|P-P'\|_\Lambda = 
\|\tilde{P}-\tilde{P}'\|_1 \leq
\|\tilde{P}-T\|_1 + 
\|T-\tilde{P}'\|_1.
\]

\paragraph{Bound $\|\tilde{P}-T\|_1$.} Decompose space
$\mathbb{C}^{N}\otimes \mathbb{C}^N$ into $\Pi_1:=P'$ and $\Pi_2:=(I-P')$. Then
we can write 
$$
\tilde{P} = \begin{bmatrix}
P_{11} & P_{12} \\
P_{21} & P_{22}
\end{bmatrix},
$$
where $P_{ij}:= \Pi_{i} \tilde{P} \Pi_j$.
Then 
$$
\|\tilde{P}-T\|_1 \leq
\|P_{11}-T\|_1 + \|P_{22}\|_1 + \|P_{12}\|_1 + \|P_{21}\|_1.
$$
We bound each term in the following. 
Since $\tilde{P} = (\Lambda\otimes \bar{\Lambda})P(\Lambda\otimes \bar{\Lambda}) \preceq \Lambda^2\otimes\bar{\Lambda}^2$, we have
$$
0\preceq P_{11}= P'\tilde{P}P'\preceq P'(\Lambda^2\otimes\bar{\Lambda}^2)P'=T
\quad\Longrightarrow\quad
T-P_{11}\succeq 0.
$$
Thus the first term
$$
\|P_{11}-T\|_1=\mathrm{Tr}(T-P_{11}) 
= \mathrm{Tr}(T)-\mathrm{Tr}(P_{11}).
$$
By $P'=\Pi_{\mathrm{L}(U^\perp,U)} = \Pi_U\otimes \bar{\Pi}_{U^\perp}$, we have
$$
\mathrm{Tr}(T) 
=\mathrm{Tr}(P'(\Lambda^2\otimes \bar{\Lambda}^2))
=\Tr(\Pi_U \Lambda^2) \overline{\Tr(\Pi_{U^\perp} \Lambda^2)}
= \eta (1-\eta) \leq \frac{1}{4},
$$
and 
$$
\mathrm{Tr}(P_{11}) = \mathrm{Tr}(P'\tilde{P}P')= \mathrm{Tr}(\tilde{P}P') \geq \frac{1}{4}- O(\sqrt{\epsilon}),
$$
where the last inequality is from \Cref{lem:upper-right}. Thus
$$
\|P_{11}-T\|_1 \leq \frac14 - \left(\frac14 - O(\sqrt{\epsilon})\right) = O(\sqrt{\epsilon}).
$$

For the second term, as $P_{22}\succeq 0$, 
$$
\|P_{22}\|_1=\mathrm{Tr}(P_{22}) = \mathrm{Tr}((I-P')\tilde{P})
= O(\sqrt{\epsilon}),
$$
where the last inequality is from \Cref{lem:upper-right}. 

For the third term, $\tilde{P}\succeq 0$ implies that 
\(
\|P_{12}\|_1 \leq \sqrt{\mathrm{Tr}(P_{11})\mathrm{Tr}(P_{22})}.
\)
By
\(
\mathrm{Tr}(P_{11}) \leq \frac14
\) and
\(
\mathrm{Tr}(P_{22}) = O(\sqrt{\epsilon}),
\)
we have
\(
\|P_{12}\|_1 = O(\epsilon^{1/4}).
\)
Finally, since $P_{21}=P_{12}^\dagger$, the fourth term $\|P_{21}\|_1=\|P_{12}\|_1 = O(\epsilon^{1/4})$.

Putting things together, we have 
\begin{equation} \label{eq:tP-T}
\|\tilde{P}-T\|_1 \leq 
2 O(\sqrt{\epsilon}) + 2 O(\epsilon^{1/4}) = O(\epsilon^{1/4}).
\end{equation}

\paragraph{Bound $\|T-\tilde{P}'\|_1$.} By \Cref{lem:T-tP}, we have
\begin{equation} \label{eq:T-tP}
\|T-\tilde{P}'\|_1 \leq 8\sqrt{\delta}. 
\end{equation}

Combining \eqref{eq:tP-T}, \eqref{eq:T-tP}, and \Cref{lem:delta-bound}, we have
\[
\|P-P'\|_\Lambda =
\|\tilde{P}-\tilde{P}'\|_1 \le
O(\epsilon^{1/4}) + 8\sqrt{\delta} = O(\epsilon^{1/4}).  \qedhere
\]
\end{proof}

\begin{lemma} \label{lem:T-tP}
\(
\|T-\tilde{P}'\|_1 \leq 8\sqrt{\delta}. 
\)
\end{lemma}

\begin{proof}
Let $Q:=\Pi_U$ and $R:=\Pi_{U^\perp}$. By
$
P'= \Pi_{\mathrm{L}(U^\perp, U)} = Q\otimes \bar{R}
$,
we have
\[
T = P'(\Lambda^2\otimes \bar{\Lambda}^2)P' = 
(Q\Lambda^2 Q)\otimes \overline{R\Lambda^2 R}, 
\quad
\tilde{P}' = (\Lambda\otimes \bar{\Lambda})P'(\Lambda\otimes \bar{\Lambda}) =
(\Lambda Q\Lambda)\otimes \overline{\Lambda R\Lambda}.
\]
Compute
\[
T-\tilde{P}'=
\left(Q\Lambda^2 Q - \Lambda Q\Lambda\right)\otimes \overline{R\Lambda^2 R} + 
(\Lambda Q\Lambda)\otimes (\overline{R\Lambda^2 R - \Lambda R\Lambda}).
\]
By triangle inequality and $\|X\otimes Y\|_1 = \|X\|_1 \|Y\|_1$, we have
\begin{equation} \label{eq:T-tP-decomp}
\|T-\tilde{P}'\|_1 \leq
\|Q\Lambda^2 Q - \Lambda Q\Lambda\|_1 \|R\Lambda^2 R\|_1 + 
\|\Lambda Q\Lambda\|_1 \|R\Lambda^2 R - \Lambda R\Lambda\|_1.
\end{equation}

First, notice that 
\begin{equation}  \label{eq:Q-R-norm}
\|R\Lambda^2 R\|_1 = \Tr(R\Lambda^2 R) \leq \Tr(\Lambda^2) = 1, \quad
\|\Lambda Q\Lambda\|_1 = \Tr(\Lambda Q\Lambda) \leq \Tr(\Lambda^2) = 1.
\end{equation}
Next, we bound
\(
\|Q\Lambda^2 Q -\Lambda Q\Lambda\|_1.
\)
Write
\begin{align*}
\Lambda Q\Lambda &= (Q+R)\Lambda Q\Lambda(Q+R) = Q\Lambda Q\Lambda Q + R\Lambda Q\Lambda R + Q\Lambda Q\Lambda R + R\Lambda Q\Lambda Q, \\
Q\Lambda^2 Q &= Q\Lambda(Q+R)\Lambda Q = Q\Lambda Q\Lambda Q + Q\Lambda R\Lambda Q.
\end{align*}
Subtracting them gives
\[
\Lambda Q\Lambda - Q\Lambda^2 Q = 
R\Lambda Q\Lambda R + Q\Lambda Q\Lambda R + R\Lambda Q\Lambda Q - Q\Lambda R\Lambda Q.
\]
Then by H\"older's inequality $\|XY\|_1 \leq \|X\|_F \|Y\|_F$, definition $\delta=\|Q\Lambda R\|_F^2 = \|R\Lambda Q\|_F^2$,
and the fact that $\|Q\Lambda Q\|_F, \|R\Lambda R\|_F\leq \|\Lambda\|_F = 1$, we have
\begin{align*}
\|R\Lambda Q\Lambda R\|_1 &\leq \|R\Lambda Q\|_F \|Q\Lambda R\|_F = \delta, 
&\|Q\Lambda Q\Lambda R\|_1 &\leq \|Q\Lambda Q\|_F \|Q\Lambda R\|_F \leq \sqrt{\delta}, \\
\|R\Lambda Q\Lambda Q\|_1 &\leq \|R\Lambda Q\|_F \|Q\Lambda Q\|_F \leq \sqrt{\delta}, 
&\|Q\Lambda R\Lambda Q\|_1 &\leq \|Q\Lambda R\|_F \|R\Lambda Q\|_F = \delta.
\end{align*}
Combining them gives
\begin{equation} \label{eq:Q-Lambda-Q}
\|Q\Lambda^2 Q - \Lambda Q\Lambda\|_1 \leq 2\delta + 2\sqrt{\delta} \leq 4\sqrt{\delta},
\end{equation}
where the last inequality is from $\delta\leq 1$.
Similarly, we can show that
\begin{equation} \label{eq:R-Lambda-R}
\|R\Lambda^2 R - \Lambda R\Lambda\|_1 \leq 4\sqrt{\delta}.
\end{equation}

Plugging \eqref{eq:Q-R-norm}, \eqref{eq:Q-Lambda-Q}, and \eqref{eq:R-Lambda-R} into \eqref{eq:T-tP-decomp}, we have
\[
\|T-\tilde{P}'\|_1 \leq 4\sqrt{\delta} + 4\sqrt{\delta} = 8\sqrt{\delta}. \qedhere
\]
\end{proof}

Now we are ready to prove \Cref{thm:weighted-stability}.

\begin{proof}[Proof of \Cref{thm:weighted-stability}]
    If $\epsilon=0$, take $P'=P$ and $\Lambda'=\Lambda$.

    Suppose next that $\epsilon\geq 1/16 = \Omega(1)$. 
    Since $N\geq 2$, we can take
    any valid pair $(P',\Lambda')$ such that $\Es_{\Lambda'}(P')=0$ and $m_{\Lambda'}(P')=1/4$.
    Then the required approximation bounds follow since
    \[
        \|\Lambda-\Lambda'\|_F^2
        =2-2\Tr(\Lambda\Lambda')\leq2 = O(\epsilon^{1/4}),
        \qquad
        \|P-P'\|_\Lambda
        \leq m_\Lambda(P)+m_\Lambda(P')\leq 2 = O(\epsilon^{1/4}).
    \]

    Now assume $0<\epsilon<1/16$. We claim that the threshold
    \(
        t:=\max\{j\in[N]:\rho_j\geq\sqrt{\epsilon}\}
    \)
    is well defined and satisfies $t<N$. Indeed, if the defining set were
    empty, then $B\prec\sqrt{\epsilon}I$, giving
    \begin{equation} \label{eq:con1}
        \frac14=m_\Lambda(P)=\Tr(\Lambda B\Lambda)
        <\sqrt{\epsilon}<\frac14,
    \end{equation}
    a contradiction. If $t=N$, then $B\succeq\sqrt{\epsilon}I$, and hence
    \begin{equation} \label{eq:con2}
        \epsilon=\Tr(\Lambda A\Lambda B)
        \geq\sqrt{\epsilon}\Tr(\Lambda A\Lambda)
        =\frac14\sqrt{\epsilon},
    \end{equation}
    which implies $\epsilon\geq1/16$, also a contradiction. Thus $1\leq t<N$.

    It remains to check that the normalization factors
    $\Tr(X^2)$ and $\Tr(Y^2)$ 
    in the definition of
    $\Lambda'$ are nonzero. Let $Q:=\Pi_U$ and $R:=\Pi_{U^\perp}$. If
    $\Tr(X^2)=0$, then $X=Q\Lambda Q=0$. Since $\Lambda\succeq0$, this implies
    $\Lambda=R\Lambda R$,
    and thus $\Lambda = R\Lambda = \Lambda R$.
    Since 
    $RBR\prec \sqrt{\epsilon} I$,
    we would have
    \[
        \frac14=m_\Lambda(P)=\Tr(\Lambda^2B) = 
        \Tr(\Lambda^2 \cdot RBR)
        \leq \sqrt{\epsilon} \Tr(\Lambda^2) = \sqrt{\epsilon} < \frac14,
    \]
    a contradiction. Similarly, if $\Tr(Y^2)=0$, then
    $\Lambda = Q\Lambda = \Lambda Q$. Thus by $QBQ\succeq \sqrt{\epsilon} Q$, we would have
    \[
        \epsilon=\Tr(\Lambda A\Lambda B)
        =\Tr(\Lambda A \Lambda \cdot QBQ)
        \geq\sqrt{\epsilon}\Tr(\Lambda A\Lambda)
        =\frac14\sqrt{\epsilon},
    \]
    which implies $\epsilon\geq1/16$, again a contradiction. Therefore
    $\Tr(X^2),\Tr(Y^2)>0$, and the pair $(P',\Lambda')$ constructed in
    \Cref{sec:ze-candidate} is well defined. The conclusion now follows from
    \Cref{fact:ze-candidate,lem:lambda-close,lem:p-close}.
\end{proof}

 \section{One-Way One-Round Quantum Coloring of Directed Cycles}
\label{sec:one-way}

In this section, we return to the quantum-LOCAL model and study one-way
one-round anonymous quantum algorithms for the directed-cycle $q$-coloring problem. 
The nodes
are arranged on a directed $n$-cycle
\[v_1\to v_2\to\cdots\to v_n\to v_1.\] 
They are initially identical and execute the same quantum algorithm.
During the single communication round, each node sends one quantum message
to its successor and receives one from its predecessor; it then performs a local
measurement and outputs a color in $[q]$. 

We express the probability of a local
collision exactly in terms of the $\Lambda$-energy of the corresponding
measurement operators. This yields a general characterization for every $q$, and
our main impossibility result for $q=4$.

\subsection{General Reduction} \label{sec:reduction}

We first present a general statement on the 
energy of POVMs, which is equivalent to the impossibility of one-way one-round quantum anonymous
algorithms for $q$-coloring directed cycles.

\begin{theorem} \label{thm:reduction}
    For every fixed integer $q\ge 2$,
    the following two statements are equivalent:
    \begin{enumerate}[
    label=(Q\arabic*),
    ref=(Q\arabic*),
    leftmargin=1.5cm,
    ]
    \item\label{item:povm} There exists a universal constant $C_q > 0$ such that
    for every choice of local dimension $N$, PSD matrix $\Lambda \in \mathbb{C}^{N\times N}$ with
    $\|\Lambda\|_F=1$, and $q$-outcome POVM $\mathcal{M}=\{P_i\}_{i\in[q]}$
    on $\mathbb{C}^N\otimes \mathbb{C}^N$,
    we have
    \[
    \sum_{i\in[q]} \Es_\Lambda(P_i) \geq C_q.    
    \]
    \item\label{item:algo} One-way one-round anonymous quantum algorithms
    cannot $q$-color directed cycles with high probability.
    \end{enumerate}
\end{theorem}

\begin{proof}
    We first show that \ref{item:povm} implies \ref{item:algo}. 
    Suppose \ref{item:povm} holds for $q$. 
    Since the local computation is unbounded, by purifying private randomness and intermediate measurements, 
    each node may store the private randomness, and intermediate outcomes
    coherently in the private workspace. Then their effects can be absorbed into
    a pure bipartite state and a final POVM. Therefore without loss of generality, 
    we may assume any one-way one-round quantum $q$-coloring algorithm has the following form
    \begin{enumerate}
    \item Each node $v_\ell$ prepares an identical bipartite state
    $\ket{\Psi}_{\A_\ell \B_\ell}$ and sends register $\B_\ell$ to $v_{\ell+1}$.
    \item Each node $v_\ell$ performs the same $q$-outcome POVM measurement
    $\mathcal{M}=\{P_i:i\in[q]\}$ on the bipartite system $\B_{\ell-1}\A_\ell$
    and outputs the color $c_\ell\in[q]$.
    \end{enumerate}
    
    We write the Schmidt decomposition of $\ket{\Psi}$ as
    \[
    \ket{\Psi}=\sum_{t\in[N]}\lambda_t\ket{\psi_t}_\A\ket{\varphi_t}_\B,
    \]
    where $\lambda_t\ge 0$ are Schmidt coefficients satisfying $\sum_t \lambda_t^2=1$.
    
    Define a local unitary $U\in U(N)$ that maps $\{\ket{\bar{\varphi}_t}\}$ 
    to $\{\ket{\psi_t}\}$. Since each node $v_\ell$ can apply $I\otimes U$
    locally on $\B_{\ell-1}\A_\ell$ right after communication,
    we may absorb $I\otimes U$ into $\mathcal{M}$, and assume without loss of generality that the shared bipartite state is
    \begin{equation}\label{eq:schmidt-canonical}
    \ket{\Psi}=\sum_{t\in[N]}\lambda_t\ket{\bar{\varphi}_t,\varphi_t} = \V(\bar{\Lambda})\in \mathbb{C}^N\otimes \mathbb{C}^N,
    \end{equation}
    where $\Lambda := \sum_{t\in[N]}\lambda_t\ket{\varphi_t}\bra{\varphi_t}$
    is a PSD matrix with $\|\Lambda\|_F=1$.
    
    For two adjacent nodes $v_\ell$ and $v_{\ell+1}$,
    they hold the registers $\B_{\ell-1}\A_\ell \B_\ell \A_{\ell+1}$ after communication. 
    Their joint state is 
    \[
    \sigma_{\B_{\ell-1}\A_\ell \B_\ell \A_{\ell+1}}:=\Tr_{\As}\!\left(\ket{\Psi}\bra{\Psi}\right)_{\B_{\ell-1}}
    \otimes \ket{\Psi}\bra{\Psi}_{\A_\ell \B_\ell}
    \otimes \Tr_{\Bs}\!\left(\ket{\Psi}\bra{\Psi}\right)_{\A_{\ell+1}}
    = \Lambda^2 \otimes \ket{\Psi}\bra{\Psi} \otimes \bar{\Lambda}^2,
    \]

    Then the probability of both nodes outputting the same color $i\in[q]$ is
    \[
    \Pr[c_\ell=c_{\ell+1}=i] = 
    \Tr\!\left(
    (P_i\otimes P_i)\sigma
    \right) = \Es_\Lambda(P_i),
    \]
    where the last equality follows from \Cref{lem:op-energy}. 

    Thus by \ref{item:povm}, we have the local collision probability
    \[
    \Pr[c_\ell=c_{\ell+1}] = \sum_{i\in[q]} \Pr[c_\ell=c_{\ell+1}=i] = \sum_{i\in[q]} \Es_\Lambda(P_i) \geq C_q,
    \]
    which is lower bounded by a constant independent of cycle length $n$. 
    
    Since the algorithm is symmetric on the cycle,
    $\Pr[c_\ell=c_{\ell+1}]$ is the same for all $\ell \in [n]$.
    As the algorithm is one-way one-round, collision events on edges whose
    pairwise cyclic distances are at least $3$ are independent. Selecting
    $\lfloor n/3\rfloor$ such edges, the algorithm produces a proper
    $q$-coloring with probability at most $(1-C_q)^{\lfloor n/3\rfloor}=o(1)$.

    For the converse direction, suppose \ref{item:povm} does not hold. Then
    for any $\epsilon>0$, there exists a $q$-outcome POVM
    $\mathcal{M}=\{P_i\}_{i\in[q]}$ and a PSD matrix $\Lambda$ with
    $\|\Lambda\|_F=1$ such that 
    \[ 
    \sum_{i\in[q]} \Es_\Lambda(P_i) \leq  \epsilon.
    \]
    Construct a one-way one-round quantum algorithm\footnote{
        Although the failure of \ref{item:povm} only asserts the
        existence of $(\{P_i\}, \Lambda)$, 
        the induced algorithm can be made uniform: 
        If \ref{item:povm} fails,
        then there exists a finite-dimensional pair $(\{P_i\}, \Lambda)$ 
        of total energy at most $\epsilon/2$.
        Since we allow unbounded local computation,
        and the energy is computable and continuous, 
        the algorithm can exhaustively search over dense enough approximations
        of $(\{P_i\}, \Lambda)$ to find a candidate pair
        with total energy at most $\epsilon$. 
    }
    by preparing the state 
    $\V(\bar{\Lambda})$, sending one register to the neighbor, and performing the POVM $\mathcal{M}$ at each node. 
    Then the local failure probability is exactly $\sum_{i\in[q]} \Es_\Lambda(P_i) \leq \epsilon$.
    By setting $\epsilon = 1/n^2$ and applying the union bound,
    the algorithm produces a proper $q$-coloring of the $n$-cycle with probability at least $1-n\cdot 1/n^2 = 1-1/n$.
    Thus it solves the directed-cycle $q$-coloring problem.
\end{proof}

\subsection{Impossibility of \texorpdfstring{$\leq 3$}{3 or Fewer} Colors}

One-way one-round quantum $q$-coloring is known to be impossible for $q\leq 3$,
from a bounded-dependence perspective \cite{gall2022non,holroyd2017finitary}. 
Namely, stationary $1$-dependent $q$-coloring processes do not exist for $q\leq 3$. 
As a corollary, Statement \ref{item:povm} holds for $q\leq 3$.

\begin{lemma}[Section 2 of \cite{gandolfi1989extremal}] \label{lem:1-dep}
    Given a stationary $1$-dependent process $(Y_t\in \{0,1\})_{t\in \mathbb{Z}}$, 
    if 
    \[
    \alpha:=\Pr[Y_t=1] \in \left[1/4,1/3\right],
    \]
    then
    \[
    \Pr[Y_t=1,Y_{t+1}=1] \geq 
    \frac{(1-2\sqrt{1-3\alpha})(1+\sqrt{1-3\alpha})^2}{27}.
    \]
\end{lemma}

Using the above lemma, we show that any operator with
mass significantly larger than $1/4$ must have large energy.

\begin{lemma} \label{coro:weighted-large-energy}
    Given any PSD $P \in \mathrm{L}(\mathbb{C}^N\otimes \mathbb{C}^N)$ such that
    $0\preceq P\preceq I$ and any PSD matrix $\Lambda \in \mathbb{C}^{N\times
    N}$ with $\|\Lambda\|_F=1$, if 
    $m_\Lambda(P) = 1/4 + \epsilon$ for some $\epsilon\geq 0$,
    then we have $\Es_\Lambda(P)\geq \min\{2\epsilon/9, 1/54\}$.    
\end{lemma}

\begin{proof}
    Consider the two-outcome POVM $\mathcal{M}:=\{P_0=I-P, P_1=P\}$ and
    the bipartite quantum state $\ket{\Psi} = \V(\bar{\Lambda})$.
    Define the process $(Y_t\in\{0,1\})_{t\in\mathbb{Z}}$ by the following algorithm:
    \begin{enumerate}
        \item Each node $t\in \mathbb{Z}$ prepares state $\ket{\Psi}$ in register $\A_{t}\B_{t}$, and 
            sends $\B_{t}$ to node $t+1$.        
        \item Each node $t \in\mathbb{Z}$ applies $\mathcal{M}$ to registers $\B_{t-1}\A_{t}$ and outputs the measurement outcome $Y_{t}\in\{0,1\}$.
    \end{enumerate}
    Since this is a one-way one-round algorithm where
    each node runs the same algorithm, $(Y_t)_{t\in\mathbb{Z}}$ is stationary and $1$-dependent. By
    \Cref{lem:op-energy},
    \[
        \Pr[Y_t=1] = m_{\Lambda}(P) = \frac14 + \epsilon,
        \quad
        \Pr[Y_t=Y_{t+1}=1] = \Es_\Lambda(P).
    \]
    Suppose first that $\epsilon\leq 1/12$. Then
    $\Pr[Y_t=1] = 1/4 + \epsilon \in [1/4,1/3]$, so
    \Cref{lem:1-dep} and the inequality
    $\sqrt{1-x}\leq 1-\frac{x}{2}$ for $x\in[0,1]$ give
    \[
    \Es_\Lambda(P) \geq 
    \frac{(1-2\sqrt{1-3(1/4+\epsilon)})(1+\sqrt{1-3(1/4+\epsilon)})^2}{27}
    \geq \frac{(1-\sqrt{1-12\epsilon})\cdot 1}{27} 
    \geq \frac{(1 - (1-6\epsilon))}{27} = \frac29 \epsilon. 
    \]
    Now suppose that $\epsilon>1/12$, and set
    $c:=1/(3m_\Lambda(P))<1$. The operator $P':=cP$ satisfies
    $0\preceq P'\preceq I$ and $m_\Lambda(P')=1/3$. Applying the preceding
    case to $P'$ and 
    observing that $\Lambda$-energy is quadratic in $P$, 
    we obtain
    \[
        c^2\Es_\Lambda(P)=\Es_\Lambda(P')\geq \frac29\left(\frac13-\frac14\right) = \frac{1}{54}.
    \]
    Hence $\Es_\Lambda(P)\geq 1/54$, completing the proof.
\end{proof}

\begin{corollary} \label{coro:3-coloring}
    Statement \ref{item:povm} in \Cref{thm:reduction} holds for $q\leq 3$. 
\end{corollary}

\begin{proof}
    As a $2$-outcome POVM can be viewed as a $3$-outcome POVM by adding an extra zero operator, 
    it suffices to show that for any $3$-outcome POVM $\{P_i\}_{i\in[3]}$, and any PSD matrix $\Lambda\in \mathbb{C}^{N\times N}$ with $\|\Lambda\|_F=1$, we have
    \[
    \sum_{i\in[3]} \Es_\Lambda(P_i) = \Omega(1).
    \]
    Notice that 
    \[
    \sum_{i\in[3]} m_{\Lambda}(P_i) = 
    \Tr\left(
    \sum_{i\in[3]} P_i (\Lambda^2\otimes \bar{\Lambda}^2)
    \right) = \Tr(\Lambda^2\otimes \bar{\Lambda}^2) = 1.
    \]
    By averaging, there exists $i^*\in[3]$ such that $m_\Lambda(P_{i^*})\geq 1/3$.
    Then by \Cref{coro:weighted-large-energy}, we have
    \[
    \sum_{i\in[3]} \Es_\Lambda(P_i) \ge \Es_\Lambda(P_{i^*}) \geq \frac{1}{54}. \qedhere
    \]
\end{proof}

\Cref{coro:3-coloring} also admits a direct operator-theoretic proof, without
relying on bounded-dependence arguments. See Appendix~\ref{appx:3-coloring} for
details. We remark that the above argument stops working from $q=4$ onward,
as $1$-dependent $4$-coloring processes do exist \cite{holroyd2018finitely}.

\subsection{Impossibility of 4 Colors}

Our main result is that Statement \ref{item:povm} holds for $q=4$, and hence
one-way one-round quantum $4$-coloring of directed cycles is impossible by \Cref{thm:reduction}. Formally,

\begin{theorem} \label{thm:q4}
    For any $4$-outcome POVM $\{P_i\}_{i\in[4]}$ 
    on $\mathbb{C}^N\otimes \mathbb{C}^N$ and any PSD matrix $\Lambda\in \mathbb{C}^{N\times N}$ with $\|\Lambda\|_F=1$, we have
    \[
    \sum_{i\in[4]} \Es_\Lambda(P_i) = \Omega(1).
    \]
    Consequently, one-way one-round
    anonymous quantum algorithms cannot solve directed-cycle $4$-coloring.
\end{theorem}

By the equivalence between anonymous quantum algorithms and quantum-LOCAL
algorithms as discussed in \Cref{sec:prelim-local}, we immediately have the following corollary.

\begin{corollary}
    One-way one-round quantum-LOCAL algorithms cannot
    solve directed-cycle $4$-coloring.
\end{corollary}

The remainder of this section proves the above theorem. 
\Cref{sec:unrobust} first proves a special case
where $P_1$ has zero energy and maximal mass, 
and \Cref{sec:robust} proves the general case.

\subsubsection{Warm-Up: One Color with Exact Zero Energy} \label{sec:unrobust}

We first consider a special case where one POVM operator $P_1$ has zero
energy and maximal mass. By the structural theorem for zero-energy
operators (\Cref{thm:ze-mm}), we know that there exists 
an orthogonal decomposition 
\[
\mathbb{C}^N = U\oplus U^\perp
\]
such that $P_1$ is the orthogonal projector onto the off-diagonal matrix space $\mathrm{L}(U^\perp,U)$.
Then by projecting the other three operators $P_2, P_3, P_4$
onto diagonal block $\mathrm{L}(U,U)$, we obtain a $3$-outcome POVM on $U\otimes \bar{U}$,
which reduces to the $3$-coloring case (\Cref{coro:3-coloring}).

\begin{lemma} \label{lem:special}
    Given any $4$-outcome POVM $\{P_i\}_{i\in[4]}$ on $\mathbb{C}^N\otimes
    \mathbb{C}^N$ and any PSD matrix $\Lambda\in \mathbb{C}^{N\times N}$ with
    $\|\Lambda\|_F=1$, if $P_1$ satisfies
    \[
    m_\Lambda(P_1) = \frac14
    \quad\text{and}\quad
    \Es_\Lambda(P_1) = 0,
    \]
    then we have
    \[
    \sum_{i=1}^4 \Es_\Lambda(P_i) = \Omega(1).
    \]
\end{lemma}

\begin{proof} 
    If $\Lambda$ is not full rank, let $\Pi_{\Lambda}$ be the orthogonal
    projector onto the support of $\Lambda$. Then we have
    $\Lambda=\Pi_\Lambda\Lambda = \Lambda\Pi_\Lambda$.
    For each $P_i$, define 
    \[
    P_i' := (\Pi_\Lambda \otimes \bar{\Pi}_\Lambda) P_i (\Pi_\Lambda \otimes \bar{\Pi}_\Lambda).
    \]
    Then one can verify that 
    \(
    m_\Lambda(P_i) = m_\Lambda(P_i')
    \)
    and 
    \(
    \Es_\Lambda(P_i) = \Es_\Lambda(P_i').
    \)
    Moreover, $\{P'_i\}$ forms a POVM on the support of $\Lambda\otimes \bar{\Lambda}$, as
    \[
    \sum_{i=1}^4 P_i' = (\Pi_\Lambda\otimes \bar{\Pi}_\Lambda) \left(\sum_{i=1}^4 P_i\right) (\Pi_\Lambda\otimes \bar{\Pi}_\Lambda)
    = \Pi_\Lambda\otimes \bar{\Pi}_\Lambda.
    \]
    Thus without loss of generality, we assume $\Lambda\succ 0$ in the following.

    As $m_\Lambda(P_1) = 1/4$ and $\Es_\Lambda(P_1) = 0$, by \Cref{thm:ze-mm},
    there exists an orthogonal decomposition $\mathbb{C}^N=U\oplus U^\perp$ such that 
    \emph{(i)} $P_1 = \Pi_{W}$ where the linear space 
    \[
    W := \mathrm{L}(U^\perp,U) = 
    \left\{    
    \left(\begin{array}{c|c}
            0 & X \\ \hline
            0 & 0
        \end{array}\right)
    : X \in \mathbb{C}^{|U|\times |U^\perp|}
    \right\}
    \]
    in the basis of $U\oplus U^\perp$, 
    \emph{(ii)} $\Tr(\Lambda^2\Pi_U)=\Tr(\Lambda^2\Pi_{U^\perp})=1/2$, and
    \emph{(iii)} $U, U^\perp$ are $\Lambda$-invariant.

    By $P_1 = \Pi_W$ and $\sum_{i\in[4]} P_i=I$, we have that $P_2 + P_3 + P_4 =
    \Pi_{W^\perp}$, and then for each $i\geq 2$, we have $P_i\preceq \Pi_{W^\perp}$,
    i.e., the support of $P_i$ is contained in $W^\perp$. Notice that in the
    basis of $U\oplus U^\perp$, 
    \[
    W^\perp = \left\{
    \left(\begin{array}{c|c}
            X & 0 \\ \hline
            Y & Z
    \end{array}\right)
    :X\in \mathbb{C}^{|U|\times |U|}, Y\in \mathbb{C}^{|U^\perp|\times |U|}, Z\in \mathbb{C}^{|U^\perp|\times |U^\perp|}
    \right\}.
    \]
    Thus if we write the spectral decomposition of $P_i$ ($i\geq 2$) as
    \[
    P_i = \sum_{k} \V\bigl(M_k^{(i)}\bigr)\V\bigl(M_k^{(i)}\bigr)^\dagger, 
    \]
    then each $M_k^{(i)}$ is in $W^\perp$, and thus can be written in block form
    \[
    M_k^{(i)} = \left(\begin{array}{c|c}
            \hat{M}_k^{(i)} & 0 \\ \hline
            \star & \star
    \end{array}\right)    
    \]
    where $\hat{M}_k^{(i)}\in \mathbb{C}^{|U|\times |U|}$. 
    Now define the reduced operators
    \[
    \hat{P}_i := \sum_{k} \V\bigl(\hat{M}_k^{(i)}\bigr)\V\bigl(\hat{M}_k^{(i)}\bigr)^\dagger
    \in \mathrm{L}(U\otimes \bar{U}).
    \]
    Equivalently, $\hat{P}_i = \Pi_{\mathrm{L}(U,U)} P_i \Pi_{\mathrm{L}(U,U)}$, so $0\preceq \hat{P}_i\preceq I$. 
    As $\Pi_{\mathrm{L}(U,U)} P_1 = \Pi_{\mathrm{L}(U,U)}\Pi_W = 0$, 
    and $\sum_{i\in[4]} P_i = I$, we have
    \[
    \sum_{i=2}^4 \hat{P}_i = \Pi_{\mathrm{L}(U,U)},
    \]
    and thus $\{\hat{P}_i\}_{i=2}^4$ is a $3$-outcome POVM on $U\otimes \bar{U}$.

    By $U$ and $U^\perp$ being $\Lambda$-invariant, $\Lambda$ can be written in the basis of $U\oplus U^\perp$ as
    \[
    \Lambda = \left(\begin{array}{c|c}
            \Lambda_U & 0 \\ \hline
            0 & \Lambda_{U^\perp}
    \end{array}\right),
    \]
    where $\Lambda_U\in \mathbb{C}^{|U|\times |U|}$ and
    $\Lambda_{U^\perp}\in \mathbb{C}^{|U^\perp|\times |U^\perp|}$ are 
    restrictions of $\Lambda$ on $U$ and $U^\perp$ respectively.
    Moreover, by $\Tr(\Lambda^2\Pi_U)=\Tr(\Lambda^2\Pi_{U^\perp})=\frac12$, we have
    \(
    \Tr(\Lambda_U^2)=\Tr(\Lambda_{U^\perp}^2)=\frac12.
    \)
    Then we define the reduced weight matrix
    \[
    \hat{\Lambda} := \sqrt{2}\Lambda_U \in \mathrm{L}(U),
    \]
    which satisfies $\|\hat{\Lambda}\|_F^2=2 \Tr(\Lambda_U^2) = 1$. Since
    \[
    \sum_{i=2}^4 m_{\hat{\Lambda}}\bigl(\hat{P}_i\bigr) =
    \Tr\left(\sum_{i=2}^4 \hat{P}_i \bigl(\hat{\Lambda}^2\otimes \bar{\hat{\Lambda}}^2\bigr)\right)
    = \Tr(\hat{\Lambda}^2\otimes \bar{\hat{\Lambda}}^2) = 1,
    \]
    by averaging, there exists $j\in\{2,3,4\}$ such that
    $m_{\hat{\Lambda}}\bigl(\hat{P}_{j}\bigr) \geq \frac13$. 
    By \Cref{coro:weighted-large-energy}, we have
    \[
    \Es_{\hat{\Lambda}}\bigl(\hat{P}_{j}\bigr) = \Omega(1).
    \]
    
    Finally, we show that $\Es_\Lambda(P_{j}) \geq \Es_{\hat{\Lambda}}(\hat{P}_{j})/8$.
    Observe that in the block form, for any $k,\ell$, we have
    \[
    \Lambda M_k^{(j)} \Lambda M_\ell^{(j)} \Lambda
    = \left(\begin{array}{c|c}
            \Lambda_U \hat{M}_k^{(j)} \Lambda_U \hat{M}_\ell^{(j)} \Lambda_U & 0 \\ \hline
            \star & \star
    \end{array}\right).
    \]
    Then
    \[
    \|\Lambda M_k^{(j)} \Lambda M_\ell^{(j)} \Lambda\|_F^2
    \geq \|\Lambda_U \hat{M}_k^{(j)} \Lambda_U \hat{M}_\ell^{(j)} \Lambda_U\|_F^2 
    = \frac18 \|\hat{\Lambda} \hat{M}_k^{(j)} \hat{\Lambda} \hat{M}_\ell^{(j)} \hat{\Lambda}\|_F^2.
    \]
    Thus by \Cref{fact:energy}, 
    \[
    \mathcal{E}_\Lambda(P_{j})
    = \sum_{k,\ell} \|\Lambda M_k^{(j)} \Lambda M_\ell^{(j)} \Lambda\|_F^2
    \geq \frac18 \sum_{k,\ell} \|\hat{\Lambda} \hat{M}_k^{(j)} \hat{\Lambda} \hat{M}_\ell^{(j)} \hat{\Lambda}\|_F^2
    = \frac18 \Es_{\hat{\Lambda}}(\hat{P}_{j})
    = \Omega(1).
    \]
    By $\sum_{i=1}^4 \Es_\Lambda(P_i)\geq \Es_\Lambda(P_j)=\Omega(1)$, we conclude the proof.
\end{proof}

\subsubsection{General Case: Proof of \texorpdfstring{\Cref{thm:q4}}{the Four-Color Impossibility}} \label{sec:robust}

The high-level idea for the general case is to use
the approximation theorem (\Cref{thm:weighted-stability}) to reduce the 
problem to the special case: If the energy of $P_1$ is already large, then we are done; otherwise, we
can approximate $P_1$ by a nearby operator $P_1'$ with exactly zero energy, and
then invoke the special case.

\begin{proof}[Proof of \Cref{thm:q4}]
    The case $N=1$ is trivial since $\{P_i\}$ 
    and $\Lambda$ become scalars such that
    $\sum_{i=1}^4 P_i = 1$ and $\Lambda=1$. 
    Then $\sum_{i=1}^4 \Es_\Lambda(P_i) = \sum_{i=1}^4 
    P_i^2 \geq (\sum_{i=1}^4 P_i)^2/4 = 1/4$.
    
    Now assume $N\geq 2$.
    By $\sum_{i=1}^4 P_i=I$, we have 
    $\sum_{i=1}^4 m_\Lambda(P_i)=\Tr(\Lambda^2\otimes \bar{\Lambda}^2)=1$.
    By averaging and without loss of generality, assume $m_\Lambda(P_1)\geq 1/4$.
    We claim that we can further assume $m_\Lambda(P_1)=1/4$. Indeed, 
    let $m_{\Lambda}(P_1) = 1/4 + \epsilon$, and consider two cases:
    \begin{enumerate}
        \item If $\epsilon=\Omega(1)$, then by \Cref{coro:weighted-large-energy}, we have
        \[
        \Es_\Lambda(P_1) \geq \min\left\{\frac29 \epsilon,\frac1{54}\right\} = \Omega(1),
        \]
        and thus 
        \(
            \sum_{i=1}^4 \Es_\Lambda(P_i) \geq \Es_\Lambda(P_1) = \Omega(1).
        \)
        \item If $\epsilon=o(1)>0$, then decompose $P_1 = \hat{P}_1 + R$,
        where $\hat{P}_1, R\succeq 0$, $m_\Lambda(\hat{P}_1)=1/4$ and $m_\Lambda(R)=\epsilon$.
        Then define 
        \[
        \hat{P}_2 := P_2 + R, \quad
        \hat{P}_3 := P_3, \quad
        \hat{P}_4 := P_4,
        \]
        so that $\{\hat{P}_i\}_{i\in[4]}$ forms a POVM, 
        with $m_\Lambda(\hat{P}_1)=1/4$. Note that 
        $\mathcal{E}_\Lambda(\hat{P}_1)\leq \mathcal{E}_\Lambda(P_1)$
        by $\hat{P}_1 \preceq P_1$,
        \begin{align*}
        \mathcal{E}_\Lambda(\hat{P}_2)
        &= 
        \mathcal{E}_\Lambda(P_2) + 
        \Tr(\Lambda A_{P_2}\Lambda B_{R}) + 
        \Tr(\Lambda A_{R}\Lambda B_{P_2}) + 
        \Tr(\Lambda A_{R}\Lambda B_{R})\\
        &\leq \mathcal{E}_\Lambda(P_2) + 
        \Tr(\Lambda B_{R}\Lambda) + 
        \Tr(\Lambda A_{R}\Lambda) +
        \Tr(\Lambda A_{R}\Lambda) 
        = \mathcal{E}_\Lambda(P_2) + 3\epsilon 
        = \mathcal{E}_\Lambda(P_2) + o(1),
        \end{align*}
        and $\mathcal{E}_\Lambda(\hat{P}_3)=\mathcal{E}_\Lambda(P_3)$,
        $\mathcal{E}_\Lambda(\hat{P}_4)=\mathcal{E}_\Lambda(P_4)$.
        Then we have
        \[
        \sum_{i=1}^4 \Es_\Lambda(P_i) \geq \sum_{i=1}^4 \Es_\Lambda(\hat{P}_i) - o(1).
        \]
        Thus $\sum_{i=1}^4 \Es_\Lambda(\hat{P}_i)=\Omega(1)$ implies $\sum_{i=1}^4 \Es_\Lambda(P_i)=\Omega(1)$.
    \end{enumerate}
    Now assume $m_\Lambda(P_1)=1/4$. If $\Es_\Lambda(P_1)=\Omega(1)$, then 
    $\sum_{i=1}^4 \Es_\Lambda(P_i) \geq \Es_\Lambda(P_1) = \Omega(1)$ and we are done.
    Otherwise, let $\Es:=\Es_\Lambda(P_1)=o(1)$.
    By \Cref{thm:weighted-stability}, there exist $Z$ and
    $\Lambda'$ such that
    \[
    \mathcal{E}_{\Lambda'}(Z) = 0,\quad 
    m_{\Lambda'}(Z)=\frac14,\quad
    \|P_1-Z\|_\Lambda = O(\Es^{1/4}),\quad\text{and}\quad
    \|\Lambda-\Lambda'\|_F^2 = O(\Es^{1/4}).
    \]
    Then by \Cref{lem:project-povm}, there exists POVM $\{P_i'\}_{i\in[4]}$ such that 
    \[
    P_1'=Z, \quad\text{and}\quad
    \|P_i-P_i'\|_\Lambda = 
    O(\Es^{1/8}) \quad\text{for each }i\in[4].
    \]
    
    As $\Es_{\Lambda'}(P_1') = 0$ and $m_{\Lambda'}(P_1')=1/4$, 
    by applying \Cref{lem:special}, we have
    \begin{equation} \label{eq:lppp}
    \sum_{i=1}^4 \Es_{\Lambda'}(P_i') = \Omega(1).
    \end{equation}

    Then by \Cref{lem:weighted-lip-psd} and $\|\Lambda-\Lambda'\|_F^2 = O(\Es^{1/4})$, we have for each $i \in [4]$, 
    \[
    |\Es_\Lambda(P_i') - 
    \Es_{\Lambda'}(P_i')| 
    \leq 6 \|\Lambda-\Lambda'\|_F = O(\Es^{1/8}),
    \]
    and combining with \eqref{eq:lppp} gives
    \begin{equation} \label{eq:lpp}
    \sum_{i=1}^4 \Es_\Lambda(P_i') \geq \sum_{i=1}^4 \Es_{\Lambda'}(P_i') - O(\Es^{1/8}) = \Omega(1).         
    \end{equation}
    Finally, for each $i\in[4]$, we have
    \[
    |\mathcal{E}_{\Lambda}(P_i)-\mathcal{E}_{\Lambda}(P_i')| = 
    |\mathrm{Tr}(\Lambda A\Lambda B) - \mathrm{Tr}(\Lambda A'\Lambda B')| 
    \leq |\mathrm{Tr}((\Lambda A\Lambda-\Lambda A'\Lambda)B)| + 
    |\mathrm{Tr}(A'(\Lambda B\Lambda-\Lambda B'\Lambda))|,
    \]
    where $A=\Tr_\As((\Lambda^2\otimes I)P_i)^\top, B=\Tr_\Bs((I\otimes \bar{\Lambda}^2)P_i)$, 
    and $A', B'$ are defined similarly for $P_i'$.
    
    For the first term, by H\"older's inequality, we have
    \[
        |\mathrm{Tr}((\Lambda A\Lambda-\Lambda A'\Lambda)B)| \leq 
        \|\Lambda A\Lambda - \Lambda A'\Lambda\|_1 \|B\|_\infty.
    \]
    As $B=\mathrm{Tr}_B(P_i(I\otimes \bar{\Lambda}^2))\preceq \mathrm{Tr}_B(I\otimes
    \bar{\Lambda}^2)=\mathrm{Tr}(\bar{\Lambda}^2)I=I$, we have $\|B\|_\infty\leq 1$.
    
    Note that 
    \[
    \Lambda A\Lambda = \mathrm{Tr}_\As((\Lambda\otimes \bar{\Lambda})P_i (\Lambda\otimes \bar{\Lambda}))^\top
    \quad\text{and}\quad
    \Lambda A'\Lambda = \mathrm{Tr}_\As((\Lambda\otimes \bar{\Lambda}) P_i' (\Lambda\otimes \bar{\Lambda}))^\top.
    \]
    By contractivity of trace norm under partial trace, we have
    \[
    \|\Lambda A\Lambda-\Lambda A'\Lambda\|_1 \leq \|(\Lambda\otimes \bar{\Lambda})(P_i-P_i')(\Lambda\otimes \bar{\Lambda})\|_1 
    = \|P_i-P_i'\|_\Lambda =    
    O(\Es^{1/8}).
    \]
    Then
    \[
    |\mathrm{Tr}((\Lambda A\Lambda-\Lambda A'\Lambda)B)| = O(\Es^{1/8}).
    \]
    Similarly, the second term $|\mathrm{Tr}(A'(\Lambda B\Lambda-\Lambda B'\Lambda))|=O(\Es^{1/8})$.
    
    Thus for each $i\in[4]$, we have
    \(
    |\mathcal{E}_{\Lambda}(P_i)-\mathcal{E}_{\Lambda}(P_i')| = 
    O(\Es^{1/8}) = o(1).
    \)
    Combining with \eqref{eq:lpp} gives
    \[
    \sum_{i=1}^4 \Es_\Lambda(P_i) \geq \sum_{i=1}^4 \Es_{\Lambda}(P_i') - o(1) = \Omega(1).
    \]
    
    Therefore, Statement \ref{item:povm}
    in \Cref{thm:reduction}
    holds for $q=4$, and by
    \Cref{thm:reduction}, no one-way one-round quantum anonymous algorithm
    can $4$-color directed cycles with high probability.
\end{proof}

\section*{AI Disclosure}
The main results and their proofs were obtained without the use of AI. We have then used AI tools for literature review, proofreading, and improving the exposition.

\bibliography{ref}
\bibliographystyle{alpha}

\appendix

\section{Missing Proofs}

\subsection{Proof of \texorpdfstring{\Cref{lem:project-povm}}{Perturbation Lemma}}
\label{appx:xy}

We first show that 
two matrices are close up to a unitary if their Gram matrices are close, which is a consequence of Uhlmann's theorem. 

\begin{lemma} \label{lem:xy}
    Given $X, Y \in \mathbb{C}^{N\times N}$, there exists a unitary $U$ such that
    \[
    \|X-UY\|_F^2 \leq \|X^\dagger X - Y^\dagger Y\|_1.
    \]
\end{lemma}

\begin{proof}
Let $\alpha = \Tr(X^\dagger X)$ and $\beta = \Tr(Y^\dagger Y)$.
If $\alpha=0$, then $X=0$ and by setting $U=I$, we have
\[
\|X-UY\|_F^2 = \|Y\|_F^2 =
\Tr(Y^\dagger Y) = 
\|Y^\dagger Y\|_1 =
\|X^\dagger X - Y^\dagger Y\|_1.
\]
Similarly, if $\beta=0$, then $U=I$ also works.

Now assume $\alpha, \beta > 0$.
Define the quantum states
\(
\rho := X^\dagger X/\alpha, 
\sigma := Y^\dagger Y/\beta. 
\)
Then $X/\sqrt{\alpha}$ and $Y/\sqrt{\beta}$ are purifications of $\rho$ and $\sigma$ respectively.
By Uhlmann's theorem \cite{nielsen2010quantum}, we have
\[
\max_{\text{unitary }U}\left|\Tr\left[
    \left(\frac{X}{\sqrt{\alpha}}\right)^\dagger U
    \left(\frac{Y}{\sqrt{\beta}}\right) 
\right]\right| = 
F(\rho, \sigma),
\]
where $F(\rho, \sigma) := \|\sqrt{\rho}\sqrt{\sigma}\|_1$ is the fidelity between $\rho$ and $\sigma$.
As fidelity is homogeneous, we have 
\(
F(\rho, \sigma) = \frac{1}{\sqrt{\alpha\beta}} F(X^\dagger X, Y^\dagger Y)
\), and thus
\[
\max_{\text{unitary }U}\left|\Tr\left[
    X^\dagger U
    Y\right]\right| = F(X^\dagger X, Y^\dagger Y).
\]
Moreover, as we optimize over all unitaries, we have
\[
\max_{\text{unitary }U} \mathrm{Re}\left(\Tr\left[X^\dagger UY\right]\right) =
\max_{\text{unitary }U} \left|\Tr\left[X^\dagger UY\right]\right| =
F(X^\dagger X, Y^\dagger Y).
\]

Let $U$ be the unitary that achieves $\mathrm{Re}(\Tr(X^\dagger UY)) = F(X^\dagger X, Y^\dagger Y)$.
Expand the Frobenius norm
\begin{align*}
\|X-UY\|_F^2 = \Tr(X^\dagger X) + \Tr(Y^\dagger Y) - 2\mathrm{Re}(\Tr(X^\dagger UY)) 
&= \Tr(X^\dagger X) + \Tr(Y^\dagger Y) - 2 F(X^\dagger X, Y^\dagger Y).
\end{align*}
By the Powers--St{\o}rmer inequality \cite[Lemma 4.1]{powers1970free}, we have
\[
\|X^\dagger X - Y^\dagger Y\|_1 \geq \|\sqrt{X^\dagger X} - \sqrt{Y^\dagger Y}\|_F^2.
\]
The right-hand side can be expanded as
\[
\|\sqrt{X^\dagger X} - \sqrt{Y^\dagger Y}\|_F^2
= \Tr(X^\dagger X) + \Tr(Y^\dagger Y) - 2\Tr(\sqrt{X^\dagger X}\sqrt{Y^\dagger Y}),
\]
and the cross term can be bounded as
\[
\Tr(\sqrt{X^\dagger X}\sqrt{Y^\dagger Y}) 
\leq \|I\|_\infty 
\| \sqrt{X^\dagger X}\sqrt{Y^\dagger Y}\|_1 
= \|\sqrt{X^\dagger X}\sqrt{Y^\dagger Y}\|_1
= F(X^\dagger X, Y^\dagger Y).
\]
Thus
\[
\|X^\dagger X - Y^\dagger Y\|_1 \geq 
\Tr(X^\dagger X) + \Tr(Y^\dagger Y) - 2 F(X^\dagger X, Y^\dagger Y)
= \|X-UY\|_F^2. \qedhere
\]
\end{proof}

Now we prove \Cref{lem:project-povm}.

\begin{proof}[Proof of \Cref{lem:project-povm}]
Define $R := I - P_1$ and $S := I - Z$. 
Then we have
\[
\|R-S\|_\Lambda = 
\|P_1-Z\|_\Lambda = \epsilon.
\]
For $2\leq i \leq 4$, define
\[
M_i := \sqrt{R^{+}}P_i\sqrt{R^{+}} + \frac13\Pi_{\ker(R)}.
\]
Then $\{M_i\}$ satisfies
\[
0\preceq M_i\preceq I, 
\quad
P_i=R^{1/2}M_iR^{1/2}
\quad\text{and}\quad
\sum_{i=2}^4 M_i=I.
\]
Let $D:= \Lambda\otimes \bar{\Lambda}$, $X:=R^{1/2} D$ and $Y:=S^{1/2} D$.
By \Cref{lem:xy}, there exists a unitary $U$ such that
\[
\|X-UY\|_F^2 \leq \|X^\dagger X - Y^\dagger Y\|_1
= \|D(R-S)D\|_1 = \|R-S\|_\Lambda = \epsilon.
\]
Therefore,
$$
\|R^{1/2} D - US^{1/2} D\|_F = \|X-UY\|_F 
 \leq \sqrt{\epsilon}.
$$
Then define the new POVM elements as 
$$
P_1'=Z, \quad\text{and}\quad
P_i'=S^{1/2} U^\dagger M_i U S^{1/2}
\quad\text{for}\quad i\geq 2.
$$
The set
$\{P_i'\}_{i\in[4]}$ forms a POVM since $P_i'\succeq 0$ and 
\[
\sum_{i=1}^4 P_i' = 
Z + S^{1/2} U^\dagger \left(\sum_{i=2}^4 M_i\right) U S^{1/2}
= Z + S^{1/2} U^\dagger U S^{1/2} = Z + S = I.
\]
It remains to bound $\|P_i-P_i'\|_\Lambda=\|D(P_i - P_i')D\|_1$ for $2\leq i\leq 4$.
We compute
\begin{align*}
D(P_i-P_i')D
&= D(R^{1/2}M_iR^{1/2}-S^{1/2}U^\dagger M_i U S^{1/2})D \\
&= X^\dagger M_i X - (UY)^\dagger M_i (UY) 
\end{align*}
By triangle inequality and H\"older's inequality, we have
\begin{align*}
\|D(P_i-P_i')D\|_1 
&\leq \|(X-UY)^\dagger M_i X\|_1 + \|(UY)^\dagger M_i (X-UY)\|_1 \\
&\leq \|X-UY\|_F \|M_i X\|_F + \|(UY)^\dagger M_i\|_F \|X-UY\|_F
\end{align*}
Note that 
$
\|M_i X\|_F 
= \|M_i R^{1/2} D\|_F 
\leq \|D\|_F = 1
$, 
and similarly, 
$\|(UY)^\dagger M_i\|_F \leq 1$. 
Combining with $\|X-UY\|_F \leq \sqrt{\epsilon}$, we have
\[
\|D(P_i-P_i')D\|_1 \leq 2 \|X-UY\|_F \leq 2\sqrt{\epsilon} = O(\sqrt{\epsilon}). \qedhere
\]
\end{proof}

\subsection{Proof of \texorpdfstring{\Cref{lem:weighted-lip-psd}}{Energy Lip Lemma}} \label{appx:weighted-lip-psd}

\begin{proof}[Proof of \Cref{lem:weighted-lip-psd}]
    Let $\Delta := \Lambda-\Lambda'$. 
    Define  $A:=\Tr_\As((\Lambda^2\otimes I)P)^\top, B:=\Tr_\Bs((I\otimes \bar{\Lambda}^2)P)$, 
    with $A', B'$ defined similarly for $\Lambda'$. Then we have
    \begin{align}
    \notag|\Es_{\Lambda}(P)-\Es_{\Lambda'}(P)| 
    &= |\Tr(\Lambda A\Lambda B) - \Tr(\Lambda' A'\Lambda' B')| \\
    \notag&= |\Tr(\Lambda A\Lambda B) 
    - \Tr(\Lambda' A'\Lambda' B)
    + \Tr(\Lambda' A'\Lambda' B)
    - \Tr(\Lambda' A'\Lambda' B')| \\
    \notag&\leq |\Tr((\Lambda A\Lambda - \Lambda' A'\Lambda')B)| + |\Tr(\Lambda' A'\Lambda'(B-B'))| \\
    \label{eq:lal}&\leq \|\Lambda A\Lambda - \Lambda' A'\Lambda'\|_1 \|B\|_\infty + \|\Lambda' A'\Lambda'\|_1 \|B-B'\|_\infty,
    \end{align}
    where the last inequality follows from H\"older's inequality $|\Tr(XY)|\leq \|X\|_1\|Y\|_\infty$.
    
    First, observe that $A', B\preceq I$, and thus $\|B\|_\infty\leq 1$, 
    and $\|\Lambda' A'\Lambda'\|_1 = \Tr(\Lambda' A'\Lambda') \leq \Tr(\Lambda'^2) = 1$.

    Next, we bound $\|\Lambda A\Lambda - \Lambda' A'\Lambda'\|_1$. 
    By definition, we have
    \[
    \Lambda A\Lambda - \Lambda' A'\Lambda' = \Tr_\As\left[(\Lambda \otimes \bar{\Lambda}) P(\Lambda \otimes \bar{\Lambda}) - (\Lambda' \otimes \bar{\Lambda}') P(\Lambda' \otimes \bar{\Lambda}')\right]^\top.
    \]
    Then
    \begin{align*}
    \|\Lambda A\Lambda - \Lambda' A'\Lambda'\|_1
    &\leq \|(\Lambda \otimes \bar{\Lambda}) P(\Lambda \otimes \bar{\Lambda}) - (\Lambda' \otimes \bar{\Lambda}') P(\Lambda' \otimes \bar{\Lambda}')\|_1 \\
    &\leq \|(\Lambda\otimes \bar{\Lambda} - \Lambda'\otimes \bar{\Lambda}') P (\Lambda \otimes \bar{\Lambda})\|_1 
    + \|(\Lambda'\otimes \bar{\Lambda}') P (\Lambda \otimes \bar{\Lambda} - \Lambda'\otimes \bar{\Lambda}')\|_1 \\
    &\leq \|\Lambda\otimes \bar{\Lambda} - \Lambda'\otimes \bar{\Lambda}'\|_F \|P(\Lambda \otimes \bar{\Lambda})\|_F
    + \|(\Lambda'\otimes \bar{\Lambda}')P\|_F \|\Lambda \otimes \bar{\Lambda} - \Lambda'\otimes \bar{\Lambda}'\|_F,
    \end{align*}
    where the first inequality follows from contractivity of trace norm under partial trace,
    the second is by triangle inequality, and the last inequality is by H\"older's inequality $\|XY\|_1\leq \|X\|_F\|Y\|_F$.
    Notice that 
    \[
    \|P(\Lambda \otimes \bar{\Lambda})\|_F =
    \sqrt{\Tr((\Lambda \otimes \bar{\Lambda}) P^2 (\Lambda \otimes \bar{\Lambda}))}
    \leq \sqrt{\Tr(\Lambda^2 \otimes \bar{\Lambda}^2)} = 1,
    \]
    and similarly $\|(\Lambda'\otimes \bar{\Lambda}')P\|_F \leq 1$. Thus
    \[
    \|\Lambda A\Lambda - \Lambda' A'\Lambda'\|_1 \leq 2 \|(\Lambda\otimes \bar{\Lambda}) - (\Lambda'\otimes \bar{\Lambda}')\|_F.
    \]
    Moreover,
    \[
    \|\Lambda \otimes \bar{\Lambda} - \Lambda' \otimes \bar{\Lambda}'\|_F
    \leq \|(\Lambda - \Lambda') \otimes \bar{\Lambda}\|_F + \|\Lambda' \otimes (\overline{\Lambda - \Lambda'})\|_F
    = \|\Delta\|_F \|\Lambda\|_F + \|\Lambda'\|_F \|\Delta\|_F
    = 2 \|\Delta\|_F. 
    \]
    Thus we have $\|\Lambda A\Lambda - \Lambda' A'\Lambda'\|_1 \leq 4 \|\Delta\|_F$.

    Finally, we bound $\|B-B'\|_\infty$. By definition, we have
    \[
    B-B' = \Tr_\Bs\left[(I\otimes \bar{\Lambda}^2)P - (I\otimes \bar{\Lambda}'^2)P\right] = \Tr_\Bs\left[(I\otimes (\bar{\Lambda}^2-\bar{\Lambda}'^2))P\right].
    \]
    To bound its $\infty$-norm, consider any unit vectors $x, y \in \mathbb{C}^N$. We have
    \begin{align*}
        |\braket{x, (B-B')y}| =
        |\Tr((B-B') y x^\dagger)| 
        &=
        |\Tr[\Tr_{\mathcal{B}}((yx^\dagger \otimes (\bar{\Lambda}^2-\bar{\Lambda}'^2)) P)]| \\
        &= 
        |\Tr[(yx^\dagger \otimes (\bar{\Lambda}^2-\bar{\Lambda}'^2)) P]| \leq
        \|yx^\dagger \otimes (\bar{\Lambda}^2-\bar{\Lambda}'^2)\|_1 
        =\|(\Lambda^2-\Lambda'^2)\|_1.
    \end{align*} 
    Notice that 
    \[
    \|\Lambda^2-\Lambda'^2\|_1 \leq
    \|(\Lambda - \Lambda')\Lambda\|_1 + \|\Lambda'(\Lambda - \Lambda')\|_1
    \leq \|\Delta\|_F \|\Lambda\|_F + \|\Lambda'\|_F \|\Delta\|_F
    = 2 \|\Delta\|_F.
    \]
    Thus we have $\|B-B'\|_\infty \leq 2 \|\Delta\|_F$.

    Plugging the above bounds into \eqref{eq:lal}, we have
    \(
    |\Es_{\Lambda}(P)-\Es_{\Lambda'}(P)| \leq 4 \|\Delta\|_F + 2 \|\Delta\|_F = 6 \|\Delta\|_F. 
    \)
\end{proof}

\section{An Alternative Proof for \texorpdfstring{$\leq 3$}{3 or Fewer} Colors}
\label{appx:3-coloring}

This appendix provides an operator-theoretic proof 
of \Cref{coro:3-coloring}, without 
relying on bounded-dependence arguments.

\begin{theorem} \label{thm:q23}
    Fix $q\in\{2,3\}$. For any POVM $\{P_i\}_{i\in[q]}$ on
    $\mathbb{C}^N\otimes \mathbb{C}^N$, and any PSD 
    $\Lambda\in\mathbb{C}^{N\times N}$ with $\|\Lambda\|_F=1$, we have
    \[
    \sum_{i=1}^q \Es_\Lambda(P_i) \geq C_q,
    \]
    where
    \[
    C_q =
    \begin{cases}
        1/8 & \text{if } q=2,\\
        1/2916 & \text{if } q=3.
    \end{cases}
    \]
\end{theorem}

\begin{proof}
    Since $\sum_{i=1}^q P_i = I$, we have
    $\sum_{i=1}^q m_\Lambda(P_i) = \Tr(\Lambda^2\otimes \bar{\Lambda}^2) = 1$.
    By averaging, there exists an index $j\in[q]$ such that
    $m:=m_\Lambda(P_j) \geq 1/q$. Consider the following two cases:
\begin{itemize}
        \item If $q=2$, then $m\geq 1/2$. By \Cref{lem:q2},
    \[
    \Es_\Lambda(P_j) \geq \frac{m(3m-1)}{2} \geq \frac{1}{8}.
    \]
    \item Suppose $q=3$. If $m\in[1/3,1/2]$, then \Cref{lem:q3} implies
    \[
    \Es_\Lambda(P_j)
    \geq \frac{m^4(4m-1)^2}{4}
    \geq \frac{(1/3)^4(4/3-1)^2}{4}
    = \frac{1}{2916}.
    \]
    If instead $m>1/2$, then \Cref{lem:q2} gives
    \[
    \Es_\Lambda(P_j) \geq \frac{m(3m-1)}{2} > \frac{1}{8} > \frac{1}{2916}.
    \]
    \end{itemize}    
    Therefore, in both cases,
    \(
    \sum_{i=1}^q \Es_\Lambda(P_i)
    \geq \Es_\Lambda(P_j)
    \geq C_q.
    \)
\end{proof}

The following two technical lemmas lower bound the energy of a single
element $P$ with $m_\Lambda(P)\geq 1/3$. The first lemma gives a positive lower
bound whenever $m_\Lambda(P)>1/3$. 

\begin{lemma}\label{lem:q2}
Given PSD $P\in \mathrm{L}(\mathbb{C}^N\otimes \mathbb{C}^N)$ with
$0\preceq P\preceq I$, and 
PSD $\Lambda\in\mathbb{C}^{N\times N}$ with $\|\Lambda\|_F=1$, 
if $m:=m_\Lambda(P)\geq 1/3$, then
\[
\Es_\Lambda(P)\geq \frac{m(3m-1)}{2}.
\]
\end{lemma}

\begin{proof}
Define matrices
\[
A:=\Tr_\As\left(P(\Lambda^2\otimes I)\right)^\top,
\qquad
B:=\Tr_\Bs\left(P(I\otimes \bar{\Lambda}^2)\right),
\qquad
\tilde P := (\Lambda\otimes \bar{\Lambda})P(\Lambda\otimes \bar{\Lambda}).
\]
Then we have
\(0\preceq \tilde P\preceq \Lambda^2\otimes \bar{\Lambda}^2\), 
\(\Lambda A\Lambda = \Tr_{\mathcal A}(\tilde P)^\top\),
and 
\(\Lambda B\Lambda = \Tr_{\mathcal B}(\tilde P)\).
Thus
\[
\Tr(\Lambda^2A)=\Tr(\Lambda^2B)=\Tr(\tilde P)=m.
\]
Define 
\[
\Es:=\Es_\Lambda(P)
\qquad
x:=\Tr(\Lambda A\Lambda A),
\qquad
y:=\Tr(\Lambda B\Lambda B).
\]
If we write
\[
R:=\Lambda^{1/2}A\Lambda^{1/2},
\qquad
S:=\Lambda^{1/2}B\Lambda^{1/2},
\]
then
\[
\Es=\Tr(RS),
\qquad
x=\Tr(R^2),
\qquad
y=\Tr(S^2),
\qquad
m=\Tr(\Lambda R)=\Tr(\Lambda S).
\]

Consider
\[
\epsilon:=\|R+S-2m\Lambda\|_F^2.
\]
Expanding the square and using $\Tr(\Lambda^2)=\|\Lambda\|_F^2=1$ gives
\begin{equation} \label{eq:epsilon}
\epsilon = \Tr(R^2)+\Tr(S^2)+2\Tr(RS)-4m\Tr(\Lambda R)-4m\Tr(\Lambda S)+4m^2\Tr(\Lambda^2)
= x+y+2\Es-4m^2.
\end{equation}
Since $\epsilon\ge 0$, we obtain
\begin{equation}\label{eq:lower-Es}
\Es\ge \frac{4m^2-(x+y)}{2}.
\end{equation}

It remains to bound $x+y$. Observe
\[
\Tr\bigl(\tilde P(I\otimes A^\top)\bigr)
= \Tr\bigl(\Tr_{\mathcal A}(\tilde P)A^\top\bigr)
= \Tr\bigl((\Lambda A\Lambda)^\top A^\top\bigr)=x,
\]
and similarly $\Tr\bigl(\tilde P(B\otimes I)\bigr)=y$. 
Next, compute 
\begin{align}
\notag\Tr\left[\tilde P\left((I-B)\otimes (I-A^\top)\right)\right] &=
\notag\Tr\left(\tilde{P}\right) - \Tr\left[\tilde P\left(I\otimes A^\top\right)\right] - \Tr\left[\tilde P\left(B\otimes I\right)\right] + \Tr\left[\tilde P\left(B\otimes A^\top\right)\right]  \\
\label{eq:PIAIB-form}&=m - x - y + \Tr\left[\tilde P\left(B\otimes A^\top\right)\right] 
\leq m - x - y + m^2,
\end{align}
where the last inequality is by $0\preceq \tilde P\preceq \Lambda^2\otimes \bar{\Lambda}^2$, and thus
\[
\Tr\bigl(\tilde P(B\otimes A^\top)\bigr)
\le \Tr\bigl((\Lambda^2\otimes \bar{\Lambda}^2)(B\otimes A^\top)\bigr)
= \Tr(\Lambda^2B)\overline{\Tr(\Lambda^2A)}=m^2.
\]
Since $(I-B)\otimes (I-A^\top)\succeq 0$ by 
$0\preceq A,B\preceq I$, and $\tilde P\succeq 0$, 
we have
\[
\Tr\left[\tilde P\left((I-B)\otimes (I-A^\top)\right)\right]\geq 0.
\]
Combining with \eqref{eq:PIAIB-form}, we have
\begin{equation}\label{eq:xy}
x+y\le m+m^2.
\end{equation}
Finally, by \eqref{eq:lower-Es} and \eqref{eq:xy}, we conclude that
\[
\Es\ge \frac{4m^2-(m+m^2)}{2}=\frac{m(3m-1)}{2}. \qedhere
\]
\end{proof}

To handle the case of $m=1/3$, we need a more refined analysis: we show that
when \eqref{eq:lower-Es} is almost tight, i.e. $\epsilon$ is close to zero, another inequality \eqref{eq:xy} must
have a large slack, which together give a positive lower bound when $m=1/3$.
Formally, we have the following lemma.

\begin{lemma} \label{lem:q3}
Given PSD $P\in \mathrm{L}(\mathbb{C}^N\otimes \mathbb{C}^N)$ with
$0\preceq P\preceq I$, and 
PSD $\Lambda\in\mathbb{C}^{N\times N}$ with $\|\Lambda\|_F=1$, 
if
    \(
    m:=\Tr\bigl(P(\Lambda^2\otimes \bar{\Lambda}^2)\bigr) \in [1/3,1/2],
    \)
    then
    \[
    \Es_{\Lambda}(P) \geq \frac{m^4(4m-1)^2}{4}.
    \]
\end{lemma}

\begin{proof}
    Adopt the notation introduced in the proof of \Cref{lem:q2}, and define
    \[
    \Delta := R + S - 2m \Lambda.
    \quad\Longrightarrow\quad
    R + S = 2m \Lambda + \Delta.
    \]
    Then by \eqref{eq:epsilon}, we have $\|\Delta\|_F^2 = \epsilon = x+y+2\Es-4m^2$.
    Set $a:=1-2m\geq0$. The definition of $\Delta$ gives the identities
    \begin{equation} \label{eq:l-r-l-s}
        \Lambda-S=a\Lambda+R-\Delta,
        \qquad
        \Lambda-R=a\Lambda+S-\Delta.
    \end{equation}

    Define another auxiliary matrix
    \(
    \hat{P} := (\Lambda^{1/2}\otimes \bar{\Lambda}^{1/2})
    P (\Lambda^{1/2}\otimes \bar{\Lambda}^{1/2}).
    \)
    Notice the relation
    \[
    \tilde{P} = 
    (\Lambda\otimes \bar{\Lambda}) P (\Lambda\otimes \bar{\Lambda}) 
    =(\Lambda^{1/2}\otimes \bar{\Lambda}^{1/2}) \hat{P}(\Lambda^{1/2}\otimes \bar{\Lambda}^{1/2}).
    \]
    We will estimate the quantity
    \[
        L:=\Tr\left[\hat{P}((\Lambda-S)\otimes(\Lambda-R)^\top)\right].
    \]

    First, we upper bound $L$. By cyclicity of trace,
    \begin{align}
    \notag \Tr\left[\hat{P} ((\Lambda - S) \otimes (\Lambda - R)^\top)\right]
    &= \Tr\left[
    \hat{P}
    \bigl((\Lambda - \Lambda^{1/2}B\Lambda^{1/2}) \otimes (\Lambda - \Lambda^{1/2}A\Lambda^{1/2})^\top\bigr)
    \right] \\
    \notag &= \Tr\left[
    (\Lambda^{1/2}\otimes\bar{\Lambda}^{1/2})
    \hat{P}
    (\Lambda^{1/2}\otimes\bar{\Lambda}^{1/2})
    \left((I - B) \otimes (I - A^\top)\right)
    \right] \\
    \label{eq:lhs}  &= \Tr\left[\tilde{P} ((I - B) \otimes (I - A^\top))\right] 
    \leq m - x - y + m^2,
    \end{align}
    where the last inequality is from \eqref{eq:PIAIB-form}.

    We next lower bound $L$ by applying identities in \eqref{eq:l-r-l-s} as
    \begin{align*}
    L={}&\Tr\left[\hat P((a\Lambda+R-\Delta)\otimes(a\Lambda+S-\Delta)^\top)\right]\\
    ={}&a^2\Tr\left[\hat P(\Lambda\otimes\bar\Lambda)\right]\\
    &+a\Tr\left[\hat P(\Lambda\otimes\bar S)\right]
      +a\Tr\left[\hat P(R\otimes\bar\Lambda)\right]
      +\Tr\left[\hat P(R\otimes\bar S)\right]\\
      &-a\Tr\left[\hat P(\Lambda\otimes\bar\Delta)\right]
      -a\Tr\left[\hat P(\Delta\otimes\bar\Lambda)\right]
      -\Tr\left[\hat P(R\otimes\bar\Delta)\right]
      -\Tr\left[\hat P(\Delta\otimes\bar S)\right]
      +\Tr\left[\hat P(\Delta\otimes\bar\Delta)\right].
    \end{align*}
    The first term $\Tr[\hat P(\Lambda\otimes\bar\Lambda)]=\Tr(\tilde P)=m$.
    The next three terms in the second line are nonnegative 
    as $a\geq0$ and $\Lambda,R,S,\hat P\succeq0$. It remains to bound
    the absolute values of error terms in the third line. 

    We write $D:=\Lambda^{1/2}\Delta\Lambda^{1/2}$. 
    As $\frac14+\frac12+\frac14=1$, the generalized H\"older's
    inequality gives
    \begin{equation*} 
    \|D\|_1
    \leq
    \|\Lambda^{1/2}\|_4 \, \|\Delta\|_2 \, \|\Lambda^{1/2}\|_4
    = \|\Delta\|_2 = \|\Delta\|_F = \sqrt{\epsilon}.
    \end{equation*}
    For every PSD $X$, by H\"older's inequality and $0\preceq P\preceq I$, 
    \begin{align*}
    \left|\Tr\left[\hat P(X\otimes\bar\Delta)\right]\right|
    &=\left|\Tr\left[P\bigl((\Lambda^{1/2}X\Lambda^{1/2})
      \otimes\bar D\bigr)\right]\right|
      \leq \|P\|_\infty \|\Lambda^{1/2}X\Lambda^{1/2}\otimes D\|_1
      \leq \Tr(\Lambda X) \sqrt{\epsilon},
    \end{align*}
    and similarly $\left|\Tr\left[\hat P(\Delta\otimes\bar X)\right]\right|
    \leq \Tr(\Lambda X) \sqrt{\epsilon}$.
    Since $\Tr(\Lambda^2)=1$ and
    $\Tr(\Lambda R)=\Tr(\Lambda S)=m$, the first four error terms are
    bounded in absolute value by $2(a+m)\sqrt\epsilon=2(1-m)\sqrt\epsilon$.
    Finally, the last error term
    \[
        \left|\Tr\left[\hat P(\Delta\otimes\bar\Delta)\right]\right|=
        \left|\Tr\left[P( D\otimes\bar D)\right]\right|
        \leq \|P\|_\infty \|D\otimes\bar D\|_1
        \leq\|D\|_1^2\leq\epsilon.
    \]
    Consequently,
    \begin{equation} \label{eq:rhs}
        L\geq (1-2m)^2m-2(1-m)\sqrt\epsilon-\epsilon.
    \end{equation}

    Combining \eqref{eq:lhs} and \eqref{eq:rhs}, we have
    \[
        m-x-y+m^2\geq(1-2m)^2m-2(1-m)\sqrt\epsilon-\epsilon.
    \]
    Substituting $\epsilon=x+y+2\Es-4m^2$ and simplifying gives
    \begin{equation} \label{eq:es-bound}
        2\Es\geq4m^3-m^2-2(1-m)\sqrt\epsilon.
    \end{equation}
    Moreover, by \eqref{eq:xy} and the assumption $m\geq 1/3$,
    \[
    \epsilon = x + y + 2\Es - 4m^2
    \leq m + m^2 + 2\Es - 4m^2
    = m(1-3m) + 2\Es
    \leq 2\Es.
    \]
    Plugging this into \eqref{eq:es-bound}, we obtain
    \begin{equation} \label{eq:q3-bound}
    2\Es \geq 4m^3 - m^2 - 2(1-m)\sqrt{2\Es}.
    \end{equation}

    Finally, we lower bound $\Es$ using \eqref{eq:q3-bound}.
    Let 
    $u:=\sqrt{2\Es}$, $c:=1-m$, and $b:=4m^3-m^2=m^2(4m-1)$.
    Then \eqref{eq:q3-bound} can be written as
    \(
    u^2 + 2c u - b \geq 0.
    \)
    Since $u\geq 0$, it follows that
    \[
    u \geq -c + \sqrt{c^2 + b}
    = \frac{b}{c + \sqrt{c^2 + b}}
    \geq \frac{b}{2\sqrt{c^2+b}}.
    \]
    Therefore,
    \[
        \Es=\frac{u^2}{2}\geq\frac{b^2}{8(c^2+b)}.
    \]
    For $m\in[1/3,1/2]$,
    \[
        c^2+b=1-2m+4m^3
        =\frac12-\frac{(1-2m)(4m^2+2m-1)}{2}
        \leq\frac12.
    \]
    Here the final inequality follows because both factors in the subtracted
    term are nonnegative on $[1/3,1/2]$.
    Consequently,
    \[
        \Es\geq\frac{b^2}{4}
        =\frac{m^4(4m-1)^2}{4}. \qedhere
    \]
\end{proof}

\end{document}